\documentclass{article}

\usepackage{microtype}
\usepackage{graphicx}
\usepackage{subcaption}
\usepackage{booktabs} %

\usepackage{hyperref}

\usepackage{rotating}
\usepackage{array}
\usepackage{amsthm}
\usepackage{amsmath}
\usepackage{float}
\usepackage{tabularx}
\usepackage{amssymb}
\usepackage[framemethod=TikZ]{mdframed}
\mdfsetup{skipabove=6pt, skipbelow=0pt}
\usepackage{color}
\usepackage{caption}

\newcommand{\rev}[1]{{{#1}}}

\usepackage[accepted]{icml2024}

\usepackage{amsmath}
\usepackage{amssymb}
\usepackage{mathtools}
\usepackage{amsthm}

\usepackage[capitalize,noabbrev]{cleveref}

\theoremstyle{plain}
\newtheorem{thm}{Theorem}[section]
\newtheorem{prop}[thm]{Proposition}
\newtheorem{lem}[thm]{Lemma}
\newtheorem{cor}[thm]{Corollary}
\theoremstyle{definition}

\newtheorem{assum}[thm]{Assumption}
\newtheorem{problem}{Problem}
\theoremstyle{remark}
\newtheorem{remark}[thm]{Remark}
\theoremstyle{plain}
\newtheorem{exmp}[thm]{Example}
\theoremstyle{definition}
\newtheorem{assumptionalt}{Assumption}

\newenvironment{assumptionp}[1]{
  \renewcommand\theassumptionalt{#1}
  \assumptionalt
}{\endassumptionalt}

\newcommand{\paren}[1]{\ensuremath{\left( #1\right)}}

\newcommand{\set}[1]{\ensuremath{\left\{ #1\right\}}}
\newcommand{\matr}[1]{\ensuremath{\begin{bmatrix} #1 \end{bmatrix}}}

\newcommand{\snorm}[1]{\ensuremath{\| #1\|}}

\newcommand{\tr}{\mathrm{tr}}

\usepackage[textsize=tiny]{todonotes}

\icmltitlerunning{Predictive Linear Online Tracking for Unknown Targets}

\begin{document}

\twocolumn[
\icmltitle{Predictive Linear Online Tracking for Unknown Targets}

\icmlsetsymbol{equal}{*}

\begin{icmlauthorlist}
\icmlauthor{Anastasios Tsiamis}{equal,ifa}
\icmlauthor{Aren Karapetyan}{equal,ifa}
\icmlauthor{Yueshan Li}{idsc}
\icmlauthor{Efe C. Balta}{ifa,ins}
\icmlauthor{John Lygeros}{ifa}
\end{icmlauthorlist}

\icmlaffiliation{ifa}{Automatic Control Laboratory, ETH Zürich, Zürich, Switzerland}
\icmlaffiliation{idsc}{Institute for Dynamic Systems and Control, ETH Zürich, Zürich, Switzerland}
\icmlaffiliation{ins}{Control and Automation Group, inspire AG, Zürich, Switzerland}

\icmlcorrespondingauthor{Anastasios Tsiamis}{atsiamis@ethz.ch}

\icmlkeywords{Online Learning, Tracking, Control, Recursive Learning}

\vskip 0.3in
]

\printAffiliationsAndNotice{\icmlEqualContribution} %

\begin{abstract}
In this paper, we study the problem of online tracking in linear control systems, where the objective is to follow a moving target. Unlike classical tracking control, the target is unknown, non-stationary, and its state is revealed sequentially, thus, fitting the framework of online non-stochastic control. 
We consider the case of quadratic costs and propose a new algorithm, called predictive linear online tracking (PLOT). The algorithm uses recursive least squares with exponential forgetting to learn a time-varying dynamic model of the target. The learned model is used in the optimal policy under the framework of receding horizon control. We show the dynamic regret of PLOT scales with $\mathcal{O}(\sqrt{TV_T})$, where $V_T$ is the total variation of the target dynamics and $T$ is the time horizon. 
Unlike prior work, our theoretical results hold for non-stationary targets. We implement PLOT on a real quadrotor and provide open-source software, thus, showcasing one of the first successful applications of online control methods on real hardware. 
\end{abstract}

\section{Introduction}
Target tracking is a fundamental control task for autonomous agents, allowing them to be used in a variety of applications including environmental monitoring~\cite{environmental}, agriculture~\cite{agriculture}, and air shows~\cite{dance} to name a few. 
Typically, an autonomous agent is given a reference trajectory, which should be tracked with as little error as possible. In the case of linear systems with quadratic costs the problem, also known as Linear Quadratic Tracking (LQT), can be formulated as follows.
Let the autonomous agent be governed by the linear dynamics
\begin{equation}\label{eq:system}
    x_{t+1} = Ax_t+Bu_t,
\end{equation}
where $x_t\in\mathbb{R}^n$ is the state, and $u_t\in\mathbb{R}^m$ is the control input, while $A\in\mathbb{R}^{n\times n}$ and $B\in\mathbb{R}^{n\times m}$ \rev{denote} the \emph{known} system dynamics. 
Given a target trajectory $r_t\in\mathbb{R}^{n\times 1}$, $t\ge 0$ the goal is finding the optimal control input of the \rev{finite-horizon} problem
\begin{equation}\label{eq:LQT}
    \min_{u_{0:T-1}}  \!\sum_{t=0}^{T-1} \left(\lVert x_t-r_t\rVert_Q^2 + \lVert u_t\rVert_R^2\right)\! +\! \lVert x_T-r_T\rVert_{Q_T}^2, 
\end{equation}
subject to the dynamics~\eqref{eq:system}, for some \rev{total time (horizon) $T$}.  The state, input, and terminal penalties, $Q\in\mathbb{R}^{n\times n},R\in\mathbb{R}^{m\times m},Q_T\in\mathbb{R}^{n\times n}$, respectively, are design choices.

In classical target tracking, the target state $r_t$ is known a priori and is precomputed.
However, in many cases, the target might be time-varying and \emph{unknown}. Such scenarios arise, for example, in the case of adversarial target tracking, wild animal tracking, moving obstacle avoidance, pedestrian tracking, etc. In such cases, the target trajectory is generated by an external unknown dynamical system, also known as ``exosystem"~\cite{nikiforov2022adaptive}. 
In this work, we assume that the target is generated by time-varying, auto-regressive (AR) dynamics. Let
\[
z_t=\begin{bmatrix}r^\top_{t}&\cdots&r^\top_{t-p+1}\end{bmatrix}^\top
\]
contain past target states, with $p$ the memory or past horizon of the dynamics. Then, the target at the next time step is given by the AR model
\begin{equation}\label{eq:exosystem}
\begin{aligned}
 r_{t+1} 
    &=S_{t+1} z_t,
    \end{aligned}
\end{equation}
where $S_t\in\mathbb{R}^{n\times np}$ is an \emph{unkonwn} matrix.

Adaptive control has a long history of dealing with dynamical uncertainty~\cite{annaswamy2021historical}, where the goal is to simultaneously control the system and adapt to the uncertainty. One of the most widely-used algorithms in the adaptive control literature with ubiquitous applications has been the celebrated Recursive Least Squares (RLS) algorithm with forgetting factors~\cite{aastrom1977theory}. While adaptive control has been extensively studied in the stochastic or time-invariant regime~\cite{guo1995performance,ljung1990adaptation}, results for non-stochastic, time-varying systems are relatively scarce. By studying the problem through the modern lens of online learning and non-stochastic control~\cite{hazan2022introduction}, we can provide guarantees for new classes of time-varying systems.

Most works on adaptive control follow one of the two possible paradigms: i) \emph{direct} control and ii) \emph{indirect} control, mirroring model-free and model-based Reinforcement Learning. In the former case, adaptation occurs directly at the control policy; in the latter case, a model of the uncertainty is kept at all times and the control policy is updated indirectly based on the current model. 
A notable benefit of indirect architectures is that they decouple control from learning, making it easier to incorporate changes in the control objective or robustness specifications.
In the case of non-stochastic control, learning the optimal feedforward law~\cite{agarwal2019online,foster2020logarithmic} fits the direct paradigm. Prior works that use predictions~\cite{li2019online,zhang2021regret} are closer to the indirect paradigm but they either do not provide a prediction method or the regret guarantees are suboptimal in the case of targets with dynamic structure.

\subsection{Contribution}

In this work, we propose a tracking control algorithm, called \rev{Predictive Linear Online Tracking} (PLOT).  
It adapts to unknown targets online using sequentially measured target data, based on the indirect paradigm.  
PLOT uses a modified version of the RLS algorithm with forgetting factors~\cite{yuan2020trading} to learn the target dynamics and predict future target states. Then, it computes a receding horizon control input using the predicted target states in a certainty equivalent fashion. To characterize closed-loop performance, we use the notion of dynamic regret~\cite{zinkevich2003online}. 
In our setting, dynamic regret compares the incurred closed-loop cost of any online algorithm against the cost incurred by the optimal control actions in hindsight, that is, the actions generated by the optimal non-causal policy that has full knowledge about future target states. 
Our contributions are the following.

\paragraph{Dynamic regret for online tracking.} 
We prove that the dynamic regret of PLOT is upper-bounded by $\mathcal{O}(\sqrt{TV_T})$, where $V_T$ is the total variation of the target dynamics $S_t$ (path length). When the target dynamics are static (zero variation), the regret becomes logarithmic and we recover prior results~\cite{foster2020logarithmic}. While prior dynamic regret bounds exist for linear quadratic control, they either focus on direct online control~\cite{baby2022optimal}, where we learn over a directly parameterized disturbance affine feedback policy, or they assume no structure for the targets~\cite{li2019online,karapetyan2023online}. Instead, we focus on indirect online control, where we learn an internal model representation for the unknown target. By setting up multiple step ahead predictors, we have more flexible and frequent feedback, thus avoiding extra logarithmic factors due to delayed learning feedback as in~\cite{baby2022optimal,foster2020logarithmic}. 
Using the total variation of the target \emph{dynamics} $S_t$ to characterize dynamic regret is a new point of view, improving prior bounds that use the variation of the target state $r_t$ itself~\cite{li2019online}.

\paragraph{Prediction of time-varying partially observed systems.} We employ the RLS algorithm to predict target states multiple time steps into the future.  Our result is of independent interest since it applies to time-varying AR systems or systems with exogenous inputs (ARX), which is a common model in system identification. We obtain dynamic regret guarantees for prediction by adapting the analysis of~\citet{yuan2020trading} to the setting of learning with delayed feedback~\cite{joulani2013online}. This result can be seen as a generalization of time-invariant Kalman filtering to the non-stochastic, time-varying, multi-step ahead case. 

\paragraph{Experimental demonstration.} 
While the theory of online non-stochastic control has matured over the past years~\cite{hazan2022introduction}, its practice is still lagging, with only \rev{a few} notable experimental demonstrations on hardware~\cite{suo2021machine,snyder2023online}. In this work, we perform extensive \rev{simulations showing how the regret analysis can be used as a tool to tune the controller hyperparameters, provide comparisons of several online control methods in simulations, and implement PLOT on a hardware setup}. To facilitate benchmarking and future developments, we make our code open-source and implement it on Crazyflie drones with a reconfigurable software architecture built on the Robot Operating System (ROS). To the best of our knowledge, our work is one of the first to demonstrate the use of online non-stochastic control \rev{with guarantees on a real quadrotor}.

\subsection{Organization and Notation} The rest of the paper is organized as follows. Section~\ref{sec:problem statement} states the assumptions 
and defines the control objective. The proposed algorithm, PLOT, is presented in Section~\ref{sec:algorithm}, and its dynamic regret is analyzed in Section~\ref{sec:analysis}. The experiment results in simulation and on hardware are presented in Section~\ref{sec:case_study}. Section~\ref{sec:conclusion} provides concluding remarks and future directions. Additional background material, detailed proofs, and further implementation and experimental details can be found in the Appendix. 

For a matrix $M$, $\| M \|$ denotes the $\ell_2$ induced  operator norm, while $\snorm{M}_F$ denotes the Frobenius norm. For positive definite $P\succ 0$, we define the weighted Frobenius norm as $\lVert M \rVert_{F,P} = \sqrt{\tr({MPM^\top})}$, where $\tr(\cdot)$ denotes the trace. For a vector $x\in\mathbb{R}^n$, $\|x\|$ denotes the Euclidean norm and $\|x\|_P = \sqrt{x^T P x}$ denotes the weighted $P$ norm for $P\succ 0$.  We use $x_{1:t}:=\{x_1, ..., x_{t}\}$ as a sequence of variables spanning from time step 1 to t. Given $n \in \mathbb{N}^+$, the identity matrix of dimension $n$ is defined by $\mathbf{I}_n$.

\section{Problem Statement}
\label{sec:problem statement}

Consider the LQT problem~\eqref{eq:LQT}. Let the target state $r_t$ be sequentially revealed. The following events happen in sequence at every time step $t$. i) The current state of the system $x_t$ and the target $r_t$ are received and the current tracking cost is incurred; ii) The controller applies an action $u_t$; iii) The system and the target evolve according to~\eqref{eq:system},~\eqref{eq:exosystem} respectively.
We assume the initial value $z_{-1}$ to be known.

Many fundamental targets can be captured by dynamics~\eqref{eq:exosystem}. For example, a target moving on a circle or a straight line with a constant speed can be captured by linear time-invariant dynamics--see Appendix~\ref{app_sec:trajectory_examples}. By considering time-varying dynamics, we can model more complex target trajectories, e.g., combinations of lines, circles, switching patterns, etc. \rev{Our results can be generalized directly to targets driven by exogenous variables (ARX) $\textstyle{r_{t+1}=S_{t+1}z_t+v_{t+1}}$--it is sufficient to redefine $\tilde{z}_t=\matr{z^\top_t&1}^\top$ and $\tilde{S}_{t}\triangleq \matr{S_{t} &v_t}$. We can thus capture even richer targets, including other control systems.
In fact, \emph{any} arbitrary bounded adversarial target can be captured by such a model; it is sufficient to take $S_{t+1}=0$ and $v_{t+1} = r_{t+1}$. We note that the above representation is not unique--see also Appendix~\ref{app_sec:trajectory_examples}. } 

We remark that the LQT problem is equivalent to the Linear Quadratic Regulator (LQR) with disturbances--see Appendix~\ref{app_sec:LQT}. Hence, the techniques applied here could also be applied to noisy system dynamics~\eqref{eq:system}, where we attach a dynamic structure to the disturbances.

By applying the one step ahead model~\eqref{eq:exosystem} multiple times, we obtain $k$-step ahead recursions of the form
\begin{equation}\label{eq:exosystem_multi_step}
r_{t+k}=S_{t+k|t}z_t,
\end{equation}
where $k\in\mathbb{N}$ is the number of future steps. The $k-$step ahead matrix $S_{t+k|t}$ is a nonlinear function (multinomial) of $S_{t+k},\dots,S_{t+1}$. 
Expressions for the multi-step ahead recursions~\eqref{eq:exosystem_multi_step} can be found in Appendix~\ref{app_sec:multi_step_dynamics}. By definition $S_{t+1|t}\equiv S_{t+1}$.
The notation $S_{t+k|t}$ indicates that given all information $z_t$ at time $t$, the product $S_{t+k|t}z_t$ acts as a $k-$step ahead prediction of $r_{t+k}$ at time $t$. 

To ensure that the online tracking problem is well-defined, we consider the following boundedness assumptions on the target dynamics and trajectory.  
\begin{assum}[Bounded Signals]
\label{assumption:boundedness}
  The target state is bounded. For some $D_r\ge 0$, we have $\snorm{r_t}\leq D_r$
  , for $t=1,...,T$.
\end{assum}
\begin{assum}[Stability]
\label{assumption:stability} Let $\mathcal{S}\triangleq\{S\in\mathbb{R}^{n\times pn}\snorm{S}\le M\}$, for some $M\ge 0$.
    The $k-$step ahead dynamics are uniformly bounded $\snorm{ S_{t+k|t}} \in\mathcal{S}$, for all $t=0,\dots,T,\,k\le T-t$. 
\end{assum}
Such boundedness conditions reflect typical assumptions in online learning~\cite{yuan2020trading}. The second condition allows general $S_k$ as long as the multi-step ahead matrices do not blow up.

\subsection{Control Objective}
Our goal is to design an online controller that adapts to the unknown target dynamics. To evaluate \rev{the} online performance, we use dynamic regret, which compares the incurred cost of the online controller to the \emph{optimal} controller in hindsight that knows all future target states in advance (called simply the \textit{optimal controller}). 
For a target state realization $\textbf{r}:=r_{0:T}$, we denote the cumulative cost by
\begin{equation*}
    J_T(u_{0:T-1};\textbf{r})\!\triangleq\!\sum_{t=0}^{T-1}\!\left(\lVert x_t\!-\!r_t\rVert^2_Q\!+\!\lVert u_t\rVert^2_R\right) + \lVert x_T-r_T\rVert_{Q_T}^2.
\end{equation*}
The optimal controller minimizes the cumulative cost, given knowledge of the whole target realization 
\begin{align*}
    u^*_{0:T-1} = \mathrm{argmin}_{u_{0:T-1}}J_T(u_{0:T-1};\textbf{r}).
\end{align*}
On the contrary, for an online policy $\pi$, the input $u_t$ is a function of the current state $x_t$ and the target states only up to time step $t$, so that  
\begin{align*}
    u^{\pi}_t= \pi_t(x_t;r_{1:t}),\quad \forall t=0,...,T-1,
\end{align*} 
where the notation $u^{\pi}$ denotes the control input under the policy $\pi$. In other words, the optimal controller $u^*_t$ is \emph{non-causal} while the online controller $u^{\pi}_t$ is \emph{causal}.

The dynamic regret is given by the cumulative difference between the cost achieved by the online causal policy and the cost achieved by the non-causal optimal controller
\begin{equation}
   \mathcal{R}(\pi) = J_T(u_{1:T-1}^{\pi};\textbf{r}) - J_T(u_{1:T-1}^*;\textbf{r}),
\end{equation}
\rev{defined for a target state realization $\textbf{r}$}. We can now summarize the main objective of the paper.

\begin{mdframed}[roundcorner=3pt, backgroundcolor=blue!10]
\begin{problem}[Dynamic Regret]
Design an online controller $\pi$ that adapts to the unknown target states $r_t$, and characterize the dynamic regret $\mathcal{R}(\pi)$. 
\end{problem}
\end{mdframed}

To make sure that the control problem is well-defined, we introduce the following assumption which is standard in the control literature~\cite{zhang2021regret}. It guarantees that the LQT controller will be (internally) stable. 
\begin{assum}[Well-posed LQT]\label{ass:LQT_stability}
    The pair $(A,B)$ is stabilizable and 
    $Q$, $R$ are symmetric positive definite.
\end{assum}
We assume the terminal cost matrix $Q_T$ is chosen as the solution $X$ of the Discrete Algebraic Riccati Equation.

\begin{assum}[Terminal Cost]\label{ass:terminal}
    We select $Q_T=X$, where $X$ is the unique solution to Riccati equation
    \begin{equation}
    \label{eq:P*}
    X =Q+A^\top XA-A^\top XB(R+B^\top X B)^{-1}B^\top X A.
\end{equation}
\end{assum}
Assumption~\ref{ass:terminal} is only introduced to streamline the presentation and it is not restrictive. If we select a different terminal cost, then the effect on the optimal total cost~\eqref{eq:LQT} will be negligible (will only differ by a constant). 

\subsection{LQT Optimal Controller}
Having full access to future target states, the optimal controller admits the following analytical solution \cite{foster2020logarithmic, goel2022power}
\begin{equation}
\label{eq:optimal action}
    u_t^*(x_t) = -K(x_t-r_t) - \underbrace{\sum_{i=t}^{T-1}K_{i-t}(Ar_{i}-r_{i+1})}_{q_t(r_{t:T})},
\end{equation}
where the feedback and feedforward matrices 
are given by
\begin{align}
\label{eq:K*}
    K&=(R+B^\top XB)^{-1}B^\top XA, \\
\label{eq:Kd}
    K_{t}&=(R+B^\top X B)^{-1}B^\top(A-BK)^{\top,t}X,\; t\ge 0,
\end{align}
with $X$ defined in~\eqref{eq:P*}.
Since the dynamics~\eqref{eq:system} are known, all gain matrices $K,K_t,\,t\ge 0$ are known. The feedback term $-K(x_t-r_t)$ can be computed since the error $x_t-r_t$ is measured before choosing the action. The only unknown is the feedforward term $q_t(r_{t:T})$.  Note that the \emph{functional form} of the feedforward term is known. The missing piece of information is the future target states.
Following the paradigm of indirect control, we will use this structural observation to decompose the problem of online tracking into \rev{one of} online prediction and control design.

A notable property of LQT is that, under Assumption~\ref{ass:LQT_stability},
matrix $(A-BK)$ has all eigenvalues inside the unit circle. 
Therefore, there exist $c_0>0$ and $\rho\in (0,1)$ such that
\begin{equation}\label{eq:LQT_gains_exponenitally_decaying}
    \lVert K_{k}\rVert \leq c_0\rho^{k} ,\text{ for all }k\ge 0.
\end{equation}
Upper bounds for $c_0,\rho$ as well as a relaxed version of Assumption~\ref{ass:LQT_stability} can be found in Appendix~\ref{app_sec:LQT}. Interestingly, the above property implies that future target states get discounted when considering the current action; we only need to know accurately the imminent target states.

\section{Predictive Linear Online Tracking}
\label{sec:algorithm}
Our proposed online tracking algorithm can be decoupled into two steps at every time step $t$: i) predicting the future target states for up to $W-$steps into the future, for some prediction horizon $W$; ii) computing the current online control action based on the certainty equivalence principle and the receding horizon control framework.

\paragraph{Target Prediction.}
For the prediction step, we employ the RLS algorithm with forgetting factors. 
At every time step $t$, the RLS algorithm provides predictions $r_{t+1|t},\dots,r_{t+W|t}$ of the future target states $r_{t+1},\dots,r_{t+W}$, for some horizon $W>0$. \rev{For every step ahead prediction, we keep $W$ separate predictors $\hat S_{t+1|t}$,..., $\hat S_{t+W|t}$ without exploiting the shared structure. This parameterization leads to a higher-dimensional but convex problem. In contrast, learning over matrices $S_{t+W},\dots,S_{t+1}$  leads to non-convex optimization problems since the multi-step ahead predictors are products of entries in $S_{t+W},\dots,S_{t+1}$--see also Remark~\ref{rem:improper_learning}.}

When predicting the target state $r_{t+k|t}$ for $k$ steps ahead, the true value $r_{t+k}$ is revealed after $k$ steps 
into the future. To deal with this \emph{delayed feedback} issue we follow the approach of~\citet{joulani2013online}. 
\rev{Within each $k-$step ahead predictor, we maintain $k$ \rev{copies} $\{\hat S_{t+k-j|t-j}\}^{k-1}_{j=0}$\rev--we refer to them as learners--which are updated independently at non-overlapping time steps.}
Hence, the estimate $\hat{S}_{t+k|t}$ is updated based only on pairs of $(r_{t-ik},z_{t-(i+1)k})$, for $i=0,1,2,\dots$.
At each time step $t$, only one of the $k$ independent \rev{learners} is invoked, which aims to minimize the following 
prediction error $\forall t>1$:
\begin{equation}\label{eq:RLS_batch}
\begin{aligned}
    \min\limits_{S\in\mathcal{S}}&\sum_{i=0}^{\lfloor (t+1)/k \rfloor-1}\gamma^{i} f_{t-ik,k}(S), \\
    &f_{\tau,k}(S)\triangleq  \lVert r_{\tau} - Sz_{\tau-k}\rVert^2
    \end{aligned}
\end{equation}
where $\gamma\in(0, 1]$ is the forgetting factor that attributes higher weights to more recent data points. Recall that $\mathcal{S}$ is the set in which the dynamics lie--see Assumption~\ref{assumption:stability}. Note that each \rev{independent learner is updated based on} at most $\lceil T/k \rceil$ data points.  
Instead of solving~\eqref{eq:RLS_batch}, we opt for a recursive implementation with projections adapted from~\citet{yuan2020trading} and can be found in Algorithm~\ref{alg:RLS}.

\begin{algorithm}[hb!]
\caption{RLS for $k$-step-ahead prediction}
\label{alg:RLS}
\begin{algorithmic}[1]
\REQUIRE{Forgetting factor $\gamma \in (0, 1)$, Regularizer \rev{$\varepsilon>0$.}}
 \STATE Initialize $k$ learners $\hat{S}_{j+k|j}\in\mathcal{S}$, $j=-1,...,k-2$.
 \STATE Initialize $P_{j|j-k} = \varepsilon \mathbf{I}$ with $j=-1,...,k-2$.
 \FOR{$t = k-1,\dots, T$}
  \STATE Predict $r_{t|t-k}=\hat{S}_{t|t-k}z_{t-k}$.
  \STATE Receive true target state $r_t$; Incur loss $f_{t,k}(\hat{S}_{t|t-k})$
  \STATE Update $P_{t|t-k} = \gamma P_{t-k|t-2k} + z_{t-k}z_{t-k}^T$;\label{code:RLSupdateP}
    \STATE $\hat{S}^*_{t+k|t} = \hat{S}_{t|t-k} + (r_{t} - \hat{S}_{t|t-k} z_{t-k})z_{t-k}^\top P_{t|t-k}^{-1}$. \label{code:RLSupdate} 
    \STATE Project $\hat{S}_{t+k|t} = \Pi^{P_{t|t-k}}_{\mathcal{S}}(\hat{S}^*_{t+k|t})$,\\
    where $\Pi^{P_{t|t-k}}_{\mathcal{S}}(Y)\triangleq \arg\min_{S\in\mathcal{S}} \snorm{S-Y}_{F,P_{t|t-k}}$.
 \ENDFOR
\end{algorithmic}
\end{algorithm}

Since we need $k$ \rev{learners} for every $k-$step ahead predictor, we have a total number of $W(W+1)/2$ \rev{learners}. \rev{However, only $W$ of them are actively updated at every time step (one active learner per predictor). }
An visualization of the learning architecture for $k=1,2,3$ can be found in \cref{fig:schedule} in \cref{app_sec:multi_step_pred_proof}.

Different from direct approaches~\cite{foster2020logarithmic,baby2022optimal}, where the learning feedback is delayed by $W$ steps, here the learning feedback delay adapts to the $k-$step ahead horizon offering more flexibility. For example, the feedback for the one-step ahead predictor is always available without waiting $W$ steps. A downside is that we update more predictors.
\begin{remark}[Improper Learning versus Single Learner]\label{rem:improper_learning}
 \rev{  To obtain the multi-step ahead predictions $\hat r_{t+k|t}$, $k=1,\dots,W$, we keep $W$ \emph{separate} predictors $\hat S_{t+1|t}$,..., $\hat S_{t+W|t}$ without exploiting the shared structure, i.e., we follow improper learning. By following this choice, we convexify the problem, potentially at the expense of higher sample complexity. An alternative, ``naive" approach is to estimate the one-step ahead predictor $\hat{S}_{t+1|t}$, using a single learner. Then, to predict $k$ steps ahead, e.g. $\hat r_{t+k|t}$, we could propagate the target states through $\hat S_{t+1}$, for $k$ time steps. However, propagating through multiple time steps can lead to instability. For example, if $z_t=r_t$ (namely the AR memory length $p=1$), we will have $\hat r_{t+k|t}=\hat S^k_{t+1}r_t$. Due to the exponent $k$, this could lead to a rapid accumulation of errors. Enforcing the boundedness constraint of Assumption 2.2 directly on $\hat S^k_{t+1}$ would not work either since it leads to a non-convex problem. A detailed comparison of PLOT and this ``naive" approach in simulation can be found in Appendix~\ref{app_sec:naive_RLS}}

\end{remark}

\paragraph{Receding Horizon Control.}
In the receding horizon control framework, instead of looking over the \rev{total time} $T$ as in~\eqref{eq:LQT}, we only look over a window of length $W$ and solve the following optimal control problem
\begin{equation}
	\begin{aligned}
     \min_{u_{t:t+W-1}}& \sum_{k=0}^{W-1}\left(\lVert x_{t+k}-r_{t+k|t}\rVert_Q^2 + \lVert u_{t+k}\rVert_R^2\right) \\
     &\qquad+ \lVert x_{t+W}-r_{t+W|t}\rVert_P^2 \\
    \text{s.t.} \quad & x_{t+k+1} = Ax_{t+k} + Bu_{t+k},\; k\le W-1,
	\end{aligned},
\label{eq:mpc formulation}
\end{equation}
where $r_{t+k|t}$ are the predictions of the RLS algorithm. We replace the true target states of~\eqref{eq:LQT} with their prediction in~\eqref{eq:mpc formulation} according to the certainty equivalence principle. The receding horizon control policy that minimizes~\eqref{eq:mpc formulation} can be computed in closed form as
\begin{equation}
\label{eq:mpc action}
    u_t^{\pi}(x_t) \!=\! -K (x_t-r_t)- \!\sum_{i=t}^{t+W-1}\!K_{i-t}(Ar_{i|t}-r_{i+1|t}).
\end{equation}
The functional form of the receding horizon control is similar to the one of the optimal controller in~\eqref{eq:optimal action}. The main difference is that the feedforward term is applied to the predicted disturbances, instead of the actual ones.

\begin{algorithm}[hbt!]
\caption{PLOT: Predictive Linear Online Tracking }\label{alg:mpc}
\begin{algorithmic}[1]
\REQUIRE{Horizon $W$, forgetting factor $\gamma \in (0, 1)$}
\STATE Compute $X, K, \{K_0, ..., K_{W-1}\}$ as in \eqref{eq:P*}, \eqref{eq:K*}, \eqref{eq:Kd}.
\STATE Initialize RLS \rev{predictors} according to the Algorithm \ref{alg:RLS}, for $k=1,...,W$, respectively.\label{code:MPCinit}
\FOR{t = 0, ..., $T-1$}
  \STATE Observe system state $x_t$, target state $r_t$. \label{code:MPCmeasurement}
  \FOR{$k = 1, ..., W$}
    \IF{$t+k \leq T$}
        \STATE Update the $k$-step-ahead \rev{predictor} and predict $r_{t+k|t}$ according to the Algorithm~\ref{alg:RLS}.
    \ELSE
        \STATE Set $r_{t+k-1|t} = 0$.\label{code:MPCcorner}
    \ENDIF
  \ENDFOR
  \STATE Compute $u_t^{\pi}$ as in~\eqref{eq:mpc action}
 \ENDFOR
\end{algorithmic}
\end{algorithm}

The main part of our proposed algorithm, Predictive Linear Online Tracking (PLOT), is described in Algorithm~\ref{alg:mpc}. We initialize $W$ \rev{RLS predictors} (Algorithm~\ref{alg:RLS}) to estimate the $k$-step-ahead dynamics $S_{t+k|t}$ for $k=1,...,W$ independently. 
At each time step, the algorithm receives the latest target and state measurements, updates the RLS predictors, and makes new predictions. Then, it computes the receding horizon control input given the predictions and the state information. The whole procedure is iterated over all time steps up to the \rev{total time} $T$.
Note that if the prediction time $k+t$ exceeds the  \rev{total time} $T$, then we set the predicted target to zero; as seen by~\eqref{eq:optimal action} only the target states up to time $T$ affect the control problem.

\section{Dynamic Regret and Tuning}
\label{sec:analysis}
To characterize the performance of the PLOT algorithm, we
provide dynamic regret bounds in terms of the total variation of the target dynamics, which is defined as
\begin{equation}\label{eq:path_length_dynamics}
   V_T\triangleq  \sum_{t=1}^{T}\lVert S_t - S_{t-1}\rVert_F.
\end{equation}
For arbitrary $\gamma\in(0,1)$, we have the following guarantees.
\begin{thm}[Dynamic Regret]
\label{thm:main}
Select a prediction horizon $W$ and a forgetting factor $\gamma\in(0, 1)$. Let $\rho$ be the decay rate of the LQT gains as in~\eqref{eq:LQT_gains_exponenitally_decaying} and let $\tilde{W}=\min\{(1-\rho)^{-1},W\}$. The dynamic regret of the PLOT policy, as given by Algorithm \ref{alg:mpc}, is upper bounded by
\begin{align*}
    &\mathcal{R}(\pi) \leq  \alpha_1\rho^{2W}(1-\rho)^{-2}T+\alpha_2 \tilde{W}^4 V_T(1-\gamma)^{-1}\\
   &\quad - \alpha_3 \tilde{W}^2 (T+1)\log \gamma- \alpha_4 \tilde{W}^3 \log(1-\gamma)+\alpha_5  \tilde{W}^3,
\end{align*}
where 
$\alpha_1,\alpha_2,\alpha_3,\alpha_4$ (given in~\eqref{eq:alpha_coeffs}) are positive constants related to system-specific constants, the state dimension $n$, and the target memory $p$.
\end{thm}
The first term in the regret bound captures the truncation effect, i.e., the fact that we only use $W$-step (instead of $T$) ahead predictions in~\eqref{eq:mpc formulation}. The other terms capture the effect of prediction errors on the control performance. Note that there is a tradeoff between large and small prediction horizons. Smaller prediction horizons lead to better prediction performance but worse truncation error and, conversely, larger prediction horizons improve the truncation term but degrade the prediction performance. 
Nonetheless, if we increase the prediction horizon $W$ past the threshold of $(1-\rho)^{-1}$, then the regret guarantees stop degrading\rev{, thus improving prior work~\cite{foster2020logarithmic,baby2022optimal}}. Essentially, $\tilde{W}$ can be thought of as the ``effective" prediction horizon; beyond $\tilde{W}$ any prediction errors have negligible effect. \rev{This is a consequence of deploying multiple step-ahead predictors of varying delays} and the exponentially decaying LQT gains,

Assume that $\gamma$ is close to $1$. Then $-\log\gamma\approx 1-\gamma$ and the dominant regret terms are the second and third terms, $V(1-\gamma)^{-1}$ and $T(1-\gamma)$ respectively. By balancing these two terms, we obtain the following interpretable rate, similar to~\citet{yuan2020trading}. 
\begin{cor}[Tuning]
\label{corollary}%
Select prediction horizon $W=-\frac{\log T}{2\log\rho}$ and forgetting factor $\gamma=1-\sqrt{\frac{\max\{V_T,\log^2T/T\}}{4MT}}$, where $\rho$ is the decay rate of the LQT gains as in~\eqref{eq:LQT_gains_exponenitally_decaying}. Then, the dynamic regret of the PLOT policy is upper-bounded by
\begin{equation*}
    \mathcal{R}(\pi) = \max\{\mathcal{O}(\sqrt{TV_T}), \mathcal{O}(\log T)\}.
\end{equation*}
\end{cor}
To deal with the truncation term, it is sufficient to choose a prediction horizon $W$, which grows logarithmically with the \rev{total time} $T$, as also noted in~\citet{zhang2021regret}. The constant $4M$ in the denominator guarantees that the forgetting factor is positive. To tune the forgetting factor and obtain the rate, we require knowledge of the path length $V_T$, which is not always available. Following the procedure of~\citet{baby2022optimal}, we can overcome this limitation by initiating multiple learners at various time steps and by running a follow-the-leading-history (FLH) meta-algorithm on top,  which can also improve the dynamic regret bound to $T^{1/3}V_T^{2/3}$. We do not explore this possibility here, since we focus on studying the performance of the \rev{RLS algorithm with forgetting factors}. When the path length is close to 0, i.e. the target dynamics are almost static, we obtain logarithmic regret guarantees. In this case, the optimal policy can be rewritten as a static affine control law, with respect to \rev{the} past target states. The regret, in this case, is the static regret with respect to the best (static) affine policy, recovering the result of~\citet{foster2020logarithmic}. 

\begin{remark}[Path Length and Complexity] \label{rem:complexity}
  In certain works~\cite{li2019online, karapetyan2023online}, the dynamic regret is given with respect to the total variation of the target states themselves, that is, $L_T=\sum_{t=2}^{T-1}\lVert r_t - r_{t-1}\rVert$. 
  In contrast, here, we capture learning complexity by the total variation of the target dynamics $V_T$. We argue that the former notion of complexity can be suboptimal in the case of targets with dynamic structure, while ours is superior in this setting. For example, consider a target that is moving on a circle with constant velocity--see \cref{exmp:circular}. Then, the former path length is linear $L_T=\mathcal{O}(T)$ while the latter is zero $V_T=0$. \rev{In control applications, many targets of interest have dynamic structures, e.g. tracking other drones, tracking moving objects, etc. Nonetheless, if the targets are fully unstructured and arbitrary, PLOT might not perform as well since it tries to fit a dynamic model to arbitrary data--see Appendix~\ref{app_sec:random}.} \rev{Still, the regret order cannot be worse than the one in prior work if we are learning over ARX models--see discussion in Appendix~\ref{app_sec:trajectory_examples}. }  
\end{remark}
The full proof of Theorem~\ref{thm:main} and Corollary~\ref{corollary} can be found in Appendix~\ref{app_sec:control_proof}. In the following, we provide a sketch of the proof.
First, similar to \citet{foster2020logarithmic}, we invoke the ``performance difference lemma"~\cite{kakade2003sample} to turn the control problem into a prediction one \footnote{In our case, the optimal controller is the true minimizer of the cost, hence, the regret in this case becomes the advantage function of~\cite{foster2020logarithmic}.}.
\begin{lem}[Performance Difference Lemma~\cite{foster2020logarithmic}]
For any policy $\pi$ such that $u^{\pi}_t=-K(x_t-r_t)+\hat{q}_t$, where $K$ is given by~\eqref{eq:K*}, and the optimal action defined in \eqref{eq:optimal action}, the dynamic regret can be written as
\begin{equation*}
\label{lemma:mpc regret}
   \mathcal{R}(\pi) =\sum_{t=0}^{T-1}\lVert \hat{q}_t - q_t(r_{t:W})\rVert^2_{(R+B^TXB)}.
\end{equation*}
\end{lem}
The above fundamental result shows that the control problem can be cast as a prediction problem, where we try to predict the affine part of the optimal control law.
Second, we bound the prediction error by analyzing the dynamic regret of the RLS algorithm following the steps of~\citet{yuan2020trading}. 
Define the prediction regret for $k-$step ahead prediction as
\[
 \mathcal{R}^{(k)}_{\mathrm{pred}}\triangleq\sum_{t=k-1}^T\snorm{\hat{S}_{t|t-k}z_{t-k}-r_t}^2. 
\]
Since the sequence $r_{0:T}$ satisfies~\eqref{eq:exosystem}, there is no subtracted term in the regret; matrices $S_{t|t-k}$ achieve zero error. In the Appendix, we prove regret bounds for the prediction problem that hold for non-realizable sequences as well.
\begin{thm}[Regret for AR system prediction]
\label{thm:ar_prediction_regret}
Define the total variation of the  $k$-step ahead dynamics as
$V^k_{T}\triangleq \sum_{t=k}^{T-k}\snorm{S_{t+k|t}-S_{t|t-k}}$. The dynamic regret of the RLS algorithm (Algorithm~\ref{alg:RLS}) for $k$-steps ahead prediction is upper bounded by
\begin{align*}
    &\mathcal{R}^{(k)}_{\mathrm{pred}} \le  \frac{\beta_1}{1-\gamma}V^k_T+\beta_2T\log \frac{1}{\gamma}+k\beta_3\log \frac{1}{1-\gamma}+k\beta_4,
\end{align*}
where 
$\beta_1,\beta_2,\beta_3,\beta_4$ are given in~\eqref{app_eq:betas} and are positive constants related to system-specific constants, the state dimension $n$, and the target memory $p$.
\end{thm}
Again, we can select the forgetting factor similar to Corollary~\ref{corollary}, to obtain more specific regret bounds.
 The dominant dimensional dependence is hidden in the coefficient $\beta_2$ and is of the order of $np^2$ which is worse by $p$ compared to the optimal linear regression rate of $np$. This is not a limitation of the RLS algorithm but an artifact of the norm of $z_t$ scaling with $\sqrt{p}$. We can overcome this limitation by imposing ``fading memory" constraints on the dynamics. If we rewrite the dynamics as $S_{t} z_{t-1}=S^{[1]}_{t} r_{t-1}+\dots+S^{[p]}_{t} r_{t-p}$, then fading memory constraints would force $S^{[i]}_t$ to decay exponentially with $i$. 
 
To conclude the proof, we combine the previous steps while also accounting for the truncation effect. The only remaining step is to upper bound the $k$-steps ahead path lengths $V^k_T$ in terms of the path length of the one-step-ahead path length $V_T$ as defined in~\eqref{eq:path_length_dynamics}. In Lemma~\ref{app_lem:perturbation}, we prove that $V^k_T\le \sqrt{p} k^2M^2 V_T$. This is why the factor $\tilde{W}^4$ that multiplies the path length in the final regret bound in Theorem~\ref{thm:main} appears with an exponent of $4$ instead of just $2$.

\section{Simulations and Experimental Validation}
\label{sec:case_study}

In this section, we demonstrate the performance of PLOT both in simulation and in hardware experiments. 
We study the setting of tracking adversarial unknown targets with a quadrotor, which is of particular interest given the number of potential applications, such as hunting adversarial drones in airports \cite{dressel2019hunting} or artistic choreography \cite{dance}, \rev{to name a few}. \rev{The code for both the simulation and hardware experiments can be accessed from \footnotesize{\url{https://gitlab.nccr-automation.ch/akarapetyan/plot}}.}

\begin{figure*}[t]
\centering
    \begin{subfigure}[t]{0.23\textwidth}
        \centering
        \includegraphics[width=0.96\linewidth]{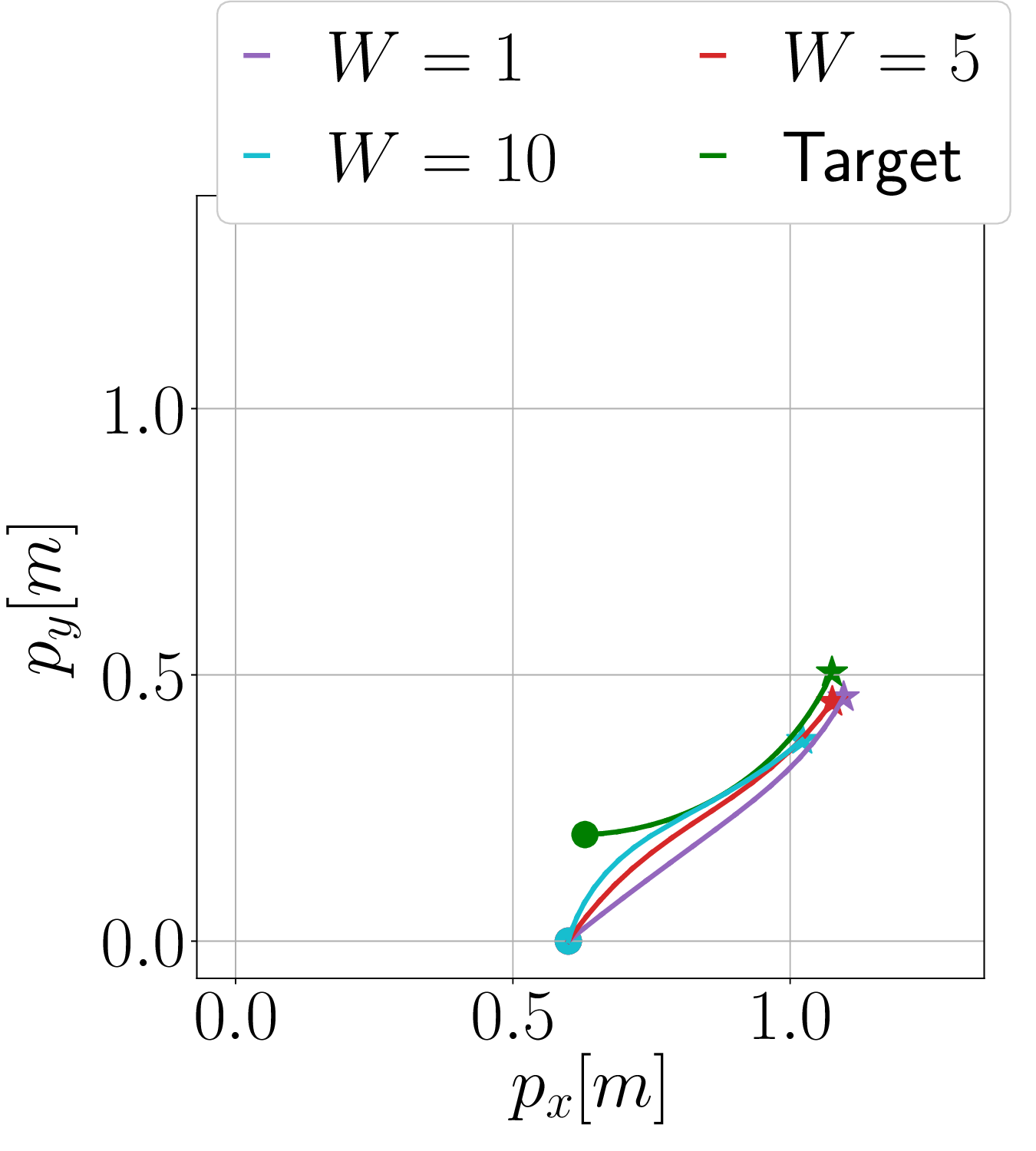}
        \caption{$T=2$ seconds}
    \end{subfigure}%
    ~ 
    \begin{subfigure}[t]{0.23\textwidth}
        \centering
        \includegraphics[width=0.96\linewidth]{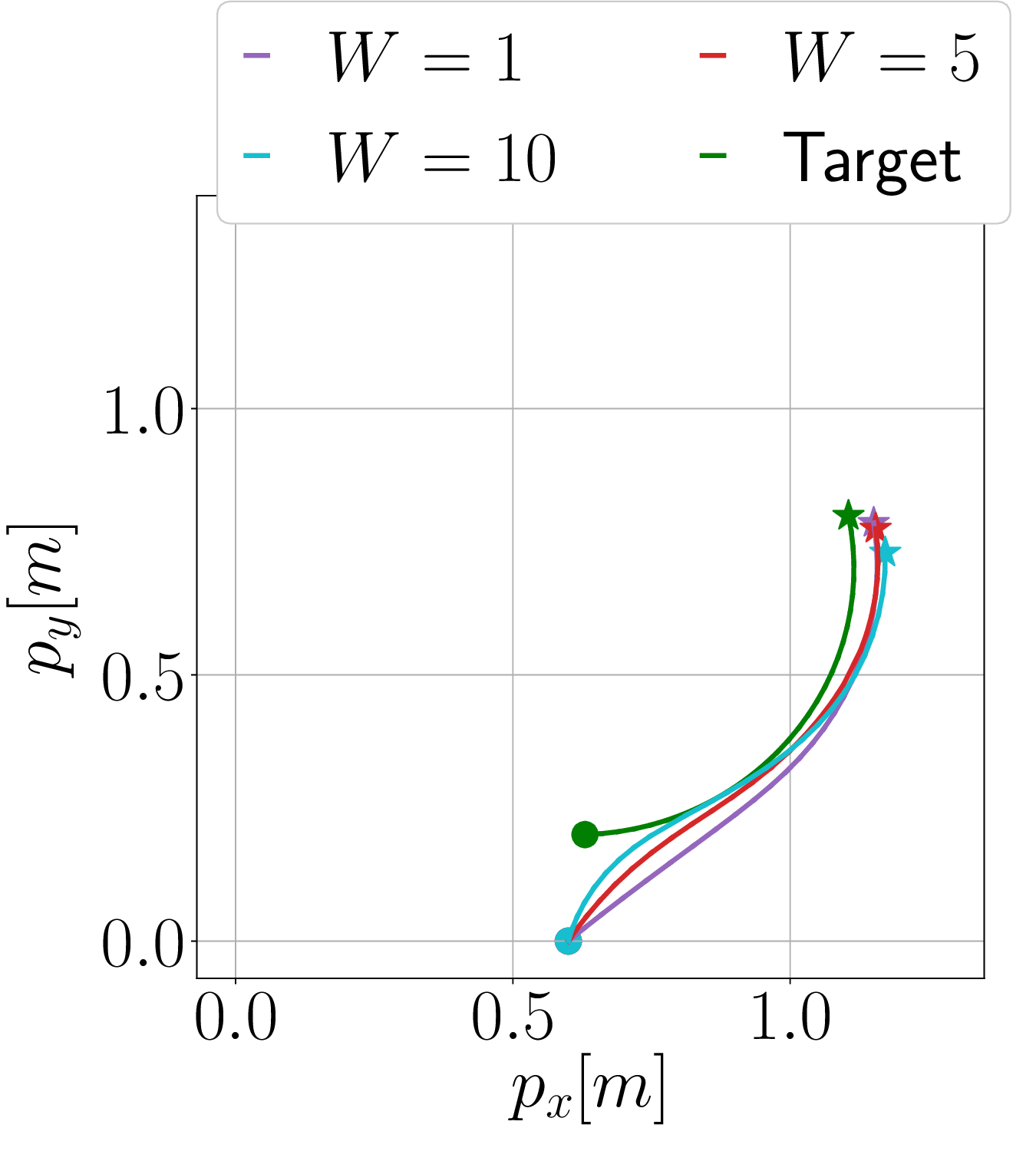}
        \caption{$T=3$ seconds}
    \end{subfigure}
    ~ 
    \begin{subfigure}[t]{0.23\textwidth}
        \centering
        \includegraphics[width=0.96\linewidth]{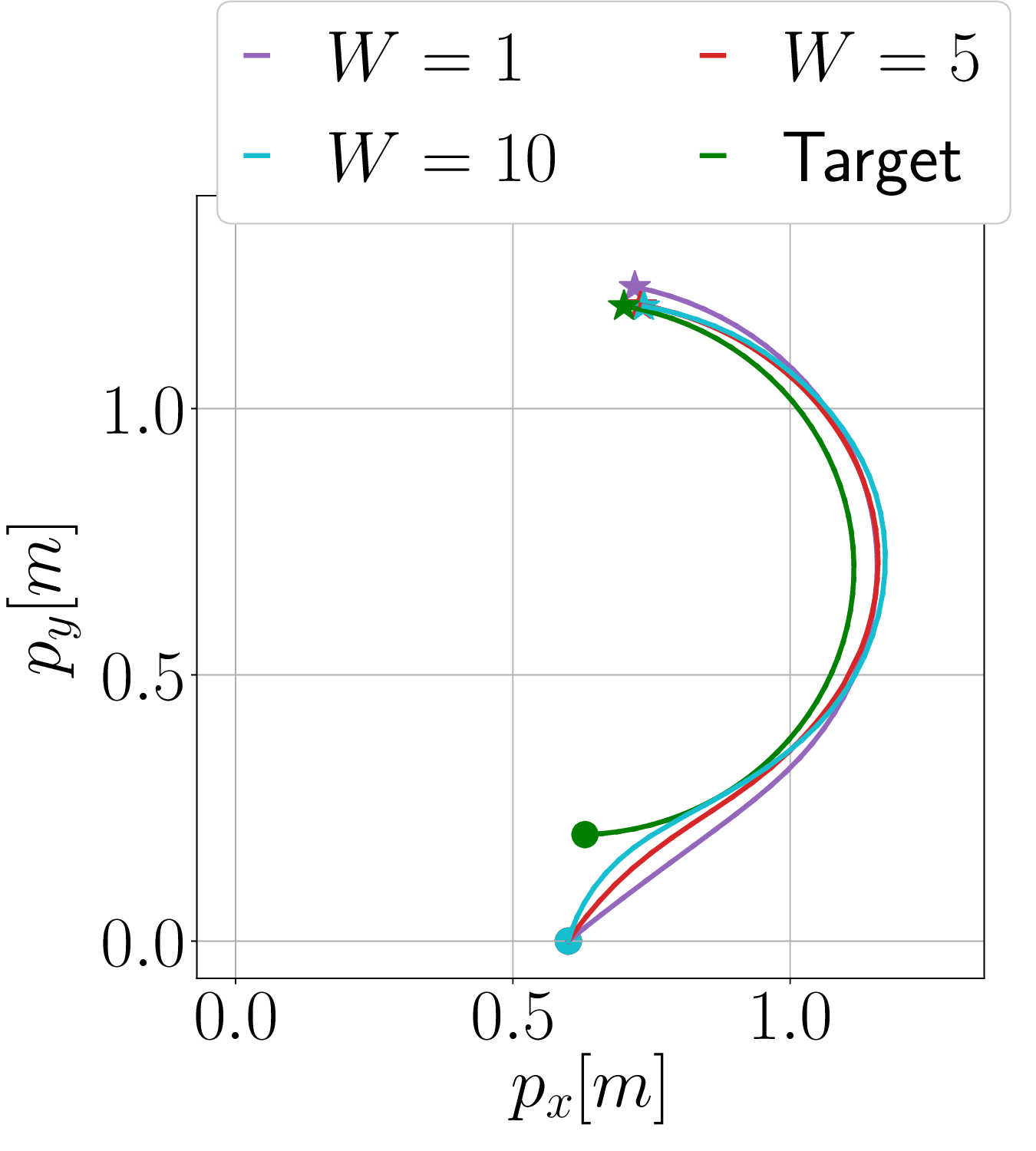}
        \caption{$T=5$ seconds}
    \end{subfigure}
    ~ 
    \begin{subfigure}[t]{0.23\textwidth}
        \centering
        \includegraphics[width=0.96\linewidth]{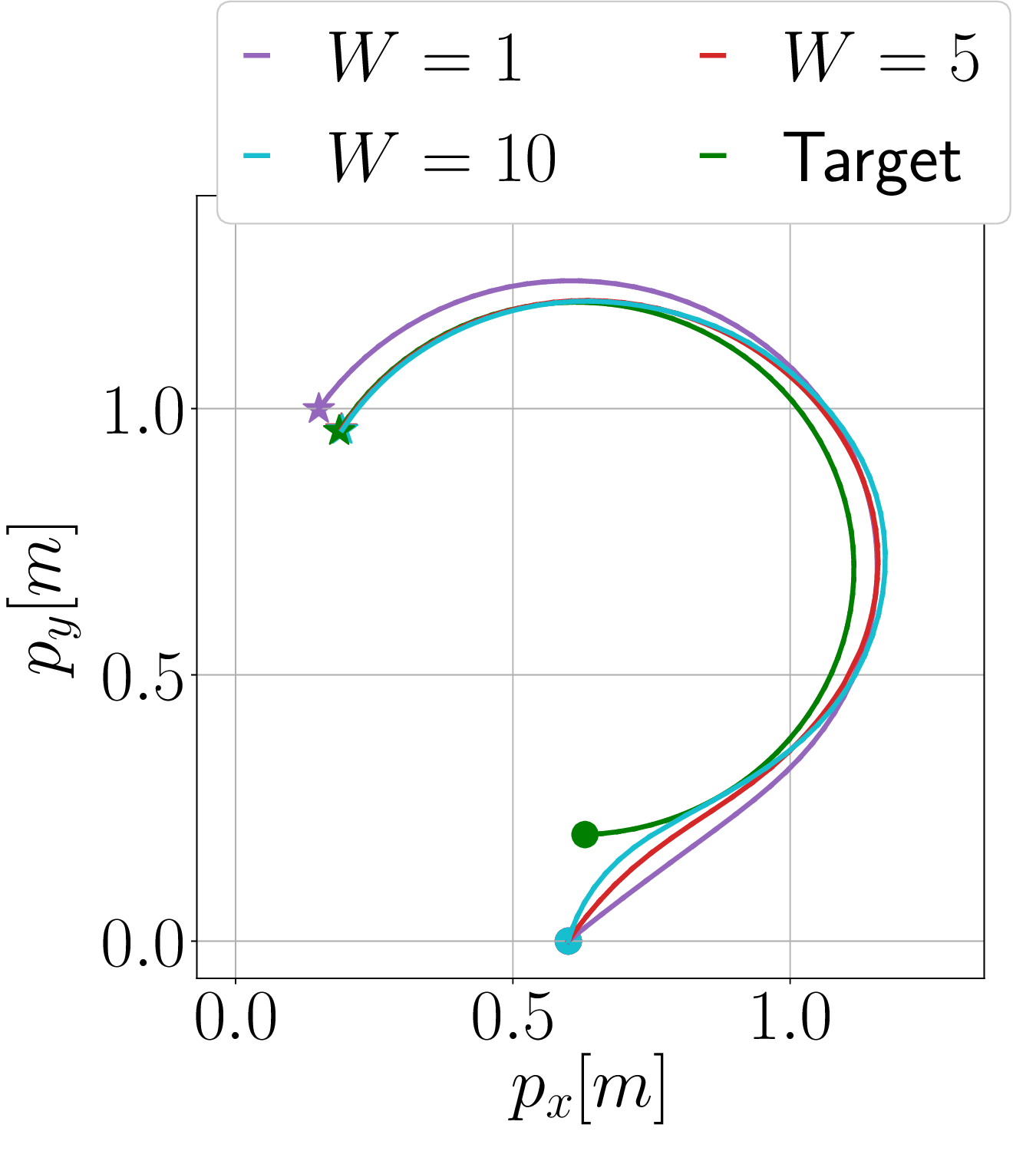}
        \caption{$T=7$ seconds}
    \end{subfigure}
    \caption{Trajectory plots of a circular target with a $V_T=0$ path length and the PLOT Algorithm for varying prediction horizon lengths, simulated for $T=2,3,5$ and $T=7$ seconds.}
    \label{fig:circle_rls_rhc_video_main}
\end{figure*}

We consider the Crazyflie 2.1 quadrotor~\cite{bitcraze}, a versatile, open-source, nano-sized quadrotor developed by Bitcraze \cite{bitcraze}, whose dynamics can be modeled with a linear time-invariant system linearized around a hovering point \cite{beuchat2019n}. For such a model, we define the state $\textstyle{x:=[\vec{p}~;~ \vec{\dot{p}}~;~\vec{\psi}]}$, where $\textstyle{\vec{p}:=\matr{x &y &z}^\top}$ is the position vector in the inertial frame, $\textstyle{\vec{\psi}:=\matr{\gamma& \beta& \alpha}^\top}$ is the attitude vector in the inertial frame with $\gamma$, $\beta$ and $\alpha$ for the roll, pitch and yaw angle, respectively. The action $\textstyle{u:=[f~;~\vec{\omega}]}$ includes the total thrust $f$ and the angular rate $\textstyle{\vec{\omega}:=\matr{\omega_x& \omega_y& \omega_z}^\top}$ in the body frame. For the detailed derivation of the linear dynamic model, see Appendix \ref{app:quadrotor_model} and \cite{beuchat2019n}. For all examples, we take the sampling time to be $T_s=0.1$ seconds and the LQR cost matrices fixed at $Q:=diag(80, 80, 80, 10, 10, 10,0.01,0.01,0.1)$, and $R:=diag(0.7, 2.5, 2.5, 2.5)$.

\subsection{Simulation Results}
\label{sec:simulation}

Here, we consider the linearized dynamics of the Crazyflie quadrotor, derived in Appendix \ref{app:quadrotor_model}.\rev{We show how PLOT successfully tracks both a static reference target with $V_T=0$ path length and a dynamic target with $V_T=\mathcal{O}(\sqrt{T})$. In addition, we also demonstrate how the proven regret guarantees can be used as a guideline for tuning the algorithm parameters such as the prediction horizon length $W$ and the forgetting factor $\gamma$. We also provide a performance comparison to other online control algorithms from the literature for tracking a non-stationary target}.

\subsubsection{Target with a $V_T=0$ path length}

Given the dynamic regret analysis and the main result in Theorem \ref{thm:main}, the smaller the path length $V_T$ the easier the control task is for PLOT. To show its tracking performance for such a target, \rev{as well as to visualize the effect of the prediction horizon $W$ on the control performance}, we consider the circular target  detailed in Example \ref{exmp:circular} that has $V_T=0$ path length. With this target revealed and measured online, as described in Section \ref{sec:problem statement}, we run  PLOT repeatedly for various horizon lengths. The trajectory plots of the online target and the quadrotor are shown in Figure \ref{fig:circle_rls_rhc_video_main} for the first $T=2,3,5$ and $T=7$ seconds. The regret plots for the considered multiple horizon lengths are shown in Figure \ref{fig:regret_rls_rhc_main} for a simulation of $T=200$ seconds. 
\begin{figure}
  \centering
\includegraphics[width=0.71\linewidth]{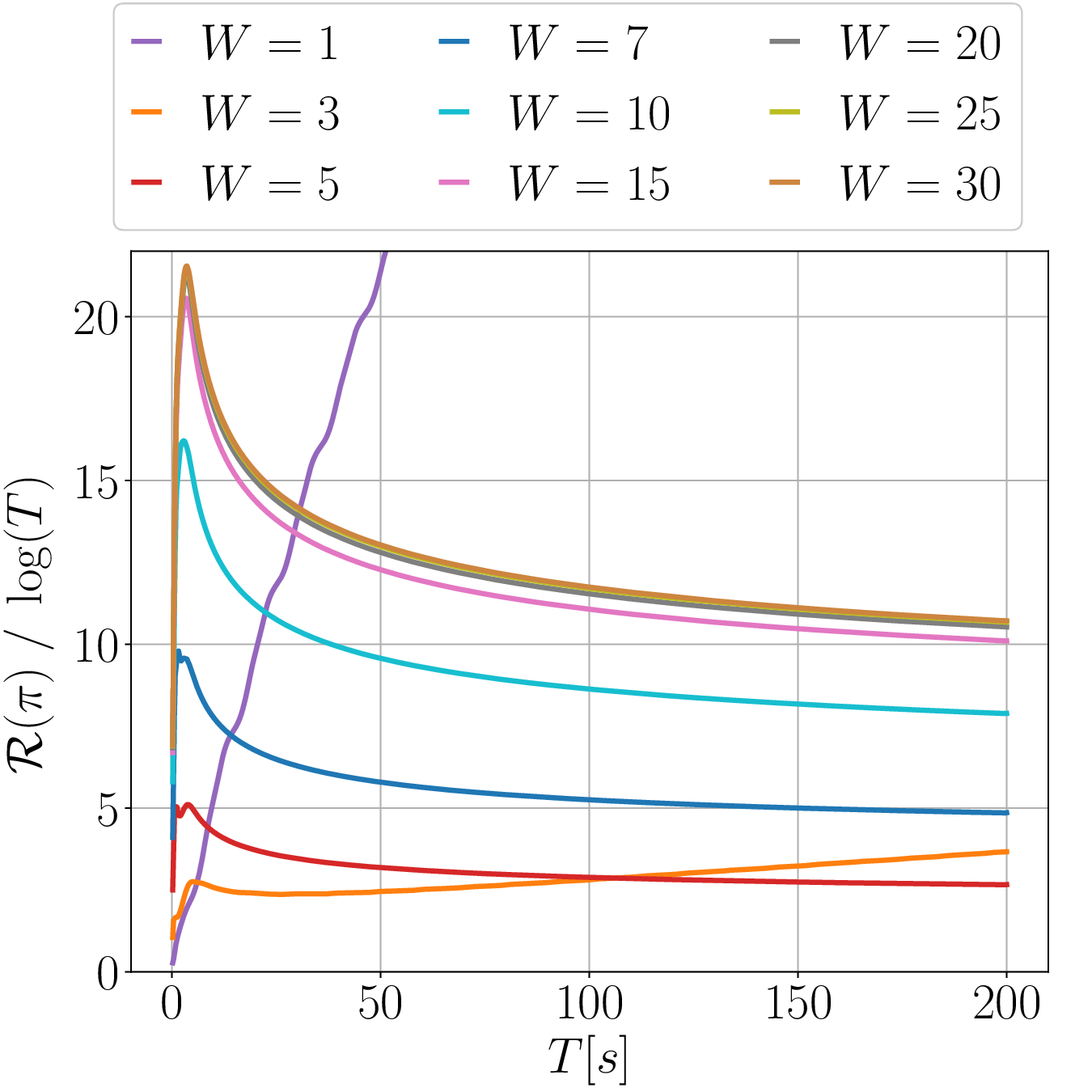}
  \caption{Log-normalized regret of the PLOT Algorithm with a range of prediction horizon lengths, simulated over a horizon of $T = 200$ seconds.}
  \label{fig:regret_rls_rhc_main}
\end{figure}
Firstly, both figures show that for too small $W$-s, PLOT exhibits a poorer tracking and regret performance \rev{due to the horizon truncation error}. As expected from Corollary \ref{corollary}, \rev{given the nature of the reference target}, PLOT with longer prediction horizons \rev{performs} better by learning the reference dynamics online, achieving sublinear regret as verified by Figure \ref{fig:regret_rls_rhc_main}. 
Additionally, the regret is increased for larger $W$-s \rev{only} up to a certain \rev{saturated} level. 
This behavior, \rev{as discussed in Section \ref{sec:analysis}, is thanks to the multi-predictor setup of PLOT, making its regret scale with $\Tilde{W}$, as opposed to linearly with $W$.} See Appendix \ref{sec:circle_reference} for further details.

\rev{
\subsubsection{Target with a $V_T = \mathcal{O}(\sqrt{T})$ Path Length}

For more challenging targets, one can make use of Corollary \ref{corollary} to tune for the best forgetting factor $\gamma$ if the path length $V_T$ is known. Here, we consider the case when only the order of the path length of a given reference is known, and, using Corollary \ref{corollary}, fix $\gamma_a = 1 - c_{\gamma}T^{-a}$, where $c_\gamma \in \mathbb{R}_+$ is a constant independent of $V_T$ and $T$, while  $a \in \mathbb{R}_+$ is such that $\gamma_a \in (0,1]$ and is tuned based on $V_T$ and $T$.

We fix a spiraling reference target with $V_T = \mathcal{O}(\sqrt{T})$, with exact dynamics detailed in Appendix \ref{sec:spiral_reference}. This order of $V_T$  results in the optimal value of $a=0.25$ and a forgetting factor $\gamma_{0.25}=0.78$ from Corollary \ref{corollary}, with $c_\gamma = 1.50$ obtained experimentally. 

\begin{figure}[!ht]
  \centering
  \includegraphics[width=0.71\linewidth]{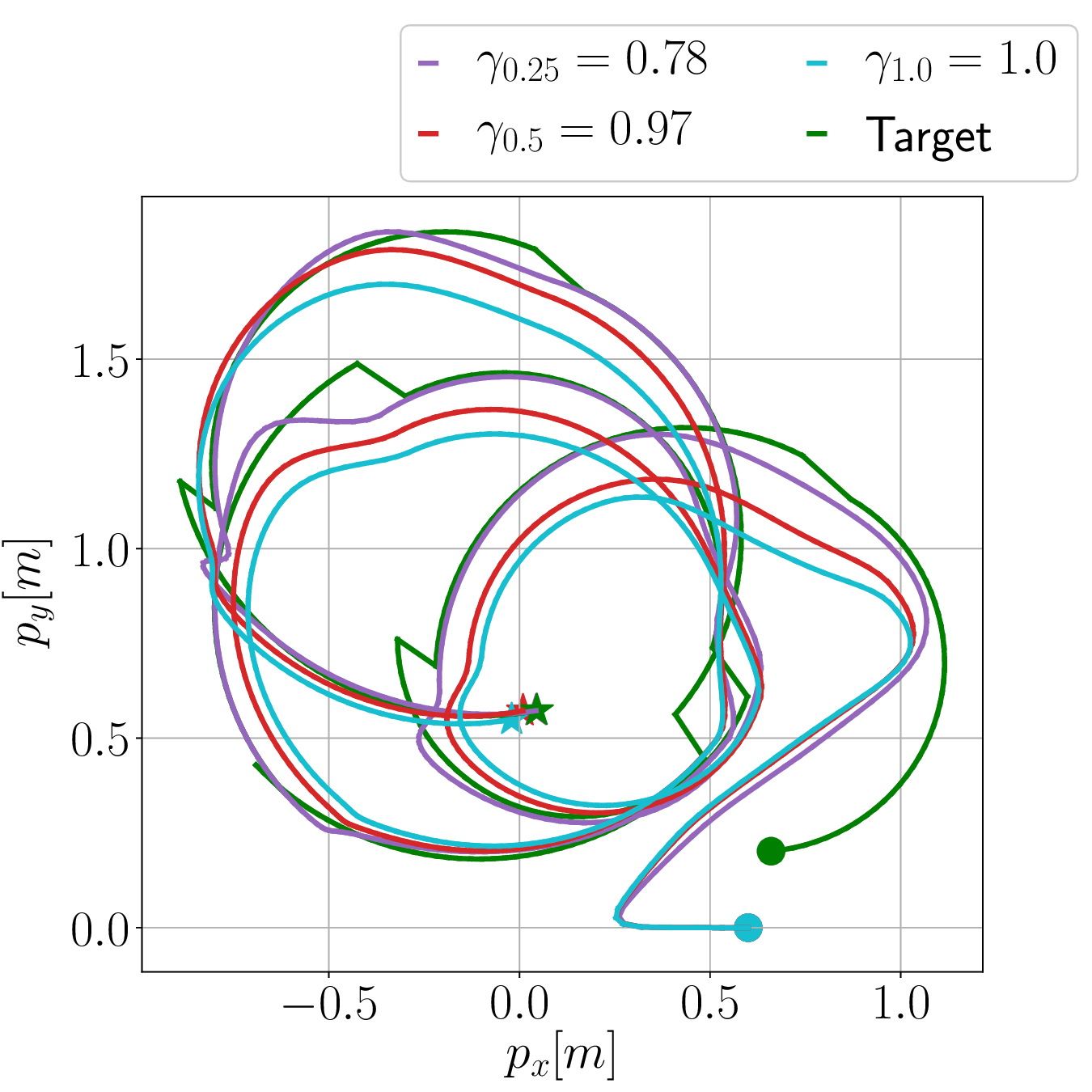}
  \caption{Trajectory plot of a spiral with a $V_T = \mathcal{O}(\sqrt{T})$ tracked with PLOT with a $W=5$ and a range of values for $\gamma$.}
  \label{fig:gamma_trajectory_main}
\end{figure}

PLOT's performance with three difference forgetting factors is shown in the trajectory plot in Figure \ref{fig:gamma_trajectory_main} for $T=40$ seconds, showing a better tracking performance for the regret-optimal $\gamma_{0.25}=0.78$. To show that $c_\gamma$ is independent of $T$, we fix, $c_\gamma$ to the tuned value of $1.50$ and perform experiments with varying horizon lengths from $T=150$ to $T=300$ seconds. For each $T$, PLOT is run with $4$ different forgetting factors with $a=0.1,~0.25,~0.5$ and $1$ (corresponding to $\gamma=1$), and the regret at the end of each experiment is shown in Figure \ref{fig:tuned_gamma_main}. The results confirm that  $\gamma_{0.25} = 1-1.5T^{-0.25}$ consistently outperforms other forgetting factors in terms of regret, as expected from Corollary \ref{corollary}. Further details on this example are provided in Appendix \ref{sec:spiral_reference}.

\begin{figure}[ht]
    \centering
    \includegraphics[width=0.71\linewidth]{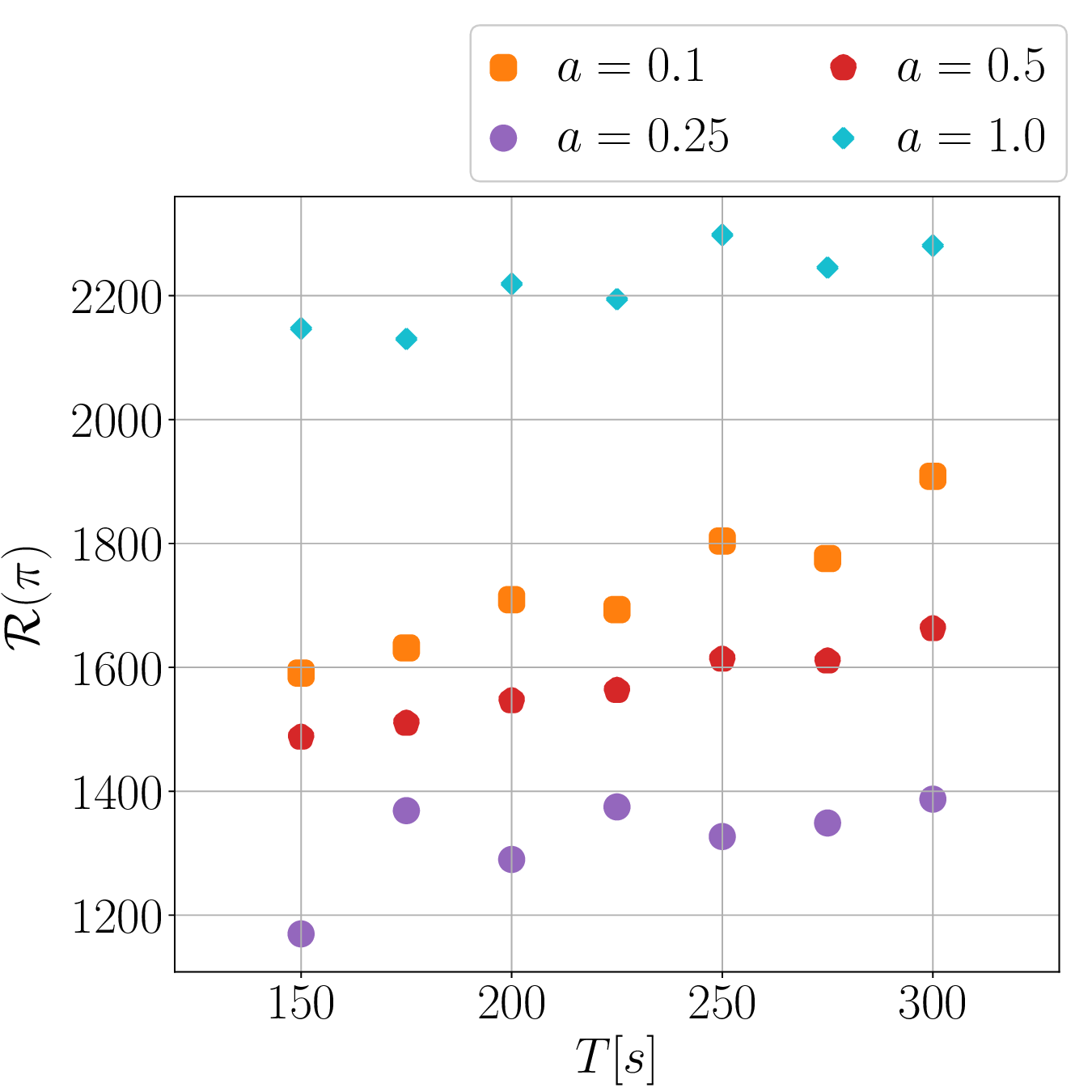}
    \caption{Regret of PLOT with varying $\gamma_a = 1-c_\gamma T^{-a}$; $\gamma_{0.25}$ is regret-optimal based on Corollary \ref{corollary}.}
    \label{fig:tuned_gamma_main}
\end{figure}
}

\subsubsection{Comparison with Other Online Control Methods}

Next, we compare the dynamic regret of PLOT to other controllers. Namely, the disturbance-action policy (DAP) with memory proposed in \cite{agarwal2019online}, the Riccatitron algorithm from \cite{foster2020logarithmic}, the SS-OGD algorithm by \cite{karapetyan2023online}, the Follow the Leader (FTL) algorithm \cite{pmlr-v32-abbasi-yadkori14}, as well as the Naive LQR controller that \rev{applies only an optimal state error feedback input} without any affine term.
The reference target is chosen to follow the  dynamics 
\begin{equation*}
    r_{t+1} =     S_{t+1|t} = 
    \begin{bmatrix}
        1 & 0 & T_s & 0\\
        0 & 1 & 0 & T_s\\
        0 & 0 & s_t\cos{\theta_t} & -s_t\sin{\theta_t}\\
        0 & 0 & s_t\sin{\theta_t} & s_t\cos{\theta_t}
    \end{bmatrix}r_t,
\end{equation*}
where $s_0 = 1, \theta_0 = 0.06$, $T_s = 0.1 s$ and $s_k = -s_{k-1}$, $\theta_k = -0.99\times\theta_{k-1}$ for every $k = \sqrt{T}$ given some $T$ \rev{that we simulate $T=200$ seconds.} 

\begin{figure}
    \centering
    \includegraphics[width=0.7\linewidth]{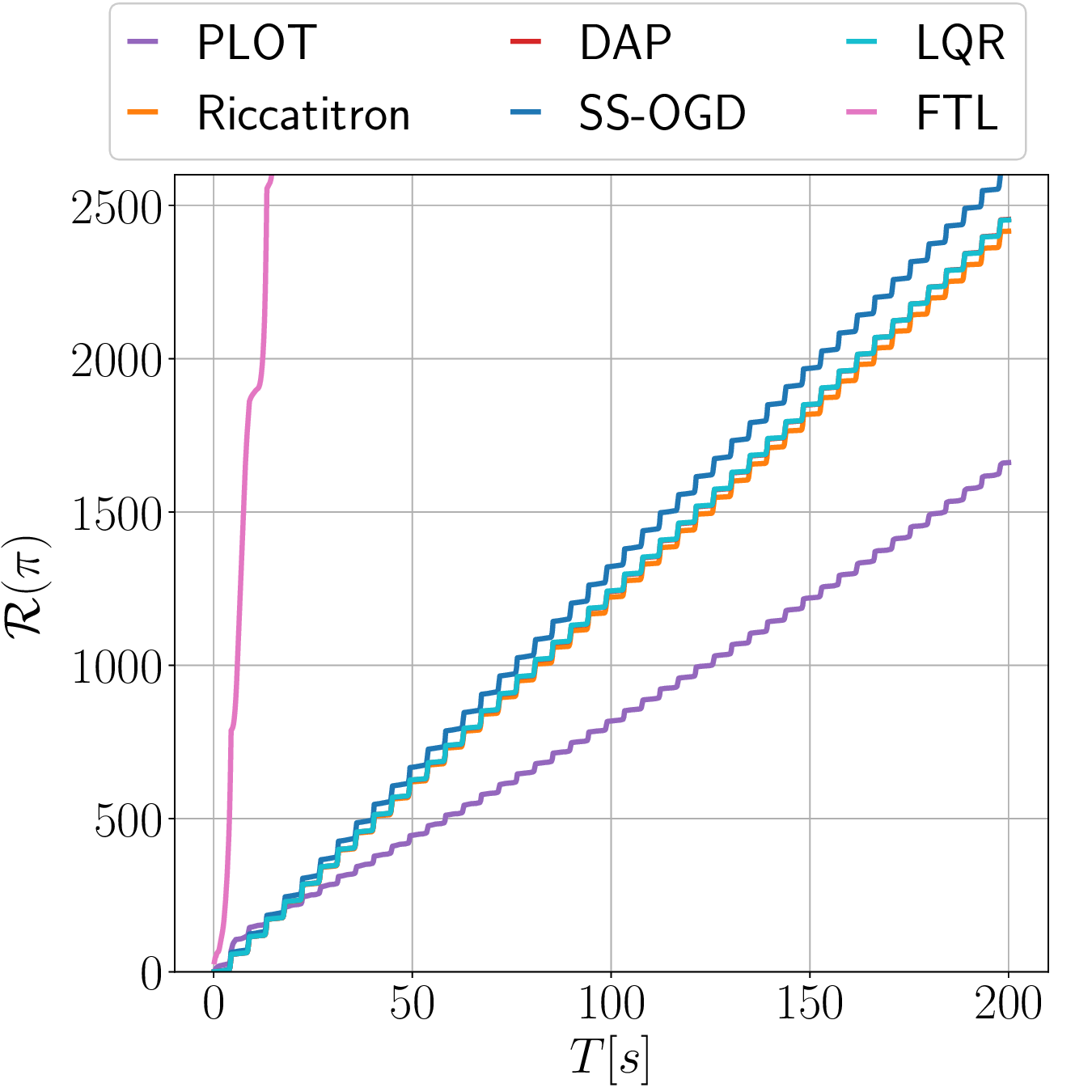}
    \caption{Dynamic Regret of online control algorithms applied to the online tracking problem.}
    \vspace{-0.25cm}
\label{fig:benchmark_regrets_main}
\end{figure}

Figure \ref{fig:benchmark_regrets_main} shows the accumulated regret of all algorithms. For the given challenging example, PLOT outperforms the others. This is mainly because, unlike all the others, PLOT implements a dynamic approach with a forgetting factor adapting to the fast-changing reference on time. Moreover, it deploys an indirect approach of learning the dynamics of the reference and then using it in the control, as opposed to a direct approach of learning the affine control term, implemented by all the others apart from \rev{Naive} LQR. For the given dynamic target the direct and static approaches produce a much smaller affine term compared to PLOT, leading to a performance very close to that of Naive LQR. \rev{The indirect approach of PLOT, however, may perform worse in cases where the reference targets lack dynamics and are mutually independent at each time step, as discussed in Remark \ref{rem:complexity}. Further details on this, as well as the considered algorithms' implementation are provided in Appendix \ref{app:benchmarks}.} 

\subsection{Experimental Validation on Quadrotors}
\label{sec:hardware}
\rev{We} validate the proposed algorithm on \rev{the Crazyflie 2.1} quadrotor. Due to the non-linear dynamics of the quadrotor and the noisy state measurements, the hardware experiments have tight requirements for the stability and robustness of the control algorithm. 
In the hardware experiments, we define the states and actions \rev{to be} the same as in the simulations. 
\rev{The control inputs are sent to the quadrotor through a radio communicator from a centralized computer that receives state measurements through a local area network connection to the motion capture system, and runs the online controller.}
The virtual reference trajectory is generated online after the drone has successfully taken off and is at a predefined hovering position.
The controller receives the target at a rate of $10$ Hz during flight. %

We implement PLOT for the linearized model of the drones derived in Appendix \ref{app:simulations} with a fixed prediction horizon of $W=5$ and a forgetting factor of $\gamma=0.8$. Figure~\ref{fig:figure8_crazyflie_main} shows the trajectory plots of PLOT and \rev{Naive} LQR for a horizon of $T=40$ seconds \rev{and an ``infinity"-shaped reference target.}. The \rev{Naive} LQR controller exhibits a delayed tracking behavior as expected, while PLOT achieves a smaller tracking error. Further details on the practical implementation of the algorithm are provided in Appendix \ref{app:implementation}.

\begin{figure}[h]
\centering
\includegraphics[width=0.71\linewidth]{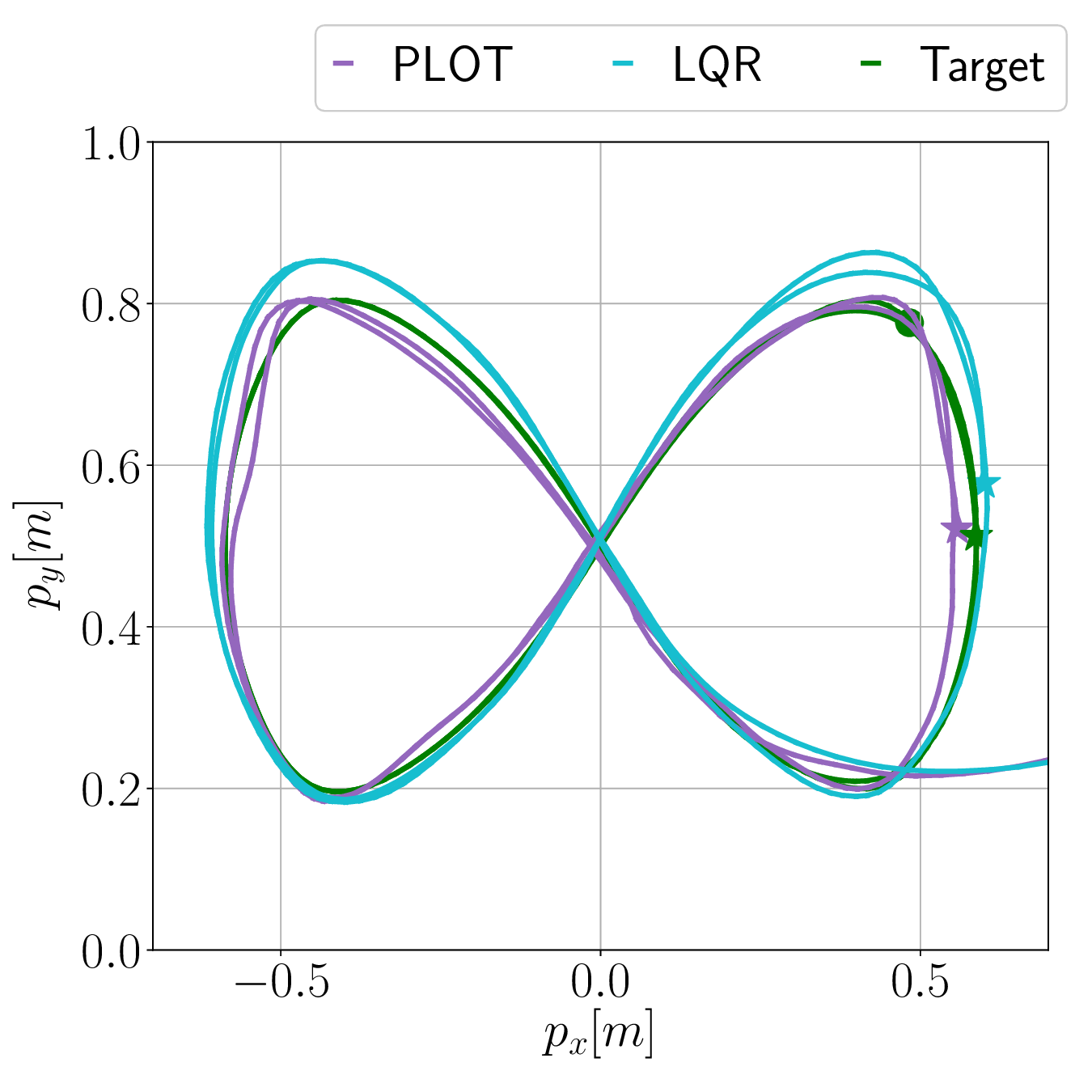}
  \caption{Trajectory plot of an ``infinity"-shaped reference tracked by the drone with the PLOT and Naive LQR controllers.}
  \vspace{-0.5cm}
  \label{fig:figure8_crazyflie_main}
\end{figure}

\section{Conclusion}
\label{sec:conclusion}
We studied online tracking of unknown targets for linear quadratic problems and provided an algorithm with dynamic regret guarantees. By exploiting dynamic structure, we can obtain sharper guarantees compared to prior work that assumes no structure.  It is open whether our regret analysis can be extended to \emph{non-realizable} targets, that is, targets that do not obey~\eqref{eq:system} exactly for matrices that satisfy Assumption~\ref{assumption:stability}. While the guarantees for prediction still hold, it is open whether the guarantees for the controller are retained. A benefit of our implicit control paradigm is that we can decompose tracking into prediction and certainty equivalent control. This makes the approach promising for Model Predictive Control tracking applications. Another open question is whether the guarantees can be extended to output tracking or non-quadratic costs.

\clearpage
\section*{Acknowledgements}

This work has been supported by the Swiss National Science Foundation under NCCR Automation (grant agreement $51\text{NF}40\_180545$), and by the  European Research Council under the ERC Advanced grant agreement  $787845$ (OCAL).

\section*{Impact Statement}

This paper presents work whose goal is to advance the field of Machine Learning. There are many potential societal consequences of our work, none which we feel must be specifically highlighted here.

\bibliography{literature_tracking}
\bibliographystyle{icml2024}

\newpage

\appendix
\onecolumn
\section*{Organization}
In Appendix~\ref{app_sec:LQT} we discuss properties of LQT control and the equivalence to the problem of LQR with disturbances. In Appendix~\ref{app_sec:AR}, we discuss issues regarding the dynamic representation of the targets. We present analytical expressions for the multi-step ahead predictors. We also discuss how the dynamic structure of the targets can lead to improved learning performance. Finally, we provide an expression for the total variation norm of multi-step ahead predictors in terms of the total variation of the system matrices $S_t$. The proof of Theorem~\ref{thm:ar_prediction_regret} can be found in Appendix~\ref{app_sec:prediction_proof}, where we also provide more details for the prediction problem. The proof of Theorem~\ref{thm:main} and Corollary~\ref{corollary} can be found in Appendix~\ref{app_sec:control_proof}.
Simulations visualizing the performance of PLOT, its hyperparemeter tuning based on the regret analysis, as well as comparison with other benchmark algorithms are provided in Appendix \ref{app:simulations}. Appendix \ref{app:implementation} provides the technical details of the implementation of PLOT on the Crazyflie quadrotors and its tracking performance on those. The code for the simulation and hardware experiments is provided in \url{https://gitlab.nccr-automation.ch/akarapetyan/plot}. 

\section{Additional Related Work}\label{app_sec:remaining_literature}

\paragraph{Online and non-stochastic control.} Online and non-stochastic control techniques have been applied to systems with known~\cite{agarwal2019online} as well as unknown system dynamics~\cite{simchowitz2020making}. In this paper, we focus on the case of known \emph{system} dynamics, where the main purpose of the online controller is to react to (non-stochastic) \emph{unknown} disturbances/targets.\footnote{Note that as shown in many works, see~\cite{foster2020logarithmic,karapetyan2023online}, the problem of tracking is equivalent under certain conditions to the problem of disturbance rejection. Hence, for the purpose of this literature review, we can use disturbance/target interchangeably.} In this setting, \citet{agarwal2019online}  achieved $\mathcal{O}(\sqrt{T})$ regret for general convex stage cost functions under a disturbance affine feedback policy. In the case of quadratic costs, and more generally strongly convex stage costs\footnote{Certain additional conditions are required.}, the guarantees were improved to logarithmic~\cite{foster2020logarithmic,simchowitz2020making,naman2019logarithmic}. The aforementioned works consider static problems in the sense that they compete against the best stationary policy. \rev{The dynamic setting was studied in~\cite{zhao2022non}, where dynamic regret guarantees of the order of $\mathcal{O}(\sqrt{TV_T})$ for general convex costs were derived for the first time, using gradient-based (first order) base learners.} In~\cite{baby2022optimal} dynamic regret guarantees for the quadratic cost case (LQR) were derived using a variant of the Follow the Leading History (FLH) algorithm. It was proved that the dynamic regret with respect to any dynamic disturbance affine policy is upper bounded by $T^{1/3}V^{2/3}_T$, where $V_T$ is the total variation of the parameters of the policy. 
Learning  disturbance affine feedback policies follows the direct online control paradigm since we adapt the policy parameters directly.

In this paper, we study implicit online control architectures, where we decouple online prediction of disturbances/targets from control design. Several works have studied the effect of predictions on receding horizon control~\cite{powerofprediction,li2019online,zhang2021regret}. However, either no method of prediction is provided or the dynamic regret guarantees are suboptimal in our setting. In particular, the regret guarantees scale with the total variation of the disturbance/target itself $\sum_{t=0}^{T-1}\snorm{r_{t+1}-r_t}$~\cite{li2019online}, which is appropriate for unstructured or slowly-varying targets but is suboptimal in the case of targets with dynamic structure; there exist targets for which the total variation of their state $r_t$ is linear, while the total variation of the target dynamics $S_t$ is zero--see Appendix~\ref{app_sec:trajectory_examples}.

 \citet{pmlr-v32-abbasi-yadkori14}~introduced the problem of non-stochastic control for target tracking, where the target can be adversarial. Logarithmic regret guarantees are provided for learning affine control policies with constant affine terms. However, policies with constant affine terms can only track static points as targets; they are not rich enough to track trajectories that are not points and, in general, time-varying targets.   Recently, \citet{niknejad2023online} studied online learning for target tracking in the case of unknown system dynamics and unknown time-invariant target dynamics. Our setting is different since we have known dynamics for the system, time-varying dynamics for the target, and we use dynamic regret. 
 Their static regret guarantee of $T^{2/3}$ is conservative for our setting, where we can obtain logarithmic regret in the static case, however, it applies to general convex costs and unknown dynamics. In \cite{karapetyan2023online}, dynamic regret guarantees are provided for target tracking, under a first order method which treats the targets as slowly varying or steady-state. Similar to~\cite{li2019online}, the dynamic regret depends on the total variation of the target state which can be suboptimal in the case of target with dynamics. \citet{zhang2022adversarial} proved $\tilde{\mathcal{O}}(\sqrt{|\mathcal{I}|})$ (for any interval $\mathcal{I}$) adaptive regret guarantees for tracking in the case of convex costs. Their setting is different since no dynamic structure is assumed for the targets, the comparator class is controllers which are constant over the interval of interest, while we focus on square losses and dynamic regret.

Although related, our setting is different from~\cite{gradu2023adaptive}
and~\cite{minasyan2021online}
where the system itself is linear time varying, there is no target, and adaptive regret is considered as a performance metric. Instead, here, the system's dynamics are time invariant, 
the target's dynamics are unknown and time varying, and we consider dynamic regret.  Safe online non-stochastic control subject to constraints has also been studied before~\cite{zhou2023safe,nonhoff2023online,li2021constraint}. Dealing with constraints is a more challenging problem in general. In many cases, the analysis is restricted to the static case~\cite{li2021constraint}, requires relaxing the notion of policy regret~\cite{zhou2023safe} or the path length scales with the target state~\cite{nonhoff2023online}, which is not always optimal in our setting as we discuss above. 

Finally, we note that when the \emph{system} dynamics are \emph{unknown}, it is no longer possible to achieve logarithmic regret in the general case, even under stochastic disturbances~\cite{simchowitz2020naive,ziemann2022regret}. It is possible to achieve $\tilde{\mathcal{O}}(\sqrt{T})$ regret~\cite{chen2021black,simchowitz2020making, mania2019certainty} instead.

\paragraph{Prediction and online least squares.} Online linear regression and the least squares method, in particular, have been studied extensively before~\cite{azoury2001relative,cesa2006prediction}. \citet{anava2013online} provided logarithmic regret guarantees for online prediction of auto-regressive (AR) and auto-regressive with moving average (ARMA) systems using the online Newton step (ONS) method. The result holds for exp-concave prediction loss functions, which includes the square prediction loss. Similar guarantees were proved for partially observed stochastic linear systems in the case of Kalman filtering~\cite{ghai2020no,tsiamis2022online}. Under first-order methods, the regret guarantees become $\mathcal{O}(\sqrt{T})$~\cite{anava2013online,kozdoba2019line}, however, first-order methods are also known to adapt to changes~\cite{zinkevich2003online,hazan2009efficient}. On the contrary, the standard versions of second-order methods like RLS or ONS are not adequate for time-varying settings, where adaptation is needed. Another challenge is that in the case of multi-step ahead prediction, these methods should be adapted to deal with delayed learning feedback.

To solve the former issue, one way is to employ multiple learners along with a meta-algorithm that chooses the best expert online, also known as the Follow the Leading History (FLH) algorithm~\cite{hazan2009efficient,baby2022optimal}. Another option is to employ exponential forgetting~\cite{yuan2020trading,ding2021discounted}. In this paper, we follow the second approach since it has been standard practice in the adaptive control literature~\cite{guo1995performance,aastrom1977theory}. Moreover, it has a simple implementation for real-time applications, like drone control.

To solve the latter issue, we employ the standard technique of~\citet{joulani2013online} that decomposes the problem of learning with delays of duration $k$ into $k$ non-overlapping and non-delayed online learning instances.

\paragraph{Control Literature.}
Common feedback control policies, such as the proportional-integral-derivative (PID) controllers find extensive use for tracking problems in practice~\cite{pan2018efficient,cervantes2001pid}. For example, integral control has been used to track constant unknown offsets.
However, they typically require fine-tuning and their tracking performance for unknown time-varying trajectories is not guaranteed in general. 
Tracking of known reference trajectories with uncertain repetitive dynamics is studied in the context of  iterative learning control~\cite{bristow2006survey,owens2005iterative} with online extensions studying dynamic regret~\cite{balta2022regret}. However, the methods do not extend to non-repetitive setting of LQT with unknown time-varying references.

Predictive controllers such as model predictive control (MPC) have also been used to track constant references, see, for example, offset-free MPC~\cite{pannocchia2007combined,maeder2009linear}.
Various tracking MPC methods are proposed for time-varying references, e.g.,~\cite{kohler2020nonlinear}.
However, unless we have an accurate prediction model for the reference trajectory, the guarantees do not extend to more general sequentially revealed unknown trajectories. 
Dual MPC methods with active exploration are proposed for tracking applications and uncertain dynamics
\cite{heirung2017dual,soloperto2019dual,parsi2022explicit}. However, transient performance, regret guarantees, or the case of unknown time-varying references are not studied.

The LQT problem with unknown targets has been studied before from the point of view of control theory and adaptive control~\cite{vamvoudakis2016optimal,modares2014linear,peterson1982bounded}. Typically, the goal is to prove asymptotic convergence or boundedness of tracking error by appealing to Lyapunov stability theory. Most results assume time-invariant target dynamics. Tracking time-varying targets is possible using adaptive control techniques, under the assumption that the parameters of the target evolve like a random walk~\cite{guo1995performance}. However, this assumption excludes adversarial targets. By using regret as a metric, we can obtain non-asymptotic guarantees that also capture adversarial, non-stochastic behaviors.

Finally, an alternative approach is to employ robust control techniques to account for worst-case disturbances/tracking errors. Notable approaches include the celebrated $\mathcal{H}_{\infty}$ control approach or mixed $\mathcal{H}_2-\mathcal{H}_{\infty}$ control~\cite{zhang2021policy}. Recently, robust controllers inspired by the notion of regret were designed~\cite{goel2023regret}.

\section{Linear Quadratic Control Properties}\label{app_sec:LQT}
In this section, we revisit some properties of the LQT controller, like (internal) stability. We also discuss how to relax Assumption~\ref{ass:LQT_stability}. Finally, we discuss that LQT is equivalent to LQR control with disturbances.

\subsection{LQT properties}
Let us first recall the notions of stabilizability and detectability. A pair of system and input matrices $(A,B)$ is \emph{stabilizable} if and only if there exists a linear feedback gain $K_0\in\mathbb{R}^{m\times n}$ such that $(A-BK_0)$ has all eigenvalues strictly inside the unit circle. 

Let us now introduce the following proposition, which shows that the closed loop matrix $A-BK$ under the feedback gain $K$ defined in~\eqref{eq:K*} is stable, that is, all of its eigenvalues are inside the unit circle. It also shows that the feedforward gains $K_t, \forall t\le T$ defined in \eqref{eq:Kd} decay exponentially fast as $t$ increases.  Under the optimal control law~\eqref{eq:optimal action}, we do not have stability in the classical sense; since the target state can be arbitrary the tracking error might not remain close to the origin. However, we have internal stability since all signals remain bounded as long as $r_t$ is bounded. A more general case with time-variant costs of the proposition is proven in Corollary 1 of~\cite{zhang2021regret}.
\begin{prop}[Stability~\cite{zhang2021regret}]
\label{prop:exp-stable}
Let Assumption~\ref{ass:LQT_stability} be in effect. Recall the definition of the feedback and feedforward gains in~\eqref{eq:K*},~\eqref{eq:Kd}. For all $t\in \mathbb{N}$, the closed loop matrix $(A-BK)$ satisfies
\[
\snorm{(A-BK)^t}\le \sqrt{\frac{\lambda_{max}(X)}{\lambda_{min}(X)}}\rho^t,\,\rho = \sqrt{1-\frac{\lambda_{min}(Q)}{\lambda_{max}(X)}}\in (0, 1)
\]
where $\lambda_{min}(\cdot)$ and $\lambda_{\max}(\cdot)$ denote the minimum and maximum eigenvalue respectively, $X$ is defined in \eqref{eq:P*}.
As a result, the coefficient matrices $K_t$ defined in~\eqref{eq:Kd} satisfy
\begin{equation*}
    \lVert K_{t}\rVert \leq c_0\rho^{t} 
\end{equation*}
for 
$c_0 = \frac{\lVert B\rVert\lambda_{max}(X)}{\lambda_{min}(R)}\sqrt{\frac{\lambda_{max}(X)}{\lambda_{min}(X)}}$.
\end{prop}

Assumption~\ref{ass:LQT_stability} enables us to quantify the constants $c_0,\rho$. In fact, we can relax Assumption~\ref{ass:LQT_stability} and replace it with the following assumption. 
\begin{assumptionp}{\ref*{ass:LQT_stability}$'$}
    The pair $(A,B)$ is stabilizable, 
    the pair $(Q,A)$ is detectable, 
    $Q$ is positive semi-definite, and $R$ is symmetric and positive definite.
\end{assumptionp}
Note that a pair $(Q,A)$ is detectable if and only if $(A^\top,Q^\top)$ is stabilizable.
Under the above assumption, we still retain stability of $(A-BK)$~\cite{anderson2005optimal}. The only difference is that we do not have an explicit characterization of $\rho$ anymore as in Proposition~\ref{prop:exp-stable}.

\subsection{LQT as LQR with disturbances}
We can recast the LQT problem into an LQR problem with adversarial disturbances, see, for example,~\cite{karapetyan2023online}. By redefining the tracking error as the system state $e_t = x_t - r_t$ and encoding the time-varying target states as the disturbance $w_t = Ar_t - r_{t+1}$, the resulting LQR problem is given as,
\begin{equation}
	\label{eq:LQR obj}
	\begin{aligned}
    \min_{u_{0:T-1}} & \sum_{t=0}^{T-1} \left(\lVert e_t\rVert_Q^2 + \lVert u_t\rVert_R^2 \right) + \lVert e_T\rVert_X^2 \\
    \text{s.t.} & \quad e_{t+1}\! =\! Ae_t+Bu_t + w_t,~ \forall t=0,...,T-1.
\end{aligned}
\end{equation}
Therefore, we can treat the online LQT problem as an online, non-stochastic LQR problem with adversarial disturbances. The converse is also true if we set $r_{t}=w_{t-1}-Ar_{t-1}$ and initialize $r_{0}=0$.

\clearpage
\section{Autoregressive Systems}\label{app_sec:AR}
In this section, we present certain properties of autoregressive systems
\[
r_{t+1}=S_{t+1}z_t,
\]
where $z_{t}=\matr{r^\top_{t}&r^\top_{t-1}&\cdots&r^\top_{t-p}}^\top$, for some past horizon $p>0$. Note that autoregressive systems can also be described by state-space equations, viewing $z_t$ as a non-minimal state representation.  
 Define the extended matrices
\begin{equation}\label{eq:exosystem_state_matrix}
\mathcal{A}_{t+1}\triangleq\matr{& &S_{t+1}& & \\\cmidrule(lr){1-5} \mathbf{I}_n&0&\cdots&0&0\\0&\mathbf{I}_n&\cdots&0&0\\ \vdots &&\ddots& &\vdots\\0&0&\cdots &\mathbf{I}_n&0 },\,\mathcal{B}\triangleq\matr{\mathbf{I}_n\\\cmidrule(lr){1-1}  0\\0\\\vdots\\0 }.
\end{equation}
Then $r_t$ can be though of as the output of the following state-space system
\begin{equation}\label{app_eq:state_space_exosystem}
\begin{aligned}
z_{t+1}&=\mathcal{A}_{t+1}z_t\\
r_{t}&=\mathcal{B}^\top z_t
\end{aligned}
\end{equation}
Our results extend directly to auto-regressive dynamics with additional exogenous variables (ARX models)
\[
r_{t+1}=S_tz_t+v_t.
\]
It is sufficient to extend the regressor vector to contain $1$ as the last element
\[
\tilde{z}_t=\matr{z^\top_t&1}^\top.
\]
We can then extend $S_t$ accordingly
\[
\tilde{S}_t=\matr{S_t&v_t}.
\]
\subsection{Representation and learning complexity}\label{app_sec:trajectory_examples}
Auto-regressive dynamics can cover many commonly encountered trajectories.
Two elementary examples include constant velocity and circular targets.
\begin{exmp}[Constant Velocity Target]
\label{exmp:const v}
Let $z_t,y_t$ denote positions in a 2D horizontal plane with $\dot{z}_t,\dot{y}_t$ the respective velocities and the target state as $r_t=[z_t, y_t, \dot{z}_t, \dot{y}_t]^T$. Let $T_s$ be the sampling time for discretizing the target dynamics.  Then, a target with constant velocity can be represented by
\begin{equation*}
   r_{t+1}=Sr_t,\quad S= 
    \begin{bmatrix}
        1 & 0 & T_s & 0\\
        0 & 1 & 0 & T_s\\
        0 & 0 & 1 & 0\\
        0 & 0 & 0 & 1
    \end{bmatrix}.
\end{equation*}
Note that there might be multiple representations. We could also use the second order representation
\begin{equation*}
   r_{t+1}=S^{[1]}r_t+S^{[2]}r_{t-1},\quad S^{[1]}= 
    \begin{bmatrix}
        2 & 0 & 0 & 0\\
        0 & 2 & 0 & 0\\
        0 & 0 & 1 & 0\\
        0 & 0 & 0 & 1
    \end{bmatrix},\,S^{[2]}=\begin{bmatrix}
        -1 & 0 & 0 & 0\\
        0 & -1 & 0 & 0\\
        0 & 0 & 0 & 0\\
        0 & 0 & 0 & 0
    \end{bmatrix},
\end{equation*}
using the fact that $y_{t}-y_{t-1}=y_{t+1}-y_t$ (and similarly for $z_t$) under constant velocity.
\end{exmp}
\begin{exmp}[Circular Target with Constant Speed]
\label{exmp:circular}
Let $z_t,y_t$ denote positions in a 2D horizontal plane with $\dot{z}_t,\dot{y}_t$ the respective velocities and the target state as $r_t=[z_t, y_t, \dot{z}_t, \dot{y}_t]^T$. Let $T_s$ be the sampling time for discretizing the target dynamics. Then $\forall k\in \mathbb{N}^+, t=1,...,T-k$, the circular target with constant speed can be represented by
\begin{equation*}
   r_{t+1}=Sr_t,\quad S= 
    \begin{bmatrix}
        1 & 0 & T_s & 0\\
        0 & 1 & 0 & T_s\\
        0 & 0 & \cos{\theta} & -\sin{\theta}\\
        0 & 0 & \sin{\theta} & \cos{\theta}
    \end{bmatrix},
\end{equation*}
where we used Euler discretization. 
\end{exmp}
By allowing time-varying dynamics, we can also capture switching patterns, e.g. waypoint tracking, switching orientation, etc.
In the case of ARX models $r_{t+1}=S_tz_t+v_t$, we can capture targets that behave like control systems themselves, e.g., this representation could be used for controlling a drone in order to track another drone.

In many previous works~\cite{karapetyan2023online,li2019online,nonhoff2023online}, learning complexity is captured by the total variation (path length) of the target itself
\[
L_T=\sum_{t=0}^{T-1}\snorm{r_{t+1}-r_t}.
\]
In contrast, here, we capture learning complexity by the total variation of the target dynamics
\[
V_T=\sum_{t=0}^{T-1}\snorm{S_{t+1}-S_t}_F.
\]
In the case of ARX models we just replace $S_t$ with $\tilde{S}_t$. As shown in~\citep[Th. 3]{li2019online}, in the case of unstructured targets, e.g. randomly generated $r_t$, the former notion of complexity is optimal. However, we argue that in the case of dynamic structure, the former notion of complexity might be suboptimal. Consider the circular target or the linear target example. The total variation of the target state is linear with $T$ since $r_t$ is constantly changing: $L_T=\Omega(T)$. However, the total variation of the target dynamics is zero $V_T=0$. Using the PLOT algorithm would give us logarithmic regret in this case, while previous methods would give us linear regret.

We stress that every target can be represented trivially by ARX models. We can just set $S_t=0$, $v_t=r_t$. In this case, the total variation of $\tilde{S}_t=\matr{S_t&v_t}$ is equal to $L_T$, as in prior work.  However, the representation is non-unique in general. In many cases, the underlying dynamic structure will imply that a lower complexity representation exists, e.g., see Examples~\ref{exmp:const v},~\ref{exmp:circular}. \rev{If the targets have no dynamic structure, then, our current setting still captures this case; the regret of PLOT will just be larger, i.e., of the same order as $L_T$. }
Hence, our dynamic regret bounds supersede the ones in prior work. 

\subsection{Multi-step ahead dynamics expressions}\label{app_sec:multi_step_dynamics}
Given the state representation~\eqref{app_eq:state_space_exosystem}, we can now represent the multi-step ahead recursions in a compact way. Let 
\[\Phi_{t+i|t}\triangleq\mathcal{A}_{t+i}\cdots\mathcal{A}_{t+1},\,\Phi_{t|t}=\mathbf{I}_{np}\] 
be the transition matrix of system~\eqref{app_eq:state_space_exosystem}.
Then, we obtain $k-$step ahead recursions of the form
\[
r_{t+k}=S_{t+k|t}z_{t},
\]
where the matrices $S_{t+k|t}$ are given by
\begin{equation}
\begin{aligned}
    S_{t+k|t}&=\mathcal{B}^\top \Phi_{t+i|t}
    \end{aligned}
\end{equation}
By definition, if $S_{0,\dots,T}$ satisfies Assumption~\ref{assumption:stability}, then $\snorm{\mathcal{B}^\top \Phi_{t+i|t}}\le M$.
\subsection{Perturbation Analysis}
Let $S_{0},\dots,S_T$ be any sequence that satisfies Assumption~\ref{assumption:stability}. We will now show that we can upper-bound the total variation norm of the $k$-step ahead matrices $V^{k}_T\triangleq \sum_{t=k}^{T-k}\snorm{S_{t+k|t}-S_{t|t-k}}$ in terms of the total variation norm of the one-step ahead matrices $V_T\triangleq \sum_{t=1}^{T}\lVert S_t - S_{t-1}\rVert_F$. 
\begin{lem} \label{app_lem:perturbation}Recall that $M$ is the upper bound on matrices $S_{t+k|t}$ and $p$ is the past horizon (memory) of the auto-regressive dynamics. Let $\tilde{M}=\max\{M,1\}$. The following inequality is true
\[
V^k_T\le \sqrt{p} k^2 \tilde{M}^2 V_T.
\]
\end{lem}
\begin{proof}
By adding and subtracting terms $S_{t+j|t+j-k}$, $j=1,\dots,k-1$ and the triangle inequality, we obtain\[\snorm{S_{t+k|t}-S_{t|t-k}}\le \sum_{j=t}^{t+k-1}\snorm{S_{j+1|j+1-k}-S_{j|j-k}}.\] Hence, we get
\begin{equation}\label{app_eq:helper_perturbation}
V^{k}_T\le k \sum_{j=k-1}^{T-1}\snorm{S_{j+1|j+1-k}-S_{j|j-k}}.
\end{equation}
Let us now analyze
\[
\Delta_{k,t}\triangleq S_{t+k+1|t+1}-S_{t+k|t}=\mathcal{B}^\top (\Phi_{t+k+1|t+1}-\Phi_{t+k|t}).
\]
Adding and subtracting $\Phi_{t+k|t+j-1}\Phi_{t+j|t+1}$, for $j=k,k-1,\dots,2$ we obtain
\begin{align*}
\Delta_{k,t}&=\mathcal{B}^\top (\mathcal{A}_{t+k+1}-\mathcal{A}_{t+k})\Phi_{t+k|t+1}+\mathcal{B}^\top\Phi_{t+k|t+k-1}(\mathcal{A}_{t+k}-\mathcal{A}_{t+k-1})\Phi_{t+k-1|t+1}+\dots\\
&+\mathcal{B}^\top\Phi_{t+k|t+1}(\mathcal{A}_{t+2}-\mathcal{A}_{t+1}).
\end{align*}
By the triangle inequality
\begin{align*}
\snorm{\Delta_{k,t}}&\le \snorm{\mathcal{B}^\top} \snorm{\mathcal{A}_{t+k+1}-\mathcal{A}_{t+k}}\snorm{\Phi_{t+k|t+1}}+\snorm{\mathcal{B}^\top\Phi_{t+k|t+k-1}}\snorm{\mathcal{A}_{t+k}-\mathcal{A}_{t+k-1}}\snorm{\Phi_{t+k-1|t+1}}+\dots\\
&+\snorm{\mathcal{B}^\top\Phi_{t+k|t+1}}\snorm{\mathcal{A}_{t+2}-\mathcal{A}_{t+1}}.
\end{align*}
By definition, all terms $\mathcal{B}^\top \Phi_{t+k|t+j-1}$, $j=k,k-1,\dots,2$ are bounded by $\max\{M,1\}$.
Meanwhile, the error terms
\[
\mathcal{A}_{t+j}-\mathcal{A}_{t+j-1}=\matr{S_{t+j}-S_{t+j-1}\\0\\\vdots\\0}
\]
can be bounded by $\snorm{S_{t+j}-S_{t+j-1}}$, for $j=k,\dots,2$. Finally, we need to bound the norm of the transition matrices. Observe that by definition
\[
\Phi_{t+j|t}=\matr{S_{t+j|t}\\S_{t+j-1|t}\\ \vdots\\ S_{t+j-p+1|t}},\,\text{ if }j\ge p,\, \Phi_{t+j|t}=\matr{S_{t+j|t}\\ S_{t+j-1|t}\\  \vdots \\  S_{t+1|t}\\ \cmidrule(lr){1-1}\begin{array}{cc}\mathbf{I}_{n(p-j)}&\mathbf{0}_{n(p-j)\times nj}\end{array}},\,\text{ if }j< p,
\]
where $\mathbf{0}_{q_1\times q_2}$ denotes the zero matrix of dimensions $q_1\times q_2$.
As a result $\snorm{\Phi_{t+j|t}}\le \sqrt{p}\tilde{M}$
Putting everything together, we obtain
\[
\snorm{\Delta_{k,t}}\le \sqrt{p}\tilde{M}^2\sum_{j=2}^{k}\snorm{S_{t+j}-S_{t+j-1}}.
\]
The results follow from the above inequality and~\eqref{app_eq:helper_perturbation}.
\end{proof}

\section{Proofs for Prediction}\label{app_sec:prediction_proof}
In this section, we prove Theorem~\ref{thm:ar_prediction_regret}. We provide regret upper bounds for prediction in terms of the total variation norm of the dynamics $V_T$. Note that the prediction guarantees hold for any sequence $r_{0:T}$ that satisfies Assumptions~\ref{assumption:boundedness},~\ref{assumption:stability}. First, we show the result for one-step ahead prediction $k=1$. Then, we generalize to $k>0$.
\subsection{Regret guarantees for one-step ahead prediction}
Consider the one-step ahead prediction problem. 
For simplicity, define $\hat{S}_{t|t-1}=\hat{S}_{t}$ and $P_{t|t-1}=P_{t}$. Then Algorithm~\ref{alg:RLS} is equivalent to Algorithm~\ref{alg:Matrix_RLS}.
\begin{algorithm}[ht!]
\caption{Recursive Least Squares, for $k=1$}\label{alg:Matrix_RLS}
\begin{algorithmic}
\REQUIRE{forgetting factor $\gamma \in (0, 1)$, initial regularization $\varepsilon$}
\STATE Initialize: $P_{-1}=\varepsilon I_{np}$, $\hat{S}_{0} \in \mathcal{S}$;
\FOR{$t = 0, \dots, T$}
    \STATE Predict $\hat{r}_{t}=\hat{S}_{t} z_{t-1}$ and incur loss $\snorm{r_{t}-\hat{S}_{t} z_{t-1}}^2$
    \STATE Update $P_{t} = \gamma P_{t-1} + z_{t-1}z_{t-1}^\top$;
    \STATE Update $\hat{S}_{t+1} = \Pi_{\mathcal{S}}^{P_t}(\hat{S}_t +(r_{t}-\hat{S}_tz_{t-1})z^\top_{t-1}P^{-1}_t)$,\\
    where $ \Pi^{P_t}_\mathcal{S}(Y)\triangleq \arg\min_{S\in\mathcal{S}} \snorm{S-Y}_{F,P_t}$;
\ENDFOR
\end{algorithmic}
\end{algorithm}
Let $r_t$, for $t\le T$ be any arbitrary target sequence satisfying Assumption~\ref{assumption:boundedness}. 
The target sequence may not necessarily satisfy~\eqref{eq:exosystem}. 
Let $S_{0},\dots,S_{T}$ be any ``comparator" sequence that satisfies Assumption~\ref{assumption:stability}. 
 Let $\hat{S}_{t}$ be the sequence generated by the RLS algorithm. Then the dynamic regret of the one-step ahead predictor versus the sequence $S_0,\dots,S_T$ is defined as
 \begin{equation}\label{app_eq:prediction_regret_k_1}
 \mathcal{R}^{(1)}_{\mathrm{pred}}(S_{0:T})\triangleq\frac{1}{2}\sum_{t=0}^T\snorm{\hat{S}_tz_{t-1}-r_t}^2-\frac{1}{2}\sum_{t=0}^T\snorm{S_tz_{t-1}-r_t}^2. 
 \end{equation}
 Note that if the sequence $r_t$ satisfies the autoregressive dynamics~\eqref{eq:exosystem} and we choose the comparator sequence to coincide with the true dynamics then the regret reduces to 
  \begin{equation*}
\mathcal{R}^{(1)}_{\mathrm{pred}}(S_{0:T})\triangleq\frac{1}{2}\sum_{t=0}^T\snorm{\hat{S}_tz_{t-1}-r_t}^2. 
 \end{equation*}
 Note that we define the projection operator $\Pi^{P_t}_\mathcal{S}(Y)$ as
 \begin{align}
    \Pi^{P_t}_\mathcal{S}(Y)\triangleq \arg\min_{S\in\mathcal{S}} \snorm{S-Y}_{F,P_t},
 \end{align}
 where $P_t$ should be symmetric positive definite. Recall that the weighted Frobenius norm is given by
 \[
\snorm{S}^2_{F,P_t}=\tr(SP_tS^\top).
 \]
To bound the regret of RLS, we adapt the proof of~\cite{yuan2020trading} while keeping track of all quantities of interest, e.g. the system dimension, logarithmic terms etc., and while working with matrices instead of vectors.

\begin{thm}\label{app_thm:RLS_k_1}
   Let $r_t,$ for $t\le T$ be any target trajectory that satisfies Assumption~\ref{assumption:boundedness}. Let $S_{0},\dots,S_{T}$ be any sequence that satisfies Assumption~\ref{assumption:stability}. Define the path length as  $V_T=\sum_{t=0}^{T-1}\snorm{S_{t+1}-S_t}_F$. The regret of the RLS algorithm (Algorithm~\ref{alg:Matrix_RLS}) for the one-step ahead prediction problem is upper bounded by
    \begin{align}\label{app_eq:RLS_base_regret_bound}
        \mathcal{R}^{(1)}_{\mathrm{pred}}(S_{0:T})&
        \le \frac{\beta_1}{1-\gamma}V_T+\beta_2 T\log \frac{1}{\gamma}+\beta_3\log \frac{1}{1-\gamma}+\beta_4
    \end{align}
    with the constants 
\begin{equation}\label{app_eq:betas}
\begin{aligned}
    \beta_1=2\sqrt{n}M\paren{\frac{\varepsilon+pD^2_r}{1-\gamma}},\,\beta_2=np (1+\sqrt{p}M)^2D^2_r,\,\\
    \beta_3=\beta_2,\,\beta_4=2\gamma\varepsilon n M^2+\beta_2\log\frac{\varepsilon+p D^2_r}{\varepsilon}.
    \end{aligned}
\end{equation}
\end{thm}
\begin{proof}
Using the law of cosines identity $\snorm{b-a}^2-\snorm{b-c}^2=-\snorm{a-c}^2-2(a-c)^\top(b-a)$, we obtain
\begin{align*}
\mathcal{R}_{\mathrm{pred}}(S_{0:T})&= \sum_{t=0}^T-\frac{1}{2}z_{t-1}^\top(\hat{S}_t-S_t)^\top(\hat{S}_t-S_t)z_{t-1}
    -(r_t-\hat{S}_tz_{t-1})^\top(\hat{S}_t-S_t)z_{t-1}\\
&=\sum_{t=0}^T -\frac{1}{2}\tr((\hat{S}_t-S_t)z_{t-1}z_{t-1}^\top(\hat{S}_t-S_t)^\top)
    -\tr((\hat{S}_t-S_t)z_{t-1}(r_t-\hat{S}_tz_{t-1})^\top),
\end{align*}
where we used the identity $a^\top b=\tr(a^\top b)=\tr(b a^\top)$.
Invoking Lemma~\ref{app_lem:gradient_upper_bound} and Lemma~\ref{app_lem:telescope}, we have
\begin{align*}
\mathcal{R}_{\mathrm{pred}}(S_{0:T})&\le 2\sqrt{p}M(\varepsilon+\frac{pD^2_r}{1-\gamma})V_T+\frac{1}{2}\max_{t\le T}\snorm{r_{t}-\hat{S}_t z_{t-1}}^2\sum_{t=0}^{T}z_{t-1}^\top P^{-1}_tz_{t-1}+2\gamma\varepsilon n M^2.
\end{align*}
The result follows from Lemma~\ref{app_lem:elliptical_potential} and the fact that $\snorm{r_t}\le D_r$, $\snorm{\hat{S}_{t}z_{t-1}}\le M\snorm{z_{t-1}}\le M\sqrt{p}D_r$ since $z_t$ a concatenation of $p$ vectors. As a result, we can bound the error $\snorm{r_{t}-\hat{S}_t z_{t-1}}^2$ by $(1+\sqrt{p}M)^2D^2_r$ for all $t\le T$. To simplify the final bound we use $T+1\le 2T$.
\end{proof}
The following technical lemmas are auxiliary results towards proving Theorem~\ref{app_thm:RLS_k_1}. Lemmas~\ref{app_lem:gradient_upper_bound},~\ref{app_lem:telescope} control the growth of the first order (gradient) term in the regret. Lemma~\ref{app_lem:information_matrix_bound} upper bounds matrix $P_t$ while Lemma~\ref{app_lem:elliptical_potential} contains a standard elliptical potential bound, tailored to the forgetting factor case.
\begin{lem}[Gradient Inner Product]\label{app_lem:gradient_upper_bound}
Consider the conditions of Theorem~\ref{app_thm:RLS_k_1}. Let $\hat{S}_t$ be the RLS estimate at time $t$ and $S_t,S_{t+1}$ be any arbitrary matrices that satisfy the constraints, i.e., $S_t,\,S_{t+1}\in\mathcal{S}$. We have
  \begin{align*}
&-2\tr(z_{t-1}(r_{t}-\hat{S}_t z_{t-1})^\top (\hat{S}_t-S_t))\\
&\le \snorm{\hat{S}_{t}-S_{t}}^2_{F,P_t}-\snorm{\hat{S}_{t+1}-S_{t+1}}^2_{F,P_t}+\snorm{r_{t}-\hat{S}_t z_{t-1}}^2 z^\top_{t-1} P^{-1}_t z_{t-1}+4\sqrt{p}M(\varepsilon+\frac{p D^2_r}{1-\gamma})\snorm{S_{t+1}-S_{t}}_{F}
 \end{align*}
\end{lem}
\begin{proof}
By the non-expansiveness of the projection operator, we have
 \begin{align*}
\snorm{\hat{S}_{t+1}-S_{t}}^2_{F,P_t}&\le \snorm{\hat{S}_{t}+(r_{t}-\hat{S}_t z_{t-1})z^\top_{t-1} P^{-1}_t-S_{t}}^2_{F,P_t}\\
&=\snorm{\hat{S}_{t}-S_{t}}^2_{F,P_t}+\snorm{(r_{t}-\hat{S}_t z_{t-1})z^\top_{t-1} P^{-1}_t}^2_{F,P_t}+2\tr(z_{t-1}(r_{t}-\hat{S}_t z_{t-1})^\top (\hat{S}_t-S_t))\\
&=\snorm{\hat{S}_{t}-S_{t}}^2_{F,P_t}+\tr\paren{(r_{t}-\hat{S}_t z_{t-1})z^\top_{t-1} P^{-1}_tz_{t-1}(r_{t}-\hat{S}_t z_{t-1})^\top}\\
&\quad+2\tr(z_{t-1}(r_{t}-\hat{S}_t z_{t-1})^\top (\hat{S}_t-S_t))\\
&=\snorm{\hat{S}_{t}-S_{t}}^2_{F,P_t}+\snorm{r_{t}-\hat{S}_t z_{t-1}}^2 z^\top_{t-1} P^{-1}_t z_{t-1}+2\tr(z_{t-1}(r_{t}-\hat{S}_t z_{t-1})^\top (\hat{S}_t-S_t))
 \end{align*}
Meanwhile, adding and subtracting $S_{t+1}$ in the norm in the left-hand side of the above inequality we obtain
  \begin{align*}
\snorm{\hat{S}_{t+1}-S_{t}}^2_{F,P_t}&=\snorm{\hat{S}_{t+1}-S_{t+1}}^2_{F,P_t}+\snorm{S_{t+1}-S_{t}}^2_{F,P_t}+2\tr((\hat{S}_{t+1}-S_{t+1})P_t(S_{t+1}-S_{t})^\top)\\
&\ge \snorm{\hat{S}_{t+1}-S_{t+1}}^2_{F,P_t}+\snorm{S_{t+1}-S_{t}}^2_{F,P_t}-2\snorm{\hat{S}_{t+1}-S_{t+1}}_{F,P_t}\snorm{S_{t+1}-S_{t}}_{F,P_t}\\
&\ge \snorm{\hat{S}_{t+1}-S_{t+1}}^2_{F,P_t}-2\snorm{\hat{S}_{t+1}-S_{t+1}}_{F,P_t}\snorm{S_{t+1}-S_{t}}_{F,P_t}\\
&\ge \snorm{\hat{S}_{t+1}-S_{t+1}}^2_{F,P_t}-4\sqrt{n}M(\varepsilon +\frac{p D^2_r}{1-\gamma})\snorm{S_{t+1}-S_{t}}_{F},
 \end{align*}
where we used $\snorm{\hat{S}_{t+1}-S_{t+1}}_F\le\sqrt{\mathrm{rank}(\hat{S}_{t+1}-S_{t+1})}\snorm{\hat{S}_{t+1}-S_{t+1}}$, $\mathrm{rank}(\hat{S}_{t+1}-S_{t+1})\le n$, $\snorm{\hat{S}_{t+1}-S_{t+1}}\le 2M $
and Lemma~\ref{app_lem:information_matrix_bound}. Combining the above two inequalities gives us the result.
\end{proof}

\begin{lem}[Telescoping Series]\label{app_lem:telescope} Consider the conditions of Theorem~\ref{app_thm:RLS_k_1}.
\[
\sum_{t=0}^{T} \snorm{\hat{S}_{t}-S_{t}}^2_{F,P_t}-\snorm{\hat{S}_{t+1}-S_{t+1}}^2_{F,P_t}\le \sum_{t=0}^{T}\tr((\hat{S}_{t}-S_{t})z_{t-1}z_{t-1}^\top (\hat{S}_{t}-S_{t})^\top)+\gamma\varepsilon\,\tr((\hat{S}_{0}-S_{0}) (\hat{S}_{0}-S_{0})^\top).
\]
\begin{proof}
Notice that
\begin{align*}
\snorm{\hat{S}_{t+1}-S_{t+1}}^2_{F,P_{t+1}}-\snorm{\hat{S}_{t+1}-S_{t+1}}^2_{F,P_t}&=\tr((\hat{S}_{t+1}-S_{t+1})(P_{t+1}-P_t) (\hat{S}_{t+1}-S_{t+1})^\top)\\
&\le \tr((\hat{S}_{t+1}-S_{t+1})z_{t}z_{t}^\top (\hat{S}_{t+1}-S_{t+1})^\top),
\end{align*}
since $P_{t+1}=\gamma P_{t}+z_{t}z_{t}^\top$, $P_t\succ 0$, and $\gamma\le 1$.
For $t=0$, we have
\[
\snorm{\hat{S}_{0}-S_{0}}^2_{F,P_{0}}=\tr((\hat{S}_{0}-S_{0})(z_{-1}z_{-1}^\top+\gamma \varepsilon I) (\hat{S}_{0}-S_{0})^\top)
\]
The result follows by dropping the last negative term $\snorm{\hat{S}_{T+1}-S_{T+1}}^2_{F,P_T}$.
\end{proof}

\begin{lem}[Design matrix bound]\label{app_lem:information_matrix_bound}
Consider the conditions of Theorem~\ref{app_thm:RLS_k_1} with $P_t=\gamma P_{t-1}+z_{t-1}z_{t-1}^\top$ and $P_{-1}=\varepsilon I$. We have \[\snorm{P_t}\le \varepsilon +\frac{pD^2_r}{1-\gamma}.\]
\end{lem}
\begin{proof}
Since $\snorm{z_t}\le \sqrt{p}\max_{t-p\le r\le t}\snorm{r_k}\le \sqrt{p}D_r$, we have recursively that
\[
\snorm{P_t}\le \gamma \snorm{P_{t-1}}+pD^2_r=\frac{1-\gamma^t}{1-\gamma}pD^2_r+\gamma^{t+1}\varepsilon\le \frac{1}{1-\gamma}pD^2_r+\varepsilon
\]
\end{proof}

\begin{lem}[Forgetting potential lemma]\label{app_lem:elliptical_potential}
Consider the conditions of Theorem~\ref{app_thm:RLS_k_1} with $P_t=\gamma P_{t-1}+z_{t-1}z_{t-1}^\top$ and $P_{-1}=\varepsilon I$. The following upper bound is true
\[\sum_{t=0}^{T}z_{t-1}^\top P^{-1}_tz_{t-1}\le np\log(\frac{\varepsilon+pD^2_r}{\varepsilon(1-\gamma)})+np(T+1)\log \frac{1}{\gamma}. \]
\end{lem}
\begin{proof}
Note that $\gamma P_{t-1}=P_t-z_{t-1}z^\top_{t-1}=P^{1/2}_t(\mathbf{I}_{np}-P^{-1/2}_t z_{t-1} z_{t-1}^\top P^{-1/2}_t)P^{1/2}_t$. Using the identity
$\det(I+DC)=\det(I+CD)$, we obtain
\[
\det(\gamma P_{t-1})=\det(P_t)(1-z_{t-1}^\top P^{-1}_tz_{t-1}),
\]
which, in turn, implies
\[
z_{t-1}^\top P^{-1}_tz_{t-1}=1-\frac{\det(\gamma P_{t-1})}{\det(P_t)}\le -\log\frac{\det(\gamma P_{t-1})}{\det(P_t)}.
\]
The inequality follows from the fact that $\gamma P_{t-1}\preceq P_t$ and the elementary inequality $x-1\ge \log x$, for $0<x\le 1$. By the properties of the determinant, we also have $\det(\gamma P_{t-1})=\gamma^{np}\det(P_{t-1})$. Hence we have
\[
z_{t-1}^\top P^{-1}_tz_{t-1}\le \log\det(P_{t})-\log\det(P_{t-1})+np\log 1/\gamma.
\]
Summing all inequalities and by telescoping
\[\sum_{t=0}^{T}z_t^\top P^{-1}_tz_t\le \log\det P_{T}-\log\det(\varepsilon I_{np})+np(T+1)\log 1/\gamma. \]
The final bound follows from the upper bound on $P_T$ given in Lemma~\ref{app_lem:information_matrix_bound}.
\end{proof}

\end{lem}
\subsection{Regret guarantees for multi-step ahead prediction}\label{app_sec:multi_step_pred_proof}
Consider now the case of $k-$steps ahead prediction. Recall that for any fixed $k=1,\dots,W$, the $k-$step ahead predictor maintains $k$ learners that are updated at non-overlapping intervals. An example for $k=1,2,3$ can be found in Figure~\ref{fig:schedule}.
\begin{figure}[b]
\centering
        \includegraphics[trim={1cm 1cm 1cm 2cm},clip,width=0.98\linewidth]{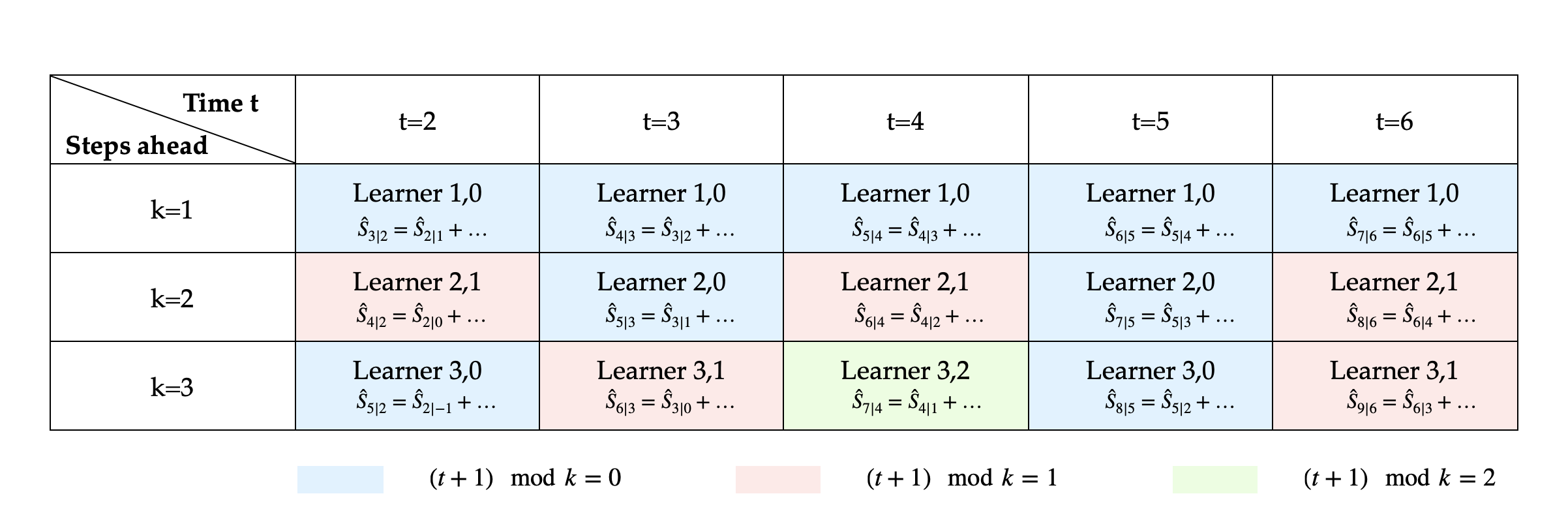}
        \caption{The updates of the $k=1,2,3$ steps ahead predictors for $t=2,\dots,6$. Projections are omitted for simplicity.
        Each predictor consists of $k$ independent learners that are updated at non-overlapping time steps. For example, for $k=2$, there are two learners and the update for $\hat{S}_{t+2|t}$ is decoupled from the one of $\hat{S}_{t+1|t-1}$. 
        }
    \label{fig:schedule}
\end{figure}

Let $S_{k-1|-1},\,S_{k|0},\dots,S_{T|T-k}$ be any comparator sequence that satisfies Assumption~\ref{assumption:stability}. In this case, the regret of $k-$step ahead prediction is defined as
\[
 \mathcal{R}^{(k)}_{\mathrm{pred}}\triangleq\sum_{t=k-1}^T\snorm{\hat{S}_{t|t-k}z_{t-k}-r_t}^2-\sum_{t=k-1}^T\snorm{S_{t|t-k}z_{t-k}-r_t}^2. 
\]
Once again if we choose the sequence to coincide with the true dynamics, we get that
\[
 \mathcal{R}^{(k)}_{\mathrm{pred}}\triangleq\sum_{t=k-1}^T\snorm{\hat{S}_{t|t-k}z_{t-k}-r_t}^2. 
\]
Recall that we have $k$ learners updated at non-overlapping time steps. Let $i$, for $i=0,\dots,k-1$, denote the index for the learner. We will decompose the problem into $k$ independent one-step ahead prediction problems and invoke Theorem~\ref{app_thm:RLS_k_1}.

Every learner $i$ is invoked to predict $r_{t+k|t}$ when $\left((t+1)\!\mod k\right)=i$. Learner $0$ is invoked to  predict $r_{k-1|-1}$, $r_{2k-1|k-1}$, $\dots$, etc. Similarly, learner $i$ is invoked to predict $r_{k+i-1|i-1}$, $r_{2k+i-1|k+i-1}$, $\dots$, etc. We will redefine the time axis to reduce the problem to one-step ahead prediction. To streamline the presentation define 
\begin{align*}
\hat{S}^{i}_{\tau+1}&\triangleq \hat{S}_{(\tau+1) k+i-1|\tau k+i-1}\\
S^{i}_{\tau+1}&\triangleq S_{(\tau+1) k+i-1|\tau k+i-1}\\
z^i_{\tau-1}&=z_{\tau k+i-1},\,r^i_{\tau}=r_{(\tau+1) k+i-1}
\end{align*}
and denote
\begin{align*}
N_{k,i}&=\lfloor \frac{T+1}{k}\rfloor,\text{ if }i\le (T+1)\mod k\\
&=\lfloor \frac{T+1}{k}\rfloor-1,\text{ otherwise.}
\end{align*}
Note that learner $i$ is invoked starting at $\tau=0$ up to $N_{i,k}$ times. 

Denote the regret of learner $i$ by
 \begin{equation}\label{app_eq:prediction_regret_multi_k}
 \mathcal{R}^{(k,i)}_{\mathrm{pred}}(S^i_{0:N_{k,i}})\triangleq\frac{1}{2}\sum_{\tau=0}^{N_{k,i}}\snorm{\hat{S}^i_{\tau+1}z^i_{\tau}-r^i_{\tau+1}}^2-\frac{1}{2}\sum_{t=0}^{N_{k,i}}\snorm{S^i_{\tau+1}z^i_{\tau}-r^i_{\tau+1}}^2. 
 \end{equation}
 We can now prove Theorem~\ref{thm:ar_prediction_regret}.

 \subsection{Proof of Theorem~\ref{thm:ar_prediction_regret}}
 Based on the above notation, the regret can be decomposed into $k$ terms
 \[
 \mathcal{R}^{(k)}_{\mathrm{pred}}=\sum_{i=0}^{k-1}\mathcal{R}^{(k,i)}_{\mathrm{pred}}(S^i_{0:N_{k,i}}).
 \]
 As a consequence of redefining the time axis, the prediction regret of every individual learner $i$ can be bounded using Theorem~\ref{app_thm:RLS_k_1}. Let $V^{k,i}_T=\sum_{\tau=0}^{N_{k,i}-1}\snorm{S^i_{\tau+1}-S^i_{\tau}}$. Then we obtain
 \begin{equation}\label{app_eq:one_learner_bound_k_steps_ahead}
\mathcal{R}^{(k,i)}_{\mathrm{pred}}(S^i_{0:N_{k,i}})\le \frac{\beta_1}{1-\gamma}V^{k,i}_{T}+\beta_2(\frac{T+1}{k})\log\frac{1}{\gamma}+\beta_3\log\frac{1}{1-\gamma}+\beta_4.
 \end{equation}
 Notice that $V^k_T=\sum_{i=0}^{k-1}V^{k,i}_T$. Hence, summing up we obtain
 \[
\mathcal{R}^{(k)}_{\mathrm{pred}}\le \frac{\beta_1}{1-\gamma}V^{k}_{T}+\beta_2(T+1)\log\frac{1}{\gamma}+k\beta_3\log\frac{1}{1-\gamma}+k\beta_4.\tag*{\qed}
 \]

\section{Regret of the PLOT algorithm}\label{app_sec:control_proof}
Let $\Sigma=B^\top XB+R$. Consider again the optimal non-causal policy
\[
u_t^*(x_t) = -K(x_t-r_t) - \underbrace{\sum_{i=t}^{T-1}K_{i-t}(Ar_{i}-r_{i+1})}_{q_t(r_{t:T})},
\]
where the optimal feedforward terms are given by
\[
q_t(r_{t:T})=\sum_{i=t}^{T-1}K_{i-t}(Ar_{i}-r_{i+1}).
\]
Consider also the causal suboptimal policy of the PLOT algorithm 
\[
    u_t^{\pi}(x_t) \!=\! -K (x_t-r_t)- \underbrace{\!\sum_{i=t}^{\min\{t+W-1,T\}}\!K_{i-t}(Ar_{i|t}-r_{i+1|t})}_{\hat{q}_t(r_{t:t+W|t})},
\]
where we define the truncated feedforward terms that use the reference predictions
\[
\hat{q}_t(r_{t:t+W|t})=\sum_{i=t}^{\min\{t+W-1,T\}}\!K_{i-t}(Ar_{i|t}-r_{i+1|t}).
\] %
Finally, define the residual feedforward terms $
\delta_{t}(r_{t+W:T})\triangleq q_t(r_{t:T})-\hat{q}_t(r_{t:t+W})$. With these definitions in hand, we can now analyze the regret of PLOT.

By invoking the performance difference lemma \rev{(Lemma~\ref{lemma:mpc regret})}, we obtain that the regret is equal to
\begin{align*}
   \mathcal{R}(\pi) &=\sum_{t=0}^{T-1}\lVert \hat{q}_t(r_{t:t+W|t}) - q_t(r_{t:T})\rVert^2_{\Sigma}\nonumber\\
   &=\sum_{t=0}^{T-1}\lVert  \hat{q}_t(r_{t:t+W|t}-r_{t:t+W})\rev{-}\delta_t(r_{t+W:T})\rVert^2_{\Sigma}.
   \end{align*}
Invoking Cauchy-Schwartz for the two summands, we can now decompose the regret  into two terms
\begin{equation}\label{app_eq:regret_decomposition}
    \mathcal{R}(\pi)\le \underbrace{2 \sum_{t=0}^{T-1}\snorm{\hat{q}_t(r_{t:t+W|t}-r_{t:t+W})}^2_{\Sigma}}_{\text{prediction error}}+\underbrace{2\sum_{t=0}^{T} \snorm{\delta_t(r_{t+W:T})}^2_{\Sigma}}_{\text{truncation error}},
\end{equation}
where the first term captures the effect of the prediction error, while the second term captures the effect of truncation.
To bound the latter we invoke the following lemma.
\begin{lem}[Truncation term]\label{app_lem:control_truncation}
The truncation term of the regret satisfies
\[
2\sum_{t=0}^{T} \snorm{\delta_t(r_{t+W:T})}^2_{\Sigma}=2\sum_{t=0}^{T-1-W}\snorm{\sum_{i=t+W}^{T-1}K_{i-t}(Ar_{i}-r_{i+1})}^2_{\Sigma}\le \alpha_1\frac{\rho^{2W}}{(1-\rho)^2}  T,
\]
where
\begin{equation}\label{eq:alpha_1}
    \alpha_1=2c^2_0(\snorm{A}+1)^2D^2_r\snorm{\Sigma}. 
\end{equation}
\end{lem}
\begin{proof}The first equality follows by the definition of $\delta_t(r_{t+W:T})$ and the fact that $\delta_t(r_{t+W:T})=0$, for $t+W>T-1$.
By Proposition~\ref{prop:exp-stable}, we have
$\snorm{K_t}\le c_0 \rho^t$. Hence 
\begin{align*}
&2\sum_{t=0}^{T-1-W}\snorm{\sum_{i=t+W}^{T-1}K_{i-t}(Ar_{i}-r_{i+1})}^2_{\Sigma}\le 2c^2_0 (\snorm{A}+1)^2\snorm{\Sigma} D^2_r\sum_{t=0}^{T-1-W}(\sum_{i=t+W}^{T-1}\rho^{i-t})^2\\
&=\alpha_1 \rho^{2W} \sum_{t=0}^{T-1-W}(\sum_{i=0}^{T-1-W-t}\rho^{i})^2\le \alpha_1 \rho^{2W} \sum_{t=0}^{T-1-W}\frac{(1-\rho^{T-W-t})^2}{(1-\rho)^2}\\
&\le \alpha_1\frac{\rho^{2W}}{(1-\rho)^2}T
\end{align*}
\end{proof}
What remains to show is that the prediction error term is upper bounded in terms of the prediction regret of the RLS algorithm.
\begin{lem}[Prediction term]\label{app_lem:control_prediction_term}
    The prediction term of the regret is upper bounded by
     \begin{align*}
     2 \sum_{t=0}^{T-1}\snorm{\hat{q}_t(r_{t:t+W|t}-r_{t:t+W})}^2_{\Sigma}\le c_1 \sum_{i=1}^{W}\rho^{i-1}\mathcal{R}^{(i)}_{pred},
     \end{align*}
     where
     \[
c_1=2c^2_0\snorm{\Sigma}\tilde{W},\,\tilde{W}=\min\{(1-\rho)^{-1},W \}.
     \]
\end{lem}
\begin{proof}
Let $t\le T-W+1$ for simplicity. The case $t>T-W+1$ is similar. Then,
\[
\hat{q}_t(r_{t:t+W|t})=\sum_{i=0}^{W-1}\!K_{i}(Ar_{t+i|t}-r_{t+i+1|t}).
\]
By regrouping the terms, and since $r_t=r_{t|t}$ we obtain
\[
\hat{q}_t(r_{t:t+W|t})=K_0Ar_{t}+\sum_{i=1}^{W-1} (K_{i}A-K_{i-1})r_{t+i|t}-K_{W-1}r_{t+W|t}.
\]
As a result
\[
\hat{q}_t(r_{t:t+W|t}-r_{t:t+W})=\sum_{i=1}^{W-1}L_{i}(r_{t+i|t}-r_{t+i})-K_{W-1}(r_{t+W|t}-r_{t+W})
\]
where $L_{i}\triangleq K_{i}A-K_{i-1}$. 
Using the properties of the LQT controller and~\eqref{eq:P*}
\begin{align*}
L_{i}=K_{i}A-K_{i-1}&=\Sigma^{-1}B^\top (A-BK)^{\top,i}XA-\Sigma^{-1}B^\top (A-BK)^{\top,i-1}X\\
&=\Sigma^{-1}B^\top (A-BK)^{\top,i-1} ((A-BK)^\top XA-X)\\
&=\Sigma^{-1}B^\top (A-BK)^{\top,i-1} (A^\top X A-K^\top B^\top XA-X)\\
&=\Sigma^{-1}B^\top (A-BK)^{\top,i-1} (-Q+K^\top B^\top X A-K^\top B^\top XA)\\
&=-\Sigma^{-1}B^\top (A-BK)^{\top,i-1}Q.
\end{align*}
Hence, by Proposition~\ref{prop:exp-stable} and the fact that $Q\preceq X$, we also get
\[
\snorm{L_{i}}\le c_0\rho^{i-1}.
\]
The difference between the truncated feedforward terms now becomes
\begin{align*}
\snorm{\hat{q}_t(r_{t:t+W|t}-r_{t:t+W})}^2_{\Sigma}&\le c^2_0\snorm{\Sigma}\paren{\sum_{i=1}^{W}\rho^{i-1}\snorm{r_{t+i}-r_{t+i|t}}}^2=c^2_0\snorm{\Sigma}\paren{\sum_{i=1}^{W}\rho^{(i-1)/2}\rho^{(i-1)/2}\snorm{r_{t+i}-r_{t+i|t}}}^2\\
&\stackrel{i)}{\le} c^2_0\snorm{\Sigma} \sum_{i=1}^{W}\rho^{i-1}\sum_{i=1}^{W}\rho^{i-1}\snorm{r_{t+i}-r_{t+i|t}}^2\le c^2_0\snorm{\Sigma}\tilde{W}\sum_{i=1}^{W}\rho^{i-1}\snorm{r_{t+i}-r_{t+i|t}}^2,
\end{align*}
where $i)$ follows by Cauchy-Schwartz. The result for $t>T-W+1$ is similar. To obtain the final bound we just need to sum over $t$.
\end{proof}
Before we proceed to the proof of Theorem~\ref{thm:main}, let us recall the following standard result.
\begin{lem}[Geometric Series]\label{app_lem:geometric_series}
Let $0\le \rho<1$. Then the following hold
\begin{align*}
S_1(\rho)\triangleq\sum_{i=1}^{N}i\rho^{i-1}&=\frac{-(N+1)\rho^{N}(1-\rho)+(1-\rho^{N+1})}{(1-\rho)^2}\le \max\set{\frac{1}{(1-\rho)^2},N^2}\\
S_2(\rho)\triangleq \sum_{i=1}^{N}i^2\rho^{i-1}&=\frac{1+\rho-(N+1)^2\rho^N+(2N^2+2N-1)\rho^{N+1}-N^2\rho^{N+2}}{(1-\rho)^3}\le \max\set{\frac{2}{(1-\rho)^3},N^3}
\end{align*}
\end{lem}
\begin{proof}
The bounds $\sum_{i=1}^{N}i^k\rho^{i-1}\le N^{k+1}$ are immediate since $0\le \rho<1$. The closed-form expressions of the sums are standard, but we repeat the proof here for completeness. For $S_1(\rho)$, the expression follows from
\begin{align*}
S_1(\rho)=\frac{d}{d\rho}(\sum_{i=0}^{N}\rho^{i})=\frac{d}{d\rho}\paren{\frac{1-\rho^{N+1}}{1-\rho}}.
\end{align*}
To show the upper bound notice that
\[
-(N+1)\rho^{N}(1-\rho)+(1-\rho^{N+1})=1-N\rho^N(1-\rho)-\rho^N\le 1.
\]
For $S_2$ we use the identity
\begin{align*}
S_2(\rho)=\rho \sum_{i=1}^{N}i(i-1)\rho^{i-2}+S_1(\rho)=\rho\frac{d}{d\rho}S_1(\rho)+S_1(\rho).
\end{align*}
To show the upper bound notice that
\begin{align*}
&1+\rho-(N+1)^2\rho^N+(2N^2+2N-1)\rho^{N+1}-N^2\rho^{N+2}\\
&=1+\rho-N^2(1-\rho)^2-2N\rho^N(1-\rho)-\rho^N-\rho^{N+1}\le 1+\rho\le 2.
\end{align*}
\end{proof}
\subsection{Proof of Theorem~\ref{thm:main}}
By~\eqref{app_eq:regret_decomposition}, Lemma~\ref{app_lem:control_truncation}, and Lemma~\ref{app_lem:control_prediction_term}, we obtain that
\[
\mathcal{R}(\pi)\le \alpha_1 \frac{\rho^{2W}}{(1-\rho)^2}T+c_1\sum_{i=1}^W\rho^{i-1}\mathcal{R}^{(i)}_{\mathrm{pred}}.
\]
Further, by Theorem~\ref{thm:ar_prediction_regret}
\begin{align*}
\mathcal{R}(\pi)&\le \alpha_1 \frac{\rho^{2W}}{(1-\rho)^2}T+c_1\beta_1\frac{1}{1-\gamma}\sum_{i=1}^W\rho^{i-1}V^i_T+c_1\beta_2(T+1)\log\frac{1}{\gamma} \sum_{i=1}^W\rho^{i-1}\\
&+c_1\paren{\beta_3\log\frac{1}{1-\gamma}+\beta_4}\sum_{i=1}^{W}i\rho^{i-1}.
\end{align*}
By the perturbation bound in Lemma~\ref{app_lem:perturbation}, we can upper bound $V^i_T\le \sqrt{p}M^2 i^2 V_T$. Hence, by invoking Lemma~\ref{app_lem:geometric_series} we finally obtain
\begin{align*}
\mathcal{R}(\pi)&\le \alpha_1 \frac{\rho^{2W}}{(1-\rho)^2}T+2c_1\beta_1\sqrt{p}M^2 \frac{1}{1-\gamma}\tilde{W}^3 V_T+c_1\beta_2\tilde{W}(T+1)\log\frac{1}{\gamma}\\
&+c_1\paren{\beta_3\tilde{W}^2\log\frac{1}{1-\gamma}+\tilde{W}^2\beta_4}\\
&=\alpha_1 \frac{\rho^{2W}}{(1-\rho)^2}T+\alpha_2\tilde{W}^4V_T(1-\gamma)^{-1}-\alpha_3\tilde{W}^2(T+1)\log\gamma\\
&-\alpha_4\tilde{W}^3\log (1-\gamma)+\alpha_{5}\tilde{W}^3.
\end{align*}
The coefficients are given by 
\begin{equation}\label{eq:alpha_coeffs}
\begin{aligned}
\alpha_1&=2c^2_0(\snorm{A}+1)^2D^2_r\snorm{\Sigma}\\
\alpha_2&=4c^2_0\snorm{\Sigma} \beta_1\sqrt{p}M^2\\
\alpha_3&=2c^2_0\snorm{\Sigma}\beta_2\\
\alpha_4&=2c^2_0\snorm{\Sigma}\beta_3\\
\alpha_5&=2c^2_0\snorm{\Sigma}\beta_4,
\end{aligned}
\end{equation}
where $\beta_1,\beta_2,\beta_3,\beta_4$ are given in~\eqref{app_eq:betas} and $\snorm{\Sigma}=\snorm{B^\top X B+R}$.\hfill $\qed$
\subsection{Proof of Corollary~\ref{corollary}}
The fifth term is constant since $\tilde{W}$ is bounded. Taking $W=-\frac{log T}{2\log \rho}$ implies that $\rho^{2W}T=1$, hence the truncation term is constant. Note that under the given choice for $\gamma$, we have $(1-\gamma)^{-1}=\mathcal{O}(T^2)$, 
since $1-\gamma\ge \log T/T^2$. Thus, the fourth term is at most logarithmic with $T$: $\log(1-\gamma)^{-1}=\mathcal{O}(\log T)$.

The second term satisfies
\[
\frac{V_T}{1-\gamma}=\sqrt{4M}\frac{V_T \sqrt{T}}{\max\{\sqrt{V_T},\log T/\sqrt{T}\}}=\sqrt{4M}\min\{\sqrt{V_T T},\log T\}
\]

Finally, for the third term, we invoke the elementary inequality
\[
\log \gamma \ge 1-\frac{1}{\gamma},
\]
which is equivalent to
\[
-(T+1)\log \gamma\le (T+1)\frac{1-\gamma}{\gamma}.
\]
Note that the maximum possible value of $V_T$ is $2MT$ since matrices $S_t$ are bounded. Hence the forgetting factor $\gamma=1-\sqrt{\frac{\max\{V_T,\log^2 T/T\}}{4MT}}\ge 1-\frac{1}{\sqrt{2}}\sqrt{\frac{V_T}{2MT}}\ge 1-\frac{1}{\sqrt{2}}$ can be lower bounded.
As a result, we obtain
\[
-(T+1)\log \gamma\le (T+1)\frac{1-\gamma}{\gamma}\le \frac{\sqrt{2}}{\sqrt{2}-1}(T+1)\sqrt{\frac{\max\{V_T,\log^2 T/T\}}{4MT}}=\mathcal{O}(\max\{\sqrt{TV_T},\log T\}).\tag*{\qed}
\]

\clearpage
\section{Simulations}
\label{app:simulations}
\subsection{The Quadrotor Model}
\label{app:quadrotor_model}
To demonstrate the performance of PLOT on a practical use-case we consider the linearized dynamics of the Crazyflie quadrotor \cite{bitcraze}, derived in \cite{beuchat2019n}. The mini quadrotor can be modeled with a continuous-time nonlinear model $g:\mathbb{R}^{9\times 1} \times \mathbb{R}^{4\times 1} \rightarrow \mathbb{R}^{9\times 1}$
\begin{equation*}
    \dot{x} = g(x,u),
\end{equation*}
where the state ${x}$ and the input ${u}$ are defined as 
\begin{equation*}
    {x} :=    \begin{bmatrix}
        {p} \\
        {\dot{p}}\\
        {\psi}
    \end{bmatrix}, \qquad 
    {u}:=     \begin{bmatrix}
        f \\
        \mathbf{\omega}
    \end{bmatrix},
\end{equation*}
with ${p}:=[p_x;p_y;p_z]^\top$ defined as the position vector in the inertial frame, ${\psi}:=[\gamma, \beta, \alpha]^\top$ as the attitude vector in the same frame and $\gamma$, $\beta$ and $\alpha$ denote the attitude angles,  roll, pitch and yaw, respectively. The action ${u}:=[f; {\omega}]^\top$ comprises of the total thrust $f$, as well as the angular rate ${\omega}:=[\omega_x, \omega_y, \omega_z]^\top$ in the body frame. We provide a brief derivation of the linearized model for the quadrotor of interest and refer the readers to \cite{beuchat2019n} for a detailed derivation. 

Given the low speeds of the quadrotor, we neglect the aerodynamic drag forces resulting in the following equations of motion for translation
\begin{equation}
\label{eq:EoM1}
    \Ddot{{p}} = 
    \begin{bmatrix}
        \ddot{p}_x \\
        \ddot{p}_y \\
        \ddot{p}_z
    \end{bmatrix} = \frac{1}{m} \left(R_{IB}
    \begin{bmatrix}
        0 \\
        0 \\
        f
    \end{bmatrix}
    + 
    \begin{bmatrix}
        0 \\
        0 \\
        -mg
    \end{bmatrix}\right),
\end{equation}
where $m$ is the mass of the quadrotor, $g$ is the gravitational acceleration constant, and $R_{IB}$ is the rotation matrix from the body frame to the inertial frame. In particular, given the attitude angles, $R_{IB}$ is given by
\begin{equation*}
    R_{IB} =
    \begin{bmatrix}
        c_{\alpha}c_{\beta} & (-s_{\alpha}c_{\gamma}+c_{\alpha}s_{\beta}s_{\gamma}) & (s_{\alpha}s_{\gamma}+c_{\alpha}s_{\beta}c_{\gamma}) \\
        s_{\alpha}c_{\beta} & (c_{\alpha}c_{\gamma}+s_{\alpha}s_{\beta}s_{\gamma}) & (-c_{\alpha}s_{\gamma
        }+s_{\alpha}s_{\beta}c_{\gamma}) \\
        -s_{\beta} & c_{\beta}s_{\gamma} & c_{\beta}c_{\gamma}
    \end{bmatrix},
\end{equation*}
where $s_{\theta}: = \sin(\theta)$ and $c_{\theta}:=\cos(\theta)$ for a given angle $\theta$. The equations of motion for rotation can similarly be derived as
\begin{equation}
\label{eq:EoM2}
    \dot{{\psi}} = 
    \begin{bmatrix}
        1 & \sin{\gamma}\tan\beta & \cos\gamma\tan\beta \\
        0 & \cos\gamma & -\sin\gamma \\
        0 & \sin\gamma\sec\beta & \cos\gamma\sec\beta
    \end{bmatrix}
    {\omega}.
\end{equation}

Using the equations \eqref{eq:EoM1} and \eqref{eq:EoM2}, the linearized dynamic for the quadrotor can then be attained as
\begin{equation*}
    \dot{\delta x} = \underbrace{\frac{\partial g}{\partial x^\top}\Bigr|_{\substack{x=x_h\\u=u_h}}}_{A_c}\; \delta x \; +\; \underbrace{\frac{\partial g}{\partial u^\top}\Bigr|_{\substack{x=x_h\\u=u_h}}}_{B_c}\; \delta u, 
\end{equation*}
where $x_h = \boldsymbol{0}\in\mathbb{R}^{9\times 1}$ and $u = [mg;0;0;0]^\top$ are the hovering position, steady state, and input, and $\delta x: = x -x_h,\; \delta u := u - u_h$. Calculating the respective Jacobians, and discretizing the resulting dynamics with a sampling time $T_s$, the following dynamics are derived
\begin{align}
\label{eq:quadrotor_model}
    x_{t+1}\; &=      
    \underbrace{\begin{bmatrix}
        1 & 0 & 0 &\;  T_s & 0 & 0 &\;  0 & 0 & 0\\
        0 & 1 & 0 &\;  0 & T_s & 0 &\;  0 & 0 & 0\\
        0 & 0 & 1 &\;  0 & 0 & T_s &\;  0 & 0 & 0\\
        \\
        0 & 0 & 0 &\;  1 & 0 & 0 &\;  0 & gT_s & 0\\
        0 & 0 & 0 &\;  0 & 1 & 0 &\;  -gT_s & 0 & 0\\
        0 & 0 & 0 &\;  0 & 0 & 1 &\;  0 & 0 & 0\\
        \\
        0 & 0 & 0 &\;  0 & 0 & 0 &\;  1 & 0 & 0\\
        0 & 0 & 0 &\;  0 & 0 & 0 &\;  0 & 1 & 0\\
        0 & 0 & 0 &\;  0 & 0 & 0 &\;  0 & 0 & 1
    \end{bmatrix}}_{A}x_t\;
    + \;\underbrace{    \begin{bmatrix}
        0 & 0 & 0 & 0 \\
        0 & 0 & 0 & 0 \\
        0 & 0 & 0 & 0 \\
        \\
        0 & 0 & 0 & 0 \\
        0 & 0 & 0 & 0 \\
        \frac{T_s}{m} & 0 & 0 & 0 \\
        \\
        0 & T_s & 0 & 0 \\
        0 & 0 & T_s & 0 \\
        0 & 0 & 0 & T_s 
    \end{bmatrix}}_{B}u_t,
\end{align}
where $x_t = \delta x_t$, and with a slight abuse of notation we take $u_t : = \delta u_t$ and consider $u_h$ as a feedforward input term required for hovering.

Throughout this section, we will consider the linearised drone dynamics \eqref{eq:quadrotor_model} for PLOT and other online control algorithms. We fix the sampling time to $T_s=0.1$ seconds and the quadratic cost matrices to $Q:=diag(80, 80, 80, 10, 10, 10,0.01,0.01,0.1)$, and $R:=diag(0.7, 2.5, 2.5, 2.5)$ to match the ones tuned for the experiments on the hardware. 

\rev{For the PLOT policy, we take $p=1$ for all experiments, and the projection step $8$ in the RLS Algorithm \ref{alg:RLS} is executed with a quadratic program solver, solving the following problem
\begin{align*}
    \hat{S}_{t+k|t} = \arg\min_{S} &\snorm{S-Y}_{F,P_{t|t-k}}\\
    \text{s.t.}\;& \|\text{vec}(S)\|_{\infty} \leq M,
\end{align*}
where $\text{vec}(S)$ is a vector formed by vertically stacking the columns of $S$. Note that, the $\|\cdot\|_{\infty}$ constraint is a conservative approximation of the original problem, resulting in a safe, and conservative estimate of the original projection $\Pi_{\mathcal{S}}^{P_{t|t-k}}$.}

In this section, we first analyze PLOT,  showing the effect of the prediction horizon, $W$ on the tracking performance, and how this is reflected in the dynamic regret bound. We also show how the regret bound can be used to adjust the forgetting factor, $\gamma$, to achieve better tracking given the order of the reference target path length. Next, we compare the proposed method with state-of-the-art online control algorithms for the non-stochastic adversarial control setting with both static and dynamic regret bounds. We show how for certain dynamic references the proposed indirect method of PLOT provides better tracking \rev{and lower dynamic regret, and how for others it can perform worse than direct methods.}

\subsection{PLOT: Regret for hyperparameter tuning}
\label{app:plot_hyperparameters}
The dynamic regret guarantees derived in Theorem \ref{thm:main} and Corollary \ref{corollary} can be used to tune the hyperparameters of the proposed method, such as prediction horizon length, $W$, or the forgetting factor $\gamma$. Though the exact bounds are often over-conservative in practice, the bound order still provides an intuition of the effect of the parameters on tracking performance.

\subsubsection{The prediction horizon $W$}
\label{sec:circle_reference}
\begin{figure}[!ht]
\centering
    \begin{subfigure}[t]{0.23\textwidth}
        \centering
        \includegraphics[width=\linewidth]{figs/circle_2sec.eps}
        \caption{$T=2$ seconds}
    \end{subfigure}%
    ~ 
    \begin{subfigure}[t]{0.23\textwidth}
        \centering
        \includegraphics[width=\linewidth]{figs/circle_3sec.eps}
        \caption{$T=3$ seconds}
    \end{subfigure}
    ~ 
    \begin{subfigure}[t]{0.23\textwidth}
        \centering
        \includegraphics[width=\linewidth]{figs/circle_5sec.eps}
        \caption{$T=5$ seconds}
    \end{subfigure}
    ~ 
    \begin{subfigure}[t]{0.23\textwidth}
        \centering
        \includegraphics[width=\linewidth]{figs/circle_7sec.eps}
        \caption{$T=7$ seconds}
    \end{subfigure}
    \caption{Trajectory plots of a circular target with a $V_T=0$ path length and the PLOT Algorithm for varying prediction horizon lengths, simulated for $T=2,3,5$ and $T=7$ seconds.}
    \label{fig:circle_rls_rhc_video}
\begin{minipage}[t]{0.48\textwidth}
  \centering
  \includegraphics[width=0.9\linewidth]{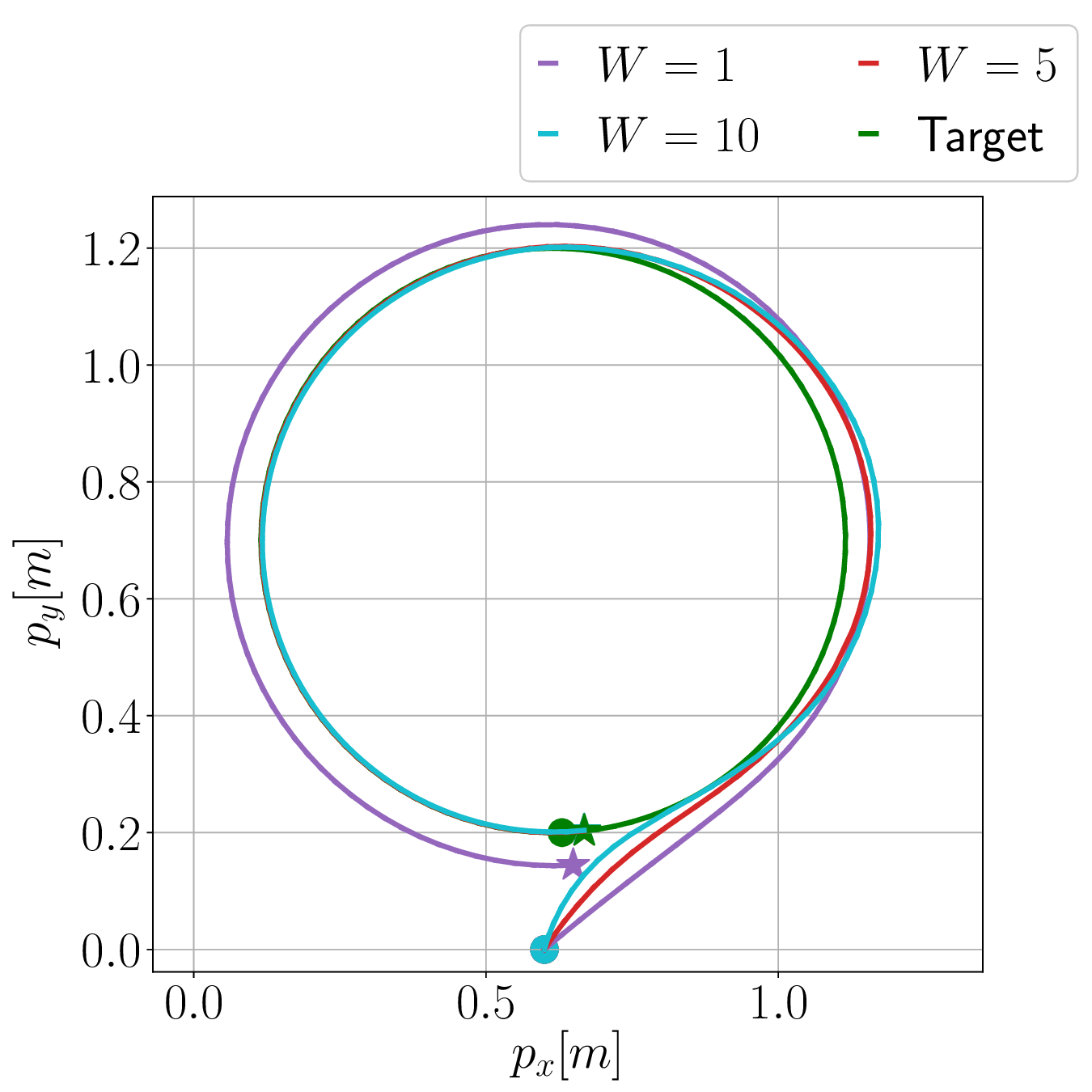}
  \caption{Trajectory plot of a circular target with a $V_T=0$ path length and the PLOT Algorithm for varying prediction horizon lengths, simulated for $T=10.7$ seconds.}
  \label{fig:circle_rls_rhc}
\end{minipage}%
 \hfill
\begin{minipage}[t]{.48\textwidth}
  \centering
  \includegraphics[width=0.9\linewidth]{figs/regret_circle_W.eps}
  \caption{Log-normalized regret of the PLOT Algorithm with a range of prediction horizon lengths, simulated over a horizon of $T = 200$ seconds.}
  \label{fig:regret_rls_rhc}
\end{minipage}
\end{figure}

In this setting, we aim to demonstrate the effect of the prediction horizon length, $W$, on the tracking performance of PLOT. We consider a static target with a $V_T=0$ path length. A simple target respecting this condition is the circular target, as introduced in Example \ref{exmp:circular} with $\theta = 0.06$ radians. With this target revealed and measured online, as described in Section \ref{sec:problem statement}, we run the PLOT algorithm repeatedly for varying numbers of horizon lengths. For all, we set the same initial state of $p= [0.6, 0.0, 0.4]^\top$, $M=10$, and forgetting factor $\gamma =0.8$. We define the augmented matrix $\Tilde{S}_t \triangleq \matr{S_{t} &v_t}$ for $t = 1, \hdots, T-1$ as in Section \ref{sec:problem statement}, and for all $k = 1, \hdots, W$, the $k$ learners are initialized as follows: $\hat{S}_{j+k|j} = \mathbf{I}_6, v_{j+k|j} = \mathbf{0}_6$  and $P_{j|j-k} = 10^{-4}\times\mathbf{I}_7$, for all $j = -1, \hdots, k-2$. Here, $\mathbf{1}_d$ denotes the $d$-dimensional vector of all ones.

The circular target, as well as the trajectory followed by the quadrotor model \eqref{eq:quadrotor_model} under the PLOT algorithm is shown in Figure \ref{fig:circle_rls_rhc_video} at $T=2, 3, 5$ and $T=7$ seconds. Figure \ref{fig:circle_rls_rhc} shows one complete round of a circle with  $T=10.7$ seconds. The corresponding regret plots (including more prediction horizons) are shown in Figure \ref{fig:regret_rls_rhc} for a longer simulation of $T=200$ seconds.

Several observations can be made from these experiments. Firstly, a prediction horizon length not long enough, e.g. $W=1$ or $W=3$ results in a poor performance, reflected in both Figures in terms of tracking and regret. While higher horizon lengths achieve efficient tracking by learning the reference dynamics on the go, achieving sublinear regret. This behavior matches the intuition Corollary \ref{corollary} provides, as a long enough $W$ is required to achieve logarithmic regret in $T$. Secondly, an increase in regret of PLOT for larger $W$-s can be observed in Figure \ref{fig:regret_rls_rhc}, as ``predicted" by the second term of the bound in Theorem \ref{thm:main}. This increase in regret, however, is saturated at a certain value as can be observed from the figure; in this example, the increase stops after around $W=15$. This behavior also is captured by the bound in Theorem \ref{thm:main}, as, crucially, the second term is weighted by increasing powers of $\rho$, which reduces regret exponentially counteracting the linear increase with $W$. This shows that the PLOT does not suffer, at least in terms of regret, due to higher prediction horizons.

\subsubsection{The forgetting factor $\gamma$}
\label{sec:spiral_reference}
As suggested by Corollary \ref{corollary}, the desired regret rate can be achieved by tuning the forgetting factor, $\gamma$, given the path length $V_T$, and the control horizon, $T$. In practice, the exact constants appearing in the expression of optimal $\gamma$ are not known. However, one can still tune the gamma based on the order of the path length and the prediction horizon. In particular, we fix $\gamma_a = 1 - c_{\gamma}T^{-a}$, where $c_\gamma \in \mathbb{R}_+$ is constant independent of $V_T$ and $T$, while  $a \in \mathbb{R}_+$ is such that $\gamma_a \in (0,1]$ and is tuned based on $V_T$ and $T$. 

\begin{figure}[t]
    \centering
    \begin{subfigure}[t]{0.23\textwidth}
        \centering
        \includegraphics[width=\linewidth]{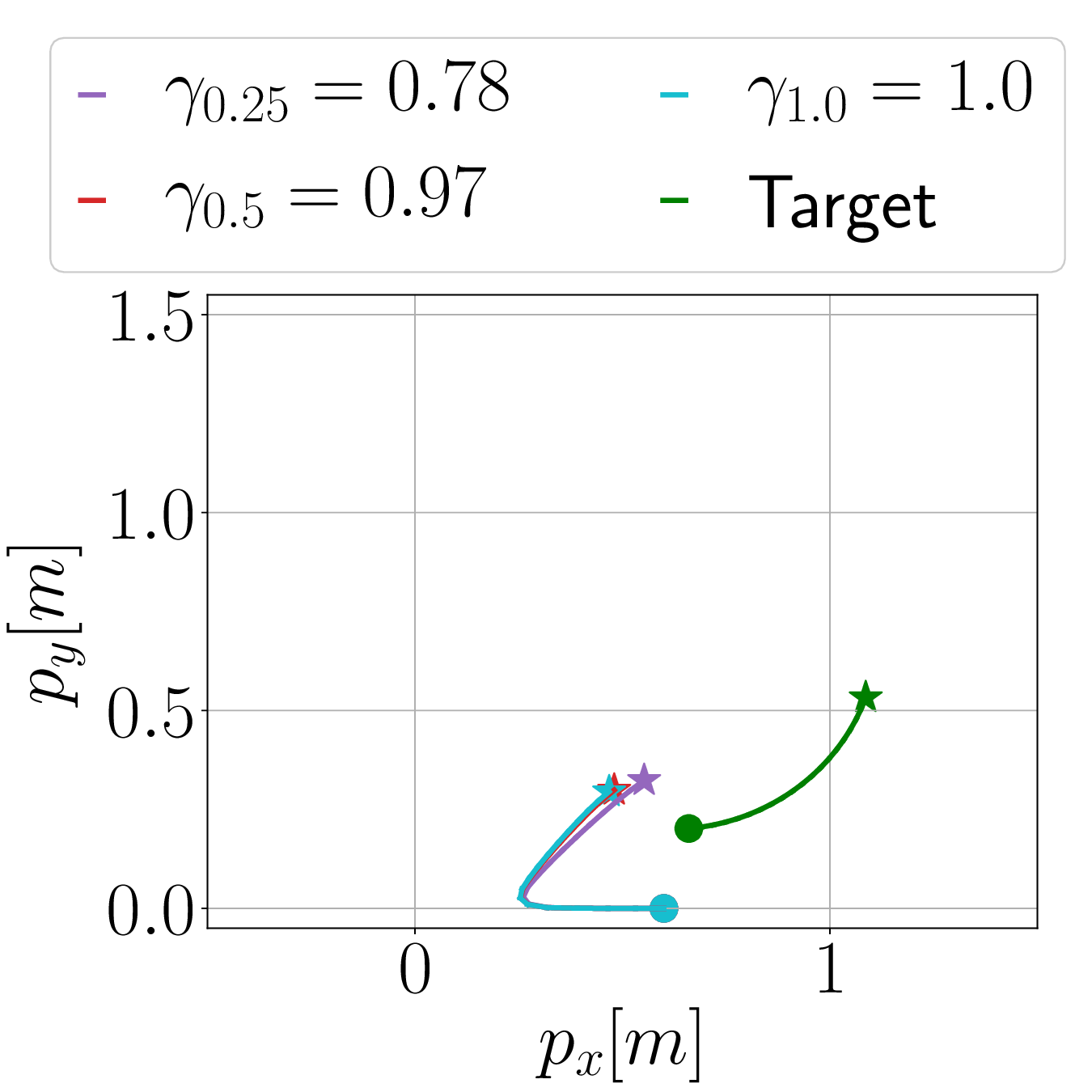}
        \caption{$T=2$ seconds}
    \end{subfigure}%
    ~ 
    \begin{subfigure}[t]{0.23\textwidth}
        \centering
        \includegraphics[width=\linewidth]{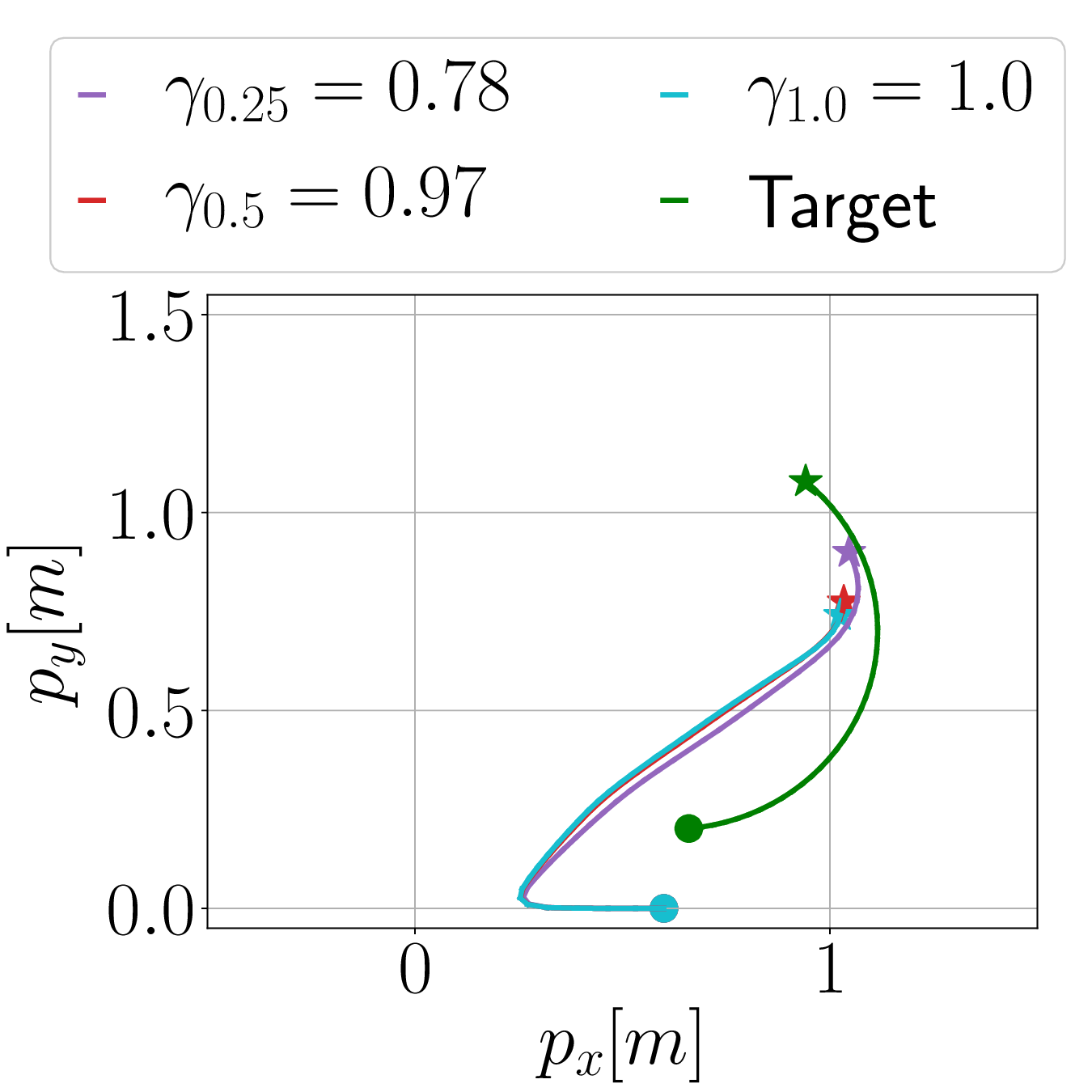}
        \caption{$T=4$ seconds}
    \end{subfigure}
    ~ 
    \begin{subfigure}[t]{0.23\textwidth}
        \centering
        \includegraphics[width=\linewidth]{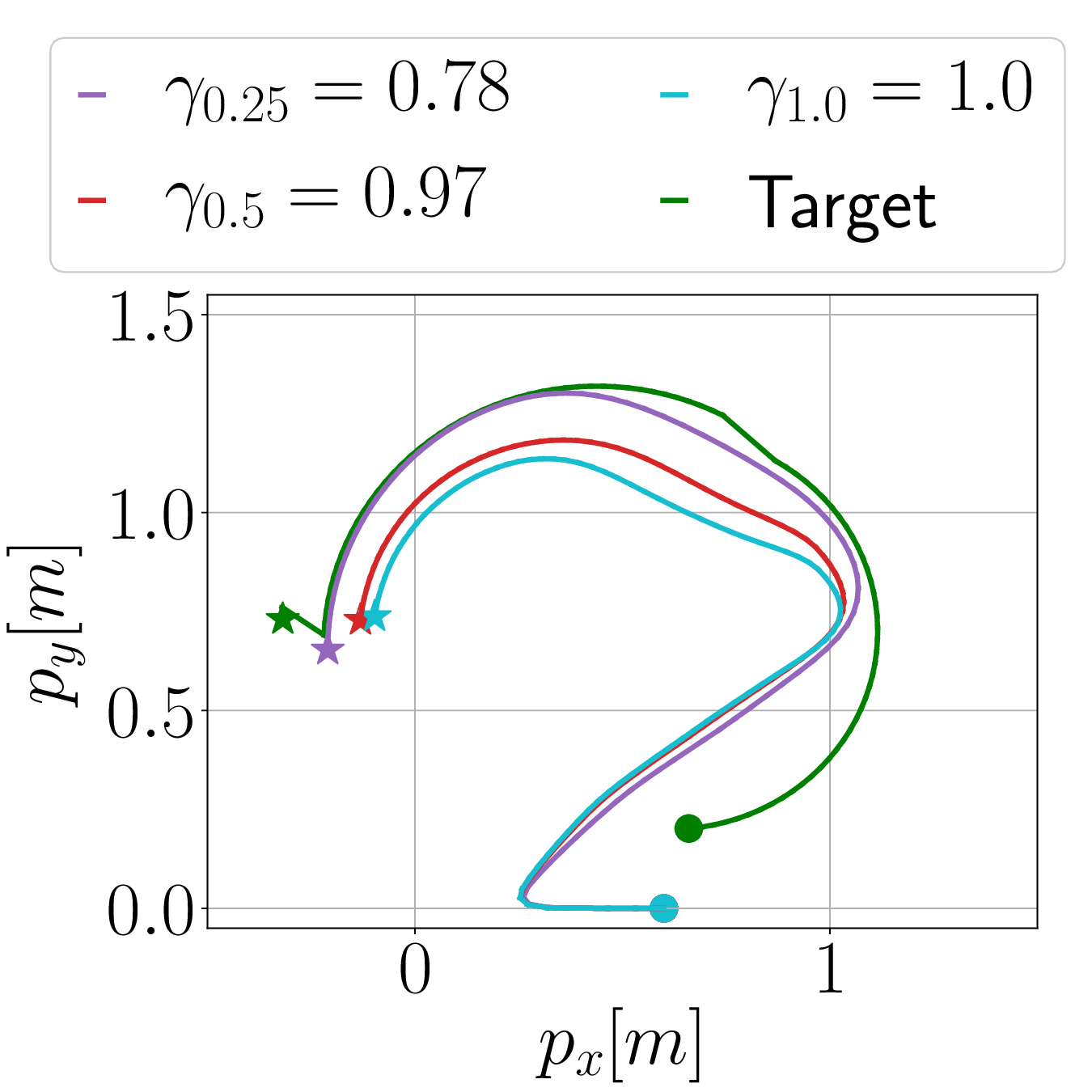}
        \caption{$T=9$ seconds}
    \end{subfigure}
    ~ 
    \begin{subfigure}[t]{0.23\textwidth}
        \centering
        \includegraphics[width=\linewidth]{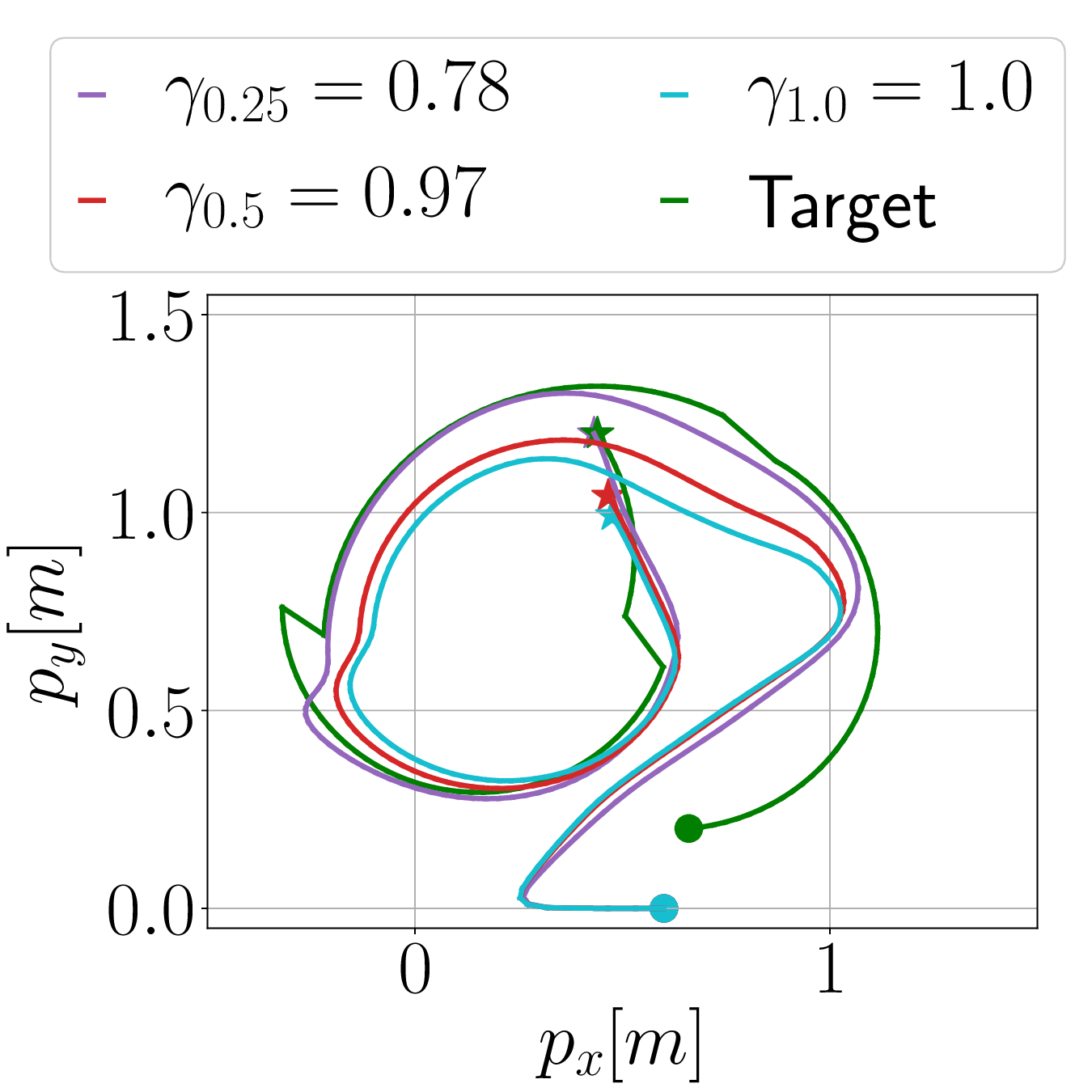}
        \caption{$T=15$ seconds}
    \end{subfigure}
    \caption{Trajectory plot of a spiral with a $V_T = \mathcal{O}(\sqrt{T})$ tracked with PLOT with $W=5$ and a range of values for $\gamma$, simulated for $T=2,4,9$ and $T=15$ seconds.}
    \label{fig:gamma_video}
\end{figure}

To show this on an example, we first fix a dynamic reference trajectory with a path length $V_T = \mathcal{O}(\sqrt{T})$ and perform $5$ different runs of PLOT for a range of values for $\gamma = \gamma_a$, with $a = 0.10, 0.20, 0.25, 0.40, 0.50$ and $1.0$. We fix $M=1$ and $W=5$. Note that for this path length the optimal value for $a$, suggested by Corollary \ref{corollary} is $a = 0.25$. The dynamics for the reference are generated by starting with the circle static dynamics in the previous example and increasing the radius of the circle every $\sqrt{T}$ time steps by a factor sampled uniformly from $0.7$ and $1.5$. In addition, a shift of $[-0.1, 0.1]^\top$ is applied to $[p_x,p_y]^\top$ at the same time step.
As in the previous example, we define the augmented matrix $\Tilde{S}_t \triangleq \matr{S_{t} &v_t}$ for $t = 1, \hdots, T-1$ as in Section \ref{sec:problem statement}. To highlight the effect of the forgetting factor, for all $k = 1, \hdots, W$, the $k$ learners are initialized as follows: $\hat{S}_{j+k|j} = \mathbf{0}_{6\times 6}, v_{j+k|j} = \mathbf{0}_6$ and $P_{j|j-k} = \mathbf{I}_7$, for all $j = -1, \hdots, k-2$.

\begin{figure}[!ht]
\centering
\begin{minipage}[t]{.48\textwidth}
  \centering
  \includegraphics[width=0.9\linewidth]{figs/switching_target_trajectory_40sec.eps}
  \caption{Trajectory plot of a spiral with a $V_T = \mathcal{O}(\sqrt{T})$ tracked with PLOT with a $W=5$ and a range of values for $\gamma$.}
  \label{fig:gamma_trajectory}
\end{minipage}
 \hfill
\begin{minipage}[t]{0.48\textwidth}
\centering
  \includegraphics[width=0.9\linewidth]{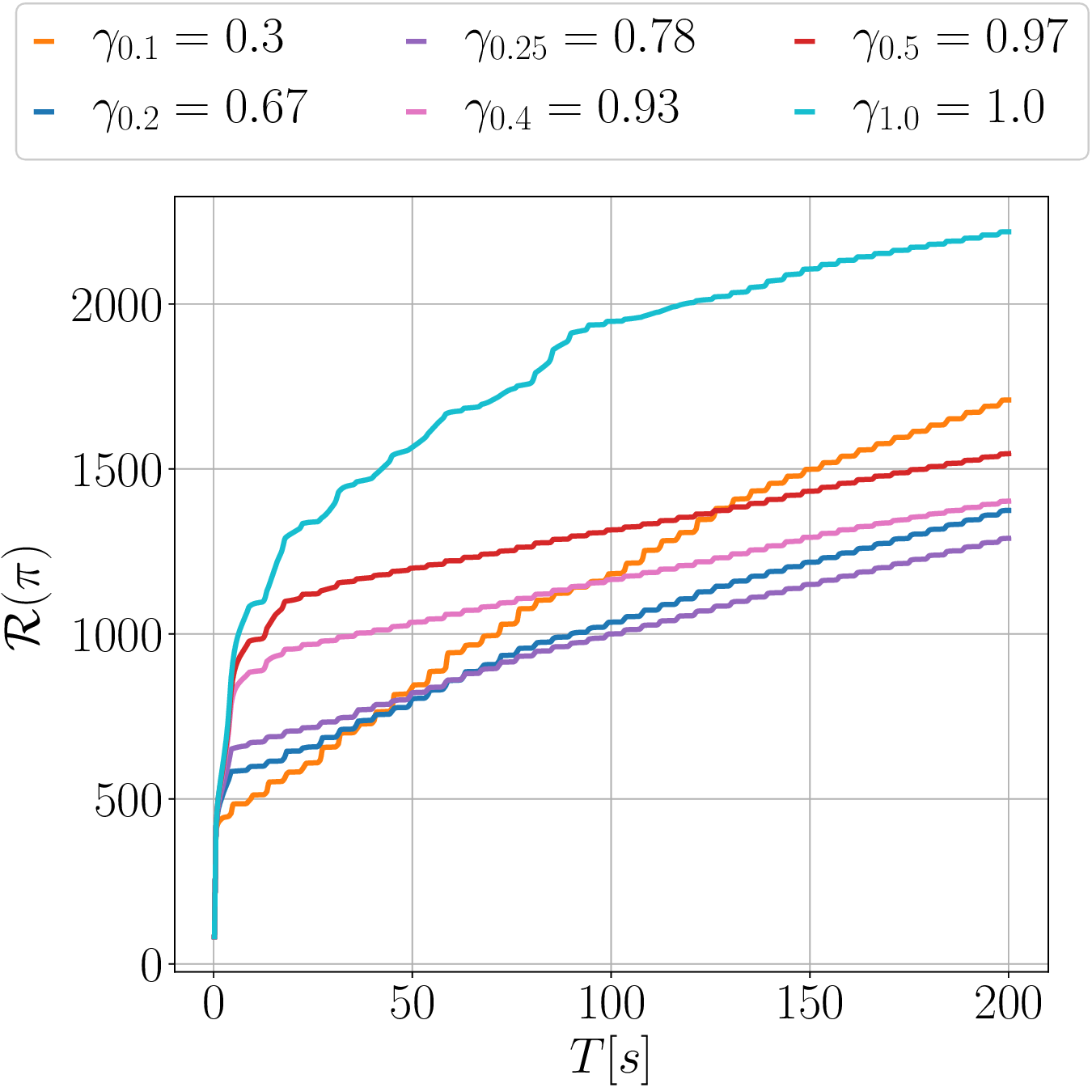}
  \caption{Regret of PLOT for the spiral target with $V_T = \sqrt{T}$ for a range of $\gamma$-s. $\gamma=0.78$ is chosen as per Corollary \ref{corollary}.}
  \label{fig:gamma}
\end{minipage}%
\end{figure}

The dynamic target trajectory together with the system trajectory under the PLOT controller with different forgetting factors is plotted in Figure \ref{fig:gamma_video} for $T=2,4,9$ and $T=15$ seconds. A longer trajectory for $T=40$ seconds is provided in Figure \ref{fig:gamma_trajectory} with the corresponding regret plot in Figure \ref{fig:gamma} for $T=200$ seconds and with more values of $\gamma$. It can be observed from the figures that $\gamma_{0.25}$ tuned optimally to scale with $1-T^{-0.25}$ achieves the lowest regret at the end of the horizon compared to other $a$-s. Intuitively, forgetting helps to adapt quickly to the fast-changing dynamics of the reference and lose the influence of the random initialization of the predictions exponentially fast. This can be observed on the plot, by noticing that the regret for $\gamma=1.0$, corresponding to PLOT without any forgetting is the highest, as expected from the theoretical results. The random initialization of learners, as well as the fast-changing dynamics of the reference, make the no-forgetting tracker maintain a larger tracking error. At the other extreme, forgetting too much, as in this case with $\gamma = 0.3$ hinders the learner from estimating the dynamics,  enough to cause higher regret.  

To show that the scaling constant is independent of horizon $T$, we fix, $c_\gamma$ to a tuned value of $1.5$ and perform experiments with varying horizon lengths from $T=150$ to $T=300$ seconds. For each $T$, PLOT is run $4$ times, each time with a different $\gamma_a$. The regret at the end of each horizon is calculated, and the results are visualized in Figure \ref{fig:tuned_gamma}. As expected from Corollary \ref{corollary}, PLOT with forgetting factor $\gamma_{0.25}$ performs the best for such a reference, across different horizon lengths.

\begin{figure}
    \centering
    \includegraphics[width=0.432\linewidth]{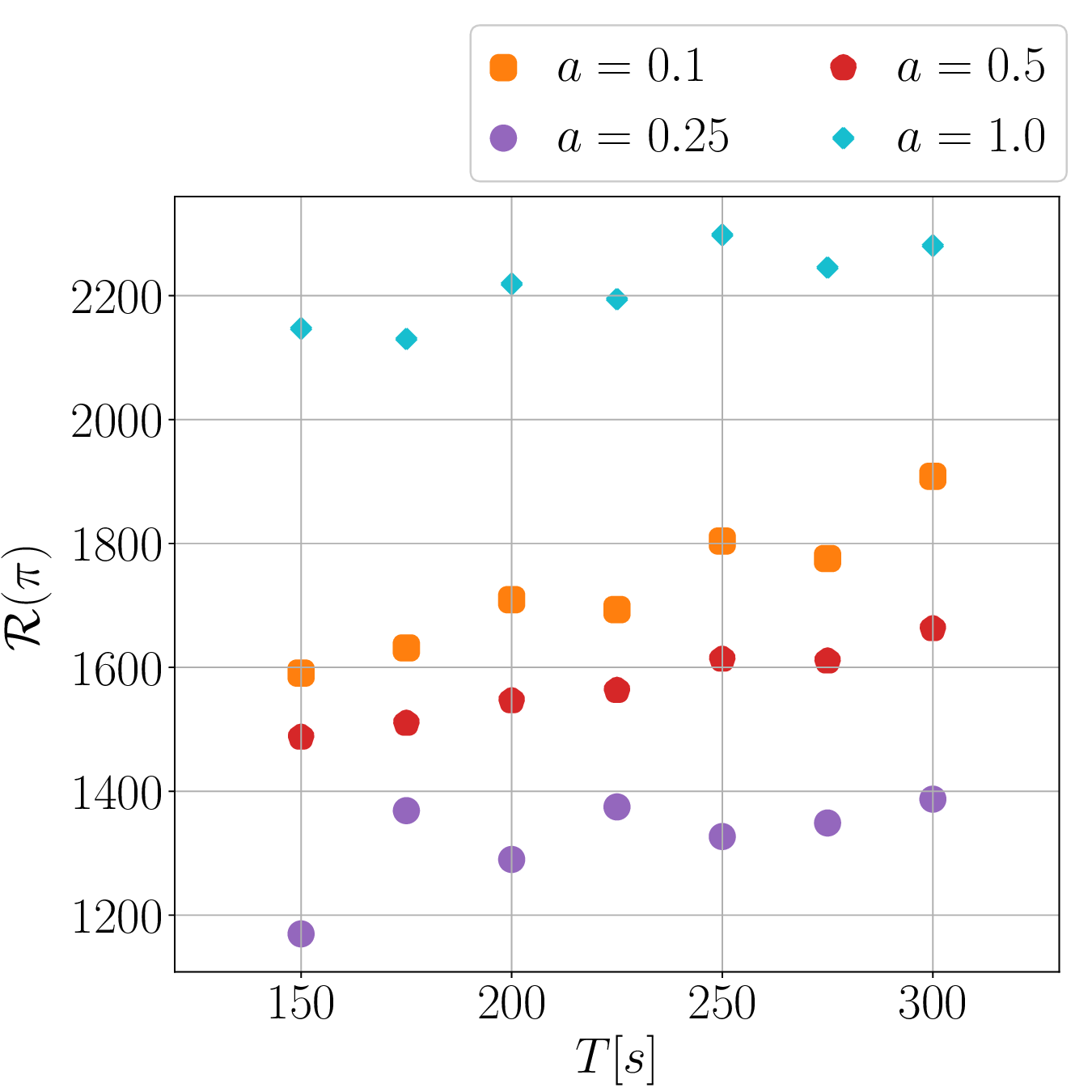}
    \caption{Regret of PLOT with varying $\gamma_a = 1-c_\gamma T^{-a}$.}
    \label{fig:tuned_gamma}
\end{figure}

\subsection{Comparison with Benchmarks}
\label{app:benchmarks}

In this section, we consider a highly dynamic reference target and compare PLOT to state-of-the-art online learning algorithms that can also be applied to the online tracking problem in the linear quadratic setting. In particular, we consider the following dynamics for the target with $p=1$

\begin{equation}
    \label{eq:dynamic_target_example}
    r_{t+1} =     S_{t+1|t} = 
    \begin{bmatrix}
        1 & 0 & T_s & 0\\
        0 & 1 & 0 & T_s\\
        0 & 0 & s_t\cos{\theta_t} & -s_t\sin{\theta_t}\\
        0 & 0 & s_t\sin{\theta_t} & s_t\cos{\theta_t}
    \end{bmatrix}r_t,
\end{equation}
where $s_0 = 1, \theta_0 = 0.06$, $T_s = 0.1 s$ and $s_k = -s_{k-1}$, $\theta_k = -0.99\times\theta_{k-1}$ for every $k = \sqrt{T}$ given some $T$. We simulate $T=200$ seconds or $2000$ time steps. 

We compare PLOT to 5 algorithms, namely, the Follow the Leader (FTL) algorithm \cite{pmlr-v32-abbasi-yadkori14}, Disturbance Action Policy controller \cite{agarwal2019online}, Riccatitron \cite{foster2020logarithmic}, SS-OGD \cite{karapetyan2023online}, and Naive LQR that only performs a static feedback on the currently observed tracking error. Apart from the latter, we tune each controller's hyperparameters to achieve the best performance for the given reference, by performing a grid search over the hyperparameters. All controllers start from the same initial state and follow the same target revealed at the same point in time. We provide the details for each below in Appendix \ref{app:hyperparameters}.

\begin{figure}
    \begin{subfigure}{.33\textwidth}
      \centering
      \includegraphics[width=\linewidth]{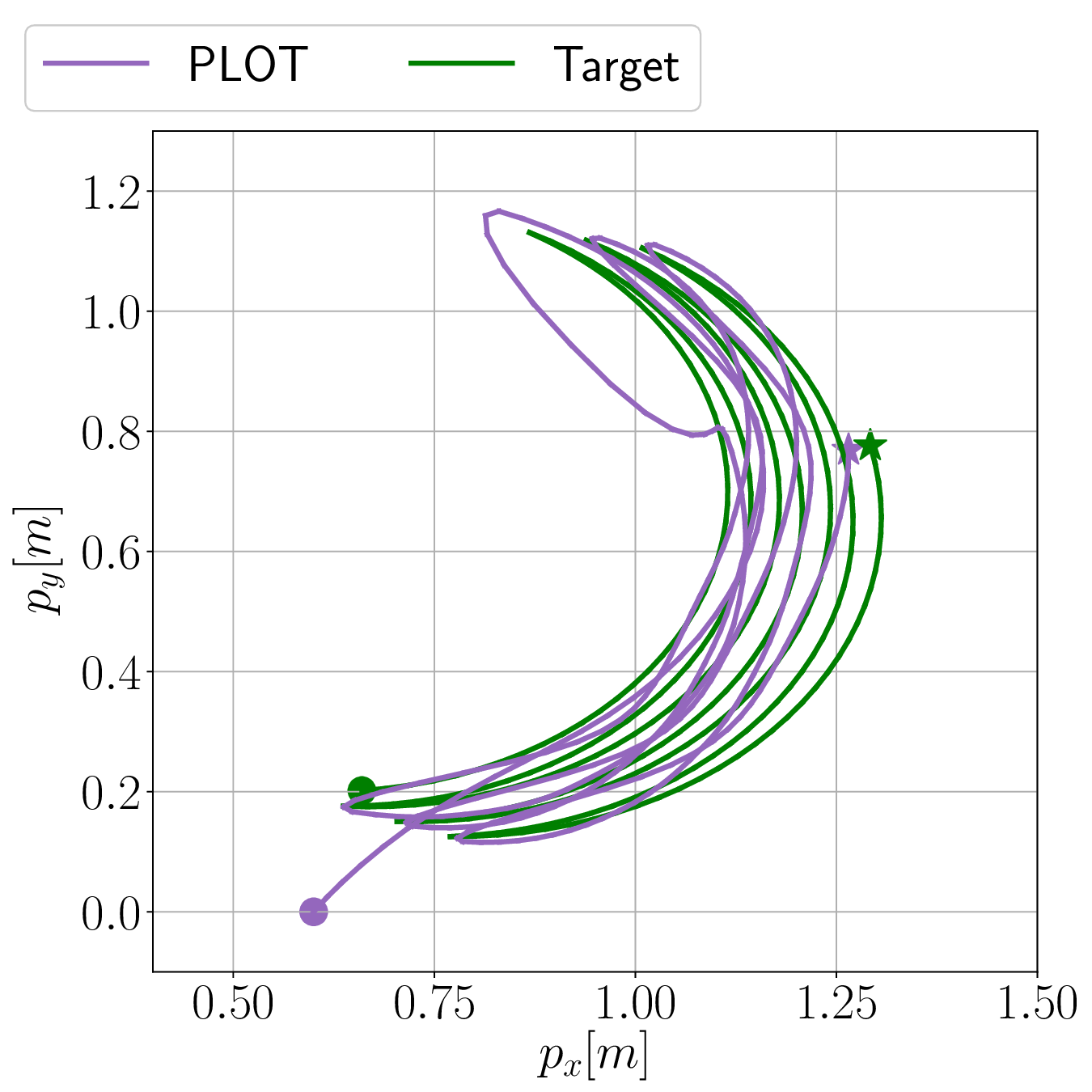}
      \caption{PLOT (ours)}
      \label{fig:bench_plot}
    \end{subfigure}%
    \begin{subfigure}{.33\textwidth}
      \centering
      \includegraphics[width=\linewidth]{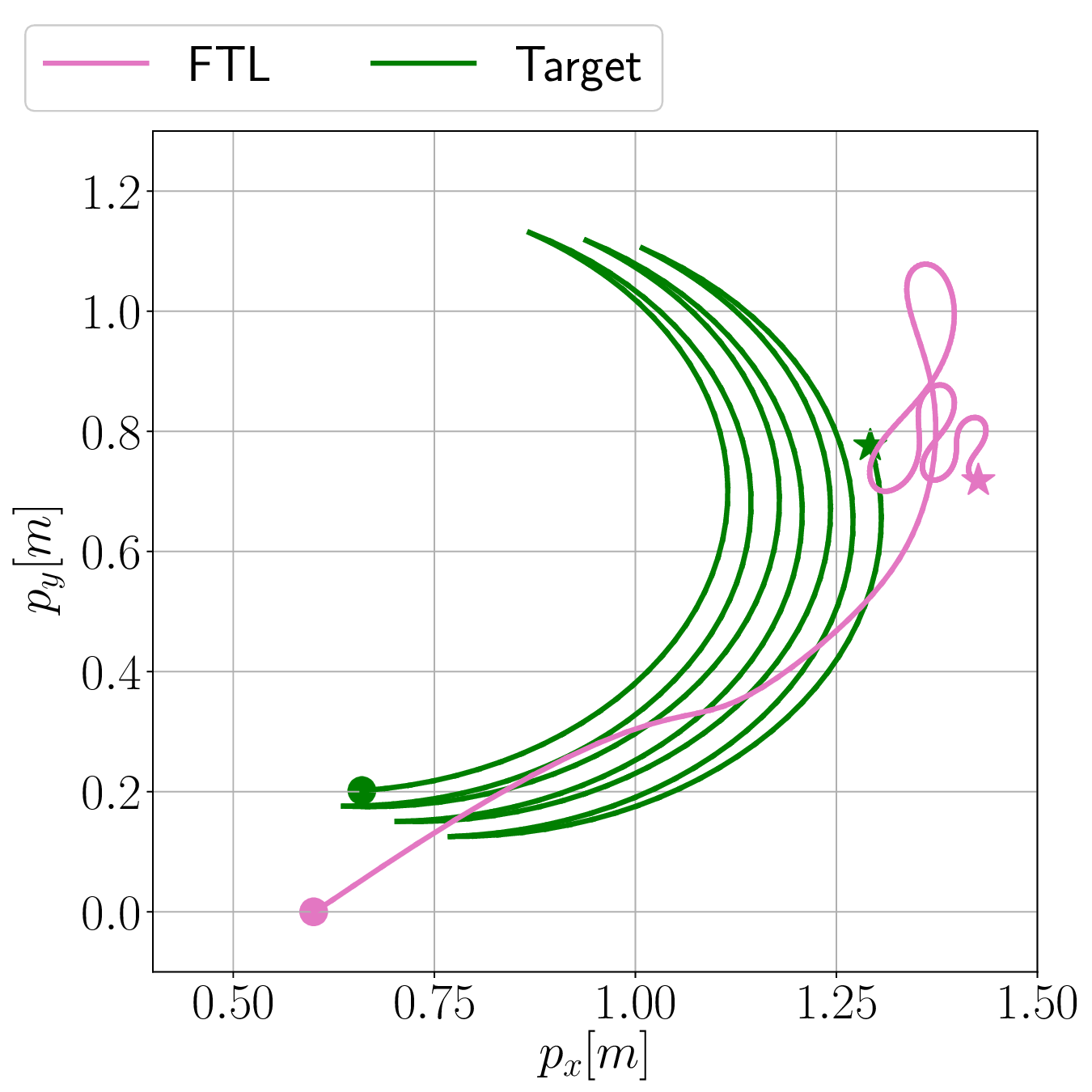}
      \caption{FTL \cite{pmlr-v32-abbasi-yadkori14}}
      \label{fig:bench_ftl}
    \end{subfigure}%
    \begin{subfigure}{.33\textwidth}
      \centering
      \includegraphics[width=\linewidth]{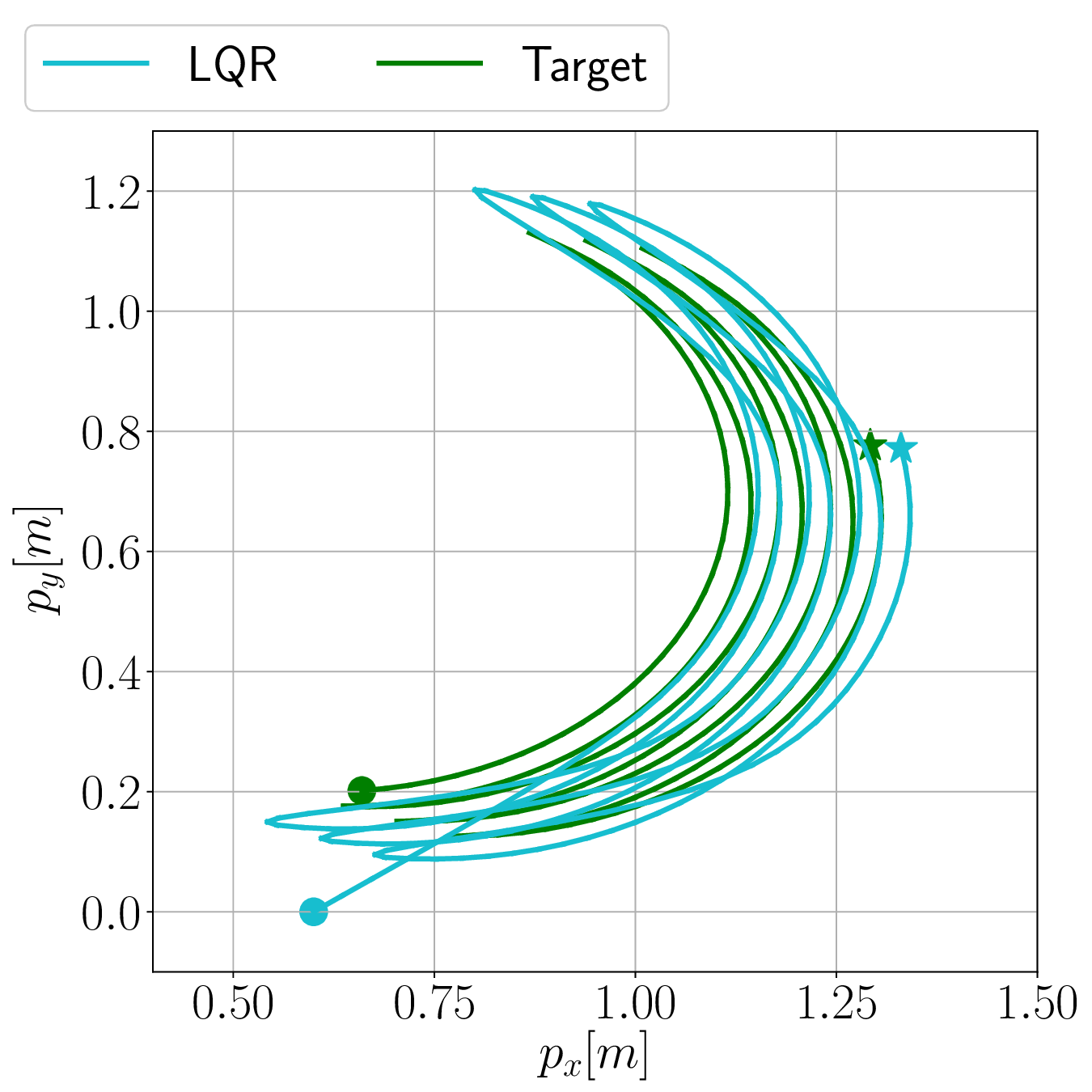}
      \caption{Naive LQR}
      \label{fig:bench_LQR}
    \end{subfigure}\\
    \begin{subfigure}{.33\textwidth}
      \centering
      \includegraphics[width=\linewidth]{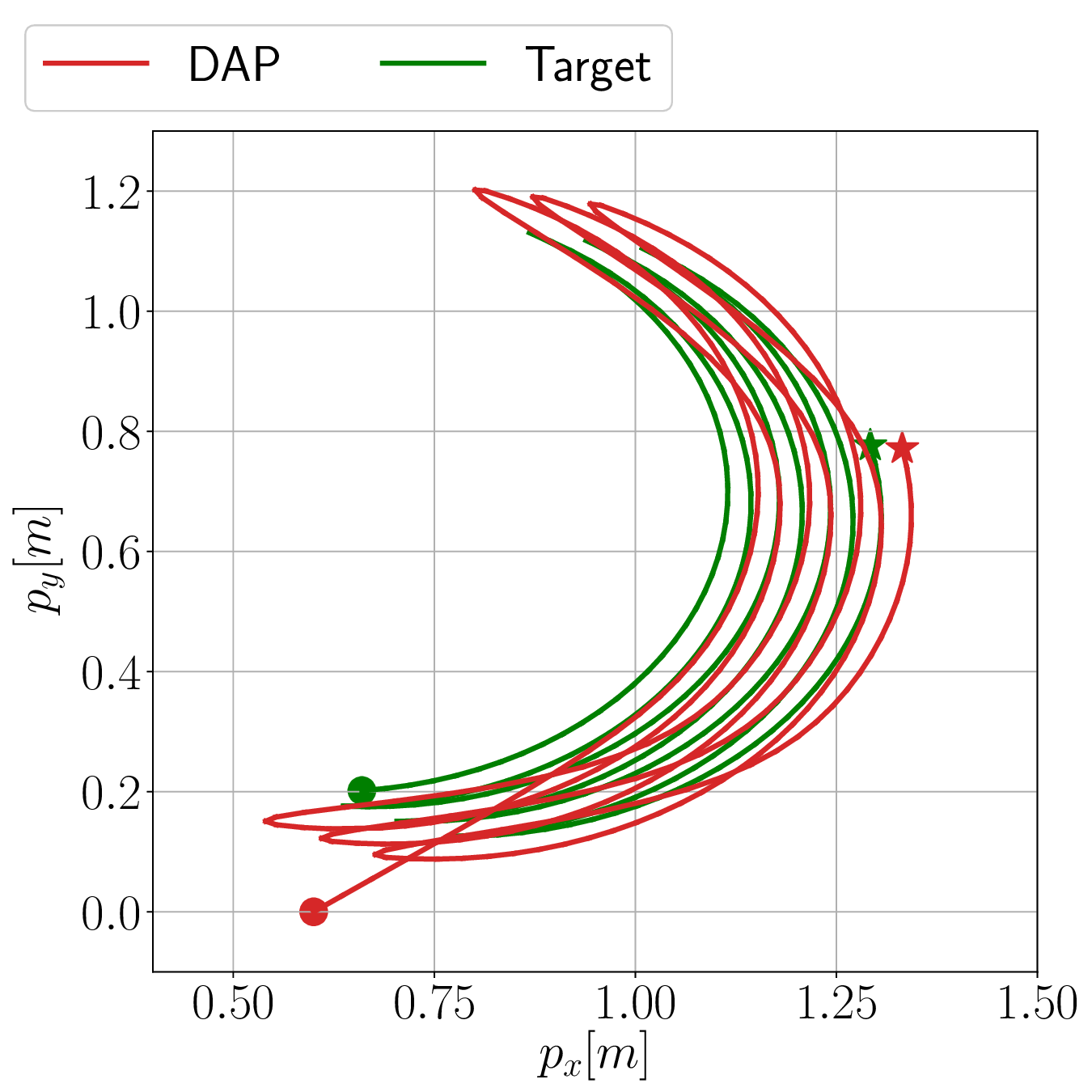}
      \caption{DAP \cite{agarwal2019online}}
      \label{fig:bench_dap}
    \end{subfigure}%
    \begin{subfigure}{.33\textwidth}
      \centering
      \includegraphics[width=\linewidth]{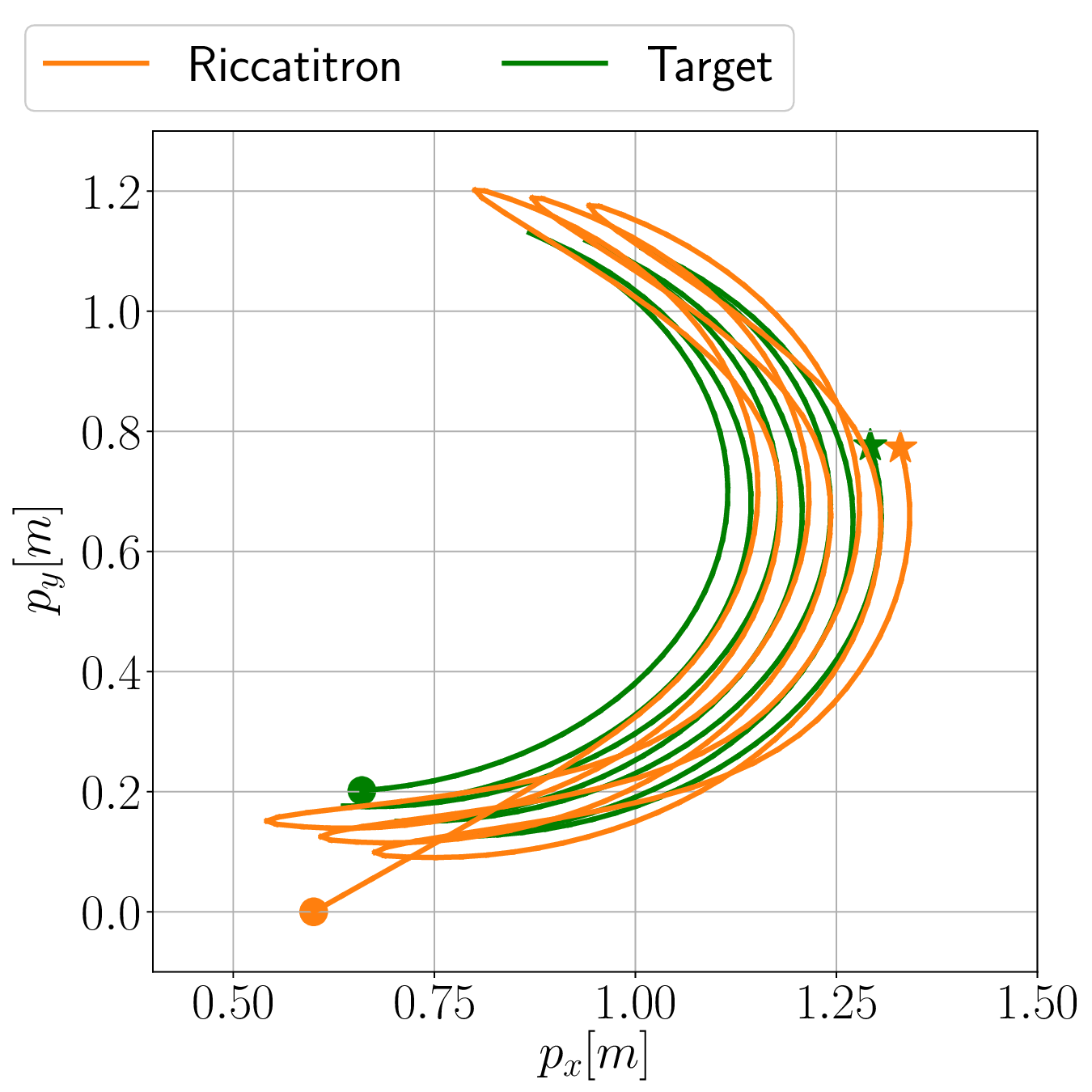}
      \caption{Riccatitron \cite{foster2020logarithmic}}
      \label{fig:bench_riccatitron}
    \end{subfigure}%
    \begin{subfigure}{.33\textwidth}
      \centering
      \includegraphics[width=\linewidth]{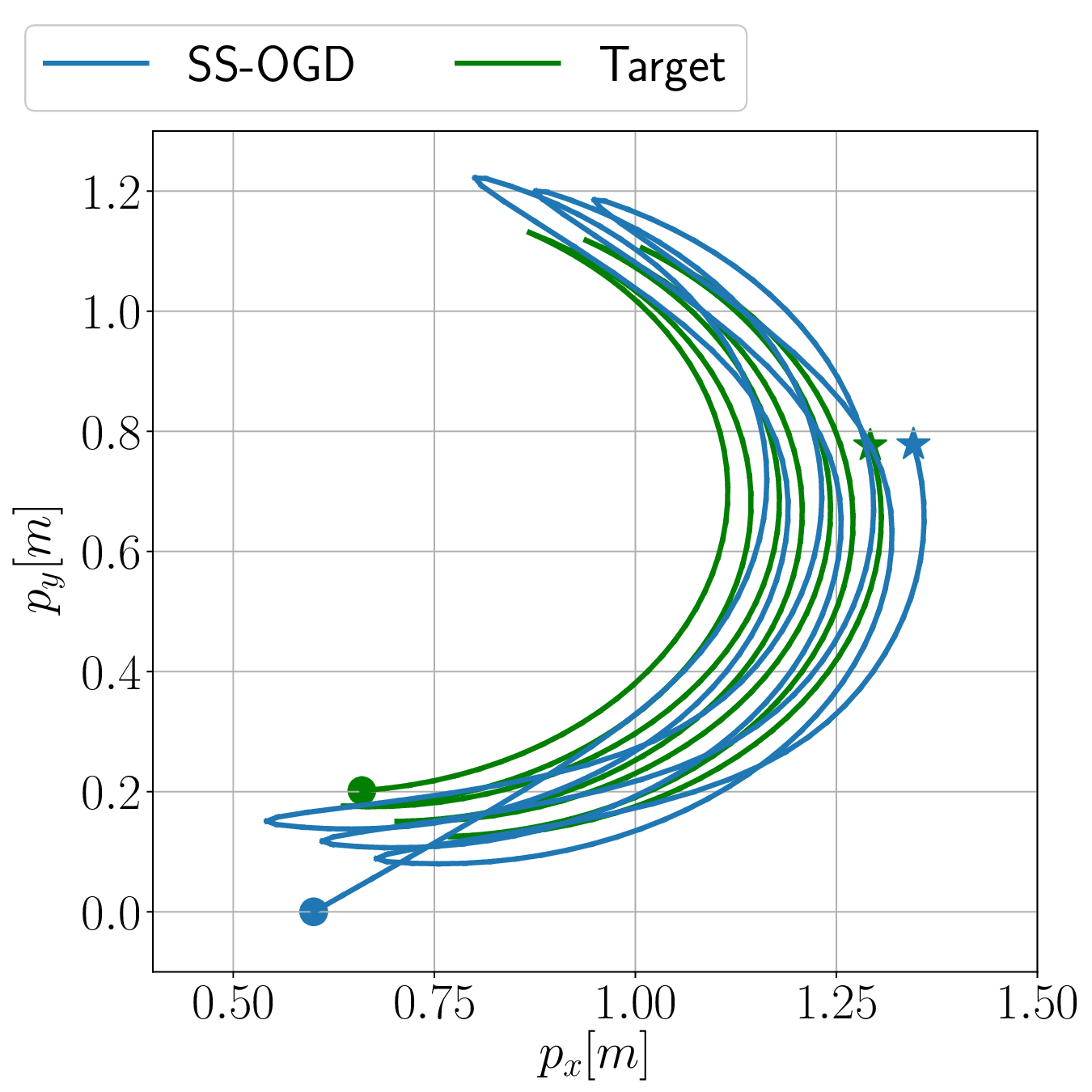}
      \caption{SS-OGD \cite{karapetyan2023online}}
      \label{fig:bench_SSOGD}
    \end{subfigure}
    \caption{Trajectory plots of a reference target with a $V_T = \mathcal{O}(\sqrt{T})$ following \eqref{eq:dynamic_target_example} tracked by PLOT and other online control algorithms for a $T=30$ seconds long horizon.}
    \label{fig:benchmark_trajectories}
 \end{figure}

The trajectories of the target and the system under the considered algorithms are plotted in Figure \ref{fig:benchmark_trajectories} for the first $30$ seconds of the simulation. The corresponding dynamic regret plot for the full $T=200$ seconds horizon is shown in Figure \ref{fig:benchmark_regrets}. Given the dynamic nature of the reference with a sublinear path length, PLOT outperforms all of the other considered benchmarks. There are two main possible reasons for this. Firstly, our approach is aimed at dynamic targets, i.e. unlike the others, it incorporates a forgetting factor that equips it with adaptive capabilities and enjoys dynamic regret guarantees. This allows it to deal better with such time-varying references to which it is hard or impossible to fit a time-invariant model better compared to the other static methods. Secondly, for such targets, the indirect approach of PLOT, i.e., learning the target dynamics and then incorporating these in the control action, is more advantageous than a direct approach deployed by all the other benchmarks. While the affine term of the optimal control action changes drastically throughout the horizon due to the sign and direction change of the target, the dynamics of the exosystem do not change as drastically. Thus, PLOT, learning the dynamics of the exosystem with forgetting is able to provide a better tracking performance, while DAP, Riccatitron, and SS-OGD learn a much smaller affine term for the controller. This results in their performance being indistinguishable from that of the naive LQR. The FTL algorithm aiming for an average best performance instead learns a control action that keeps it roughly in the center of the evolving spiral target ending up with a dynamic regret of an order higher compared to the others.

\begin{figure}
    \centering
    \includegraphics[scale=0.3]{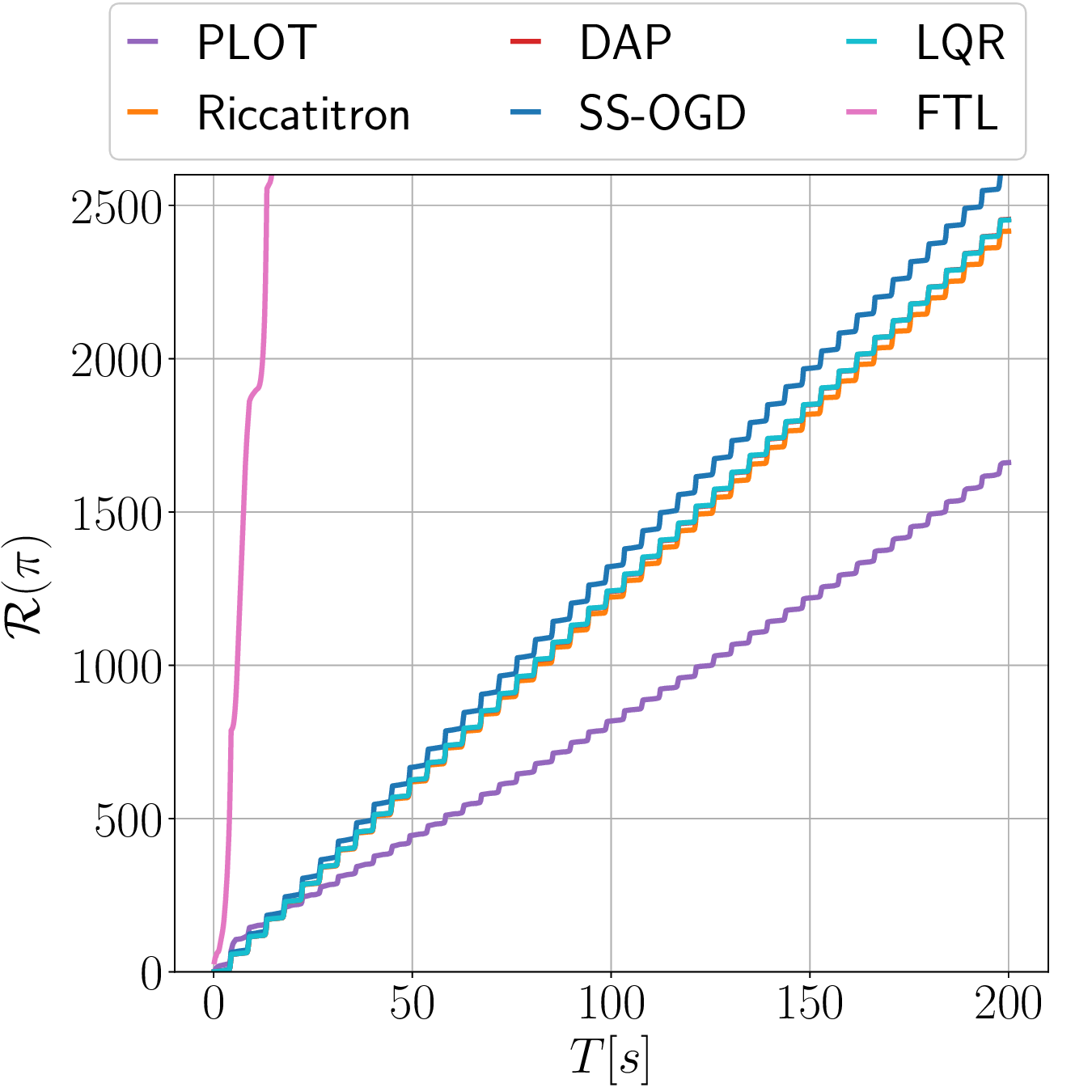}
    \caption{Dynamic Regret of Online Control algorithms applied to the online tracking problem of the unknown target following \eqref{eq:dynamic_target_example}.}
    \label{fig:benchmark_regrets}
\end{figure}

Table \ref{tab:benchmark} compares the average computational time for each iteration of the algorithm in milliseconds, the size of the memory that needs to be reserved and updated for the online execution of the algorithm, and the accumulated regret at the end of the given target trajectory. The average computational time is obtained by averaging the time it takes to compute the control action at each timestep over the entire horizon of $2000$ timesteps. The memory for each controller shows the size of the past state and control variables that need to be stored in the memory. For example for the naive LQR controller, this is $0$, while for SS-OGD it is $4$, as at each timestep it requires the past control input.

\rev{Naive LQR, performing only a single matrix multiplication at each timestep is unsurprisingly the fastest, requiring no extra memory. The computational time of FTL and SS-OGD is of the same order, however, for this benchmark, both incur a higher cost and therefore a higher regret. When PLOT is implemented by approximating step $8$ of Algorithm \ref{alg:RLS} with an unweighted Frobenious norm projection, it can be implemented by a thresholding operation on singular values, requiring no quadratic programs to be solved. In this case, PLOT, Riccatitron, and DAP have similar computation times but are all an order higher compared to the Naive LQR. While compared to it, Riccatitron and DAP attain almost the same regret, PLOT outperforms both, requires less memory than  Riccatitron, and is around a millisecond faster than DAP. When the full weighted projection is carried out in Algorithm \ref{alg:RLS}, step $8$, by approximately solving it with a quadratic program consuming $9.3$ milliseconds, PLOT attains the lowest regret.}

\begin{table}[h]
\caption{Comparison of Online Control Algorithms for Tracking the Unknown Target with Dynamics \eqref{eq:dynamic_target_example}}
\label{tab:benchmark}
\vskip 0.15in
\begin{center}
\begin{small}
\begin{sc}
\begin{tabular}{lcccr}
\toprule
Algorithm & Comp. Time [ms] & Memory & Regret \\
\midrule
PLOT (Ours)    & $10.7^*$ & $54$ & $ \mathbf{1661}$ \\
PLOT (Ours, w/ unweighted projection)    & $2.10$ & $54$ & $ {1685}$ \\
Riccatitron \cite{foster2020logarithmic} & $2.30$ & $90$ & ${2414}$\\
DAP \cite{agarwal2019online}    &$3.50$ & $45$ & $2454$ \\
FTL \cite{pmlr-v32-abbasi-yadkori14}   & ${0.40}$& $\mathbf{0}$ & $40004$         \\
SS-OGD \cite{karapetyan2023online} & $0.33$ & $4$ & $2608$\\
Naive LQR      & $\mathbf{0.25}$ & $\mathbf{0}$ & $2453$        \\
\bottomrule
\end{tabular}
\end{sc}
\end{small}\\
* - 9.3 ms spent on the weighted projection with a QP solver.
\end{center}
\vskip -0.1in
\end{table}

\subsubsection{Benchmark Implementation Details}
\label{app:hyperparameters}

\textbf{PLOT}:
The prediction horizon is set to $W=6$, the forgetting factor to $\gamma = 0.9$, and $M=10$. The the augmented matrix is defined as $\Tilde{S}_t \triangleq \matr{S_{t} &v_t}$ for $t = 1, \hdots, T-1$ as in Section \ref{sec:problem statement}, and for all $k = 1, \hdots, W$, the $k$ learners are initialized as follows: $\hat{S}_{j+k|j} = \mathbf{I}_6, v_{j+k|j} = \mathbf{0}_6$ and $P_{j|j-k} = 10^{-4}\times\mathbf{I}_7$, for all $j = -1, \hdots, k-2$. 

\textbf{Riccatitron}: Following the notation of \cite{simchowitz2020making}, the horizon of the learners is set to $h=5$ and the disturbance action policy length to $m=5$ decided upon by a search over the best parameter. The latter choice increases the memory requirement of Riccatitron in Table \ref{tab:benchmark} to achieve a good tracking performance. We disable the projection, but note that even with projection enabled the results are almost identical. The Online Newton Step (ONS) update is chosen, with The learning rate set to $\eta_{ons}=0.2$ and $\varepsilon_{ons} = 1$.

\textbf{DAP}: As per the notation of \cite{agarwal2019online}, and to match the memory allocation for PLOT, we fix the horizon length for the disturbance action policy to $H=5$. The projection is performed as detailed in \cite{agarwal2019online}, and the learning rate is set to $\frac{0.1}{\sqrt{T}}$.

\textbf{FTL}: The algorithm is implemented as detailed in \cite{pmlr-v32-abbasi-yadkori14} and contains no hyperparameters to be tuned.

\textbf{SS-OGD}: The affine control term is initialized with a vector of zeros, and the learning rate is set to $\alpha = 0.001*\mathbf{I}_4$, following the notation of \cite{karapetyan2023online}.

\textbf{Naive LQR}: The naive LQR controller performs state feedback on the observed error
\begin{equation*}
    u_t^{LQR}(x_t) \!=\! -K (x_t-r_t), \qquad \forall t = 0,\hdots T-1,
\end{equation*}
with $K$ defined in \eqref{eq:K*}.

\rev{
\subsubsection{Random Targets}\label{app_sec:random}

The dependence of the dynamic regret of PLOT on $V_T$ implies that for targets having a dynamic structure, including the one in the preceding example, PLOT's tracking performance is better compared to targets that lack dynamics. While the latter case is still covered by the structure imposed on the targets, as discussed in Section \ref{app_sec:AR}, the regret can scale linearly with $T$ in the worst case when $r_t$ at each time step is independent from the others. In this case, the indirect approach of PLOT may result in a poorer performance compared to direct approaches like DAP, Riccatitron, or SS-OGD when $\gamma$, $W$, and $M$ are not optimally tuned for the specific realization.

\begin{figure}
    \centering
    \includegraphics[scale=0.3]{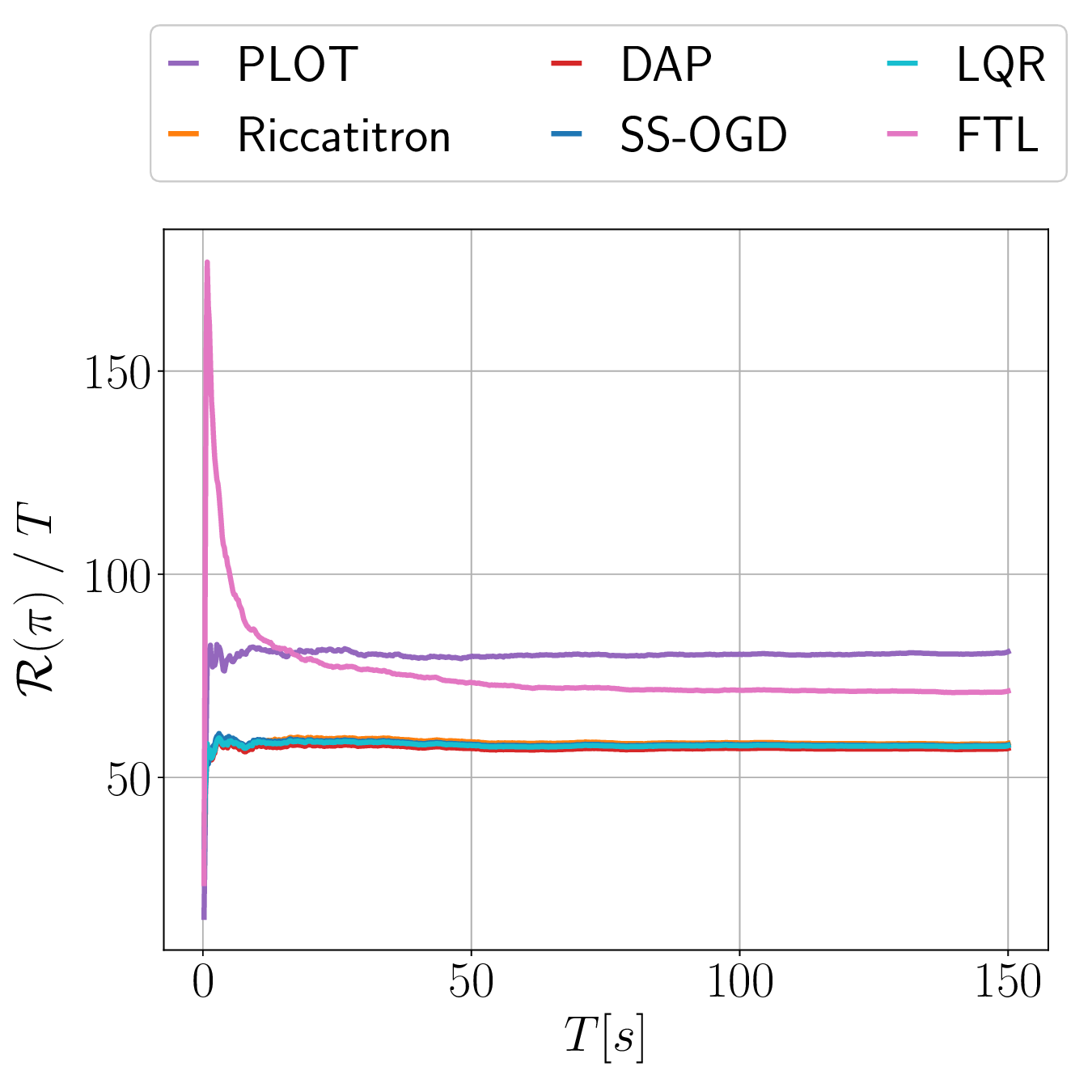}
    \caption{Average Dynamic Regret of Online Control algorithms applied to the online tracking problem of the unknown target following \eqref{eq:uniform_target}.}
    \label{fig:random_targets}
\end{figure}

To showcase this marginal example, we consider random targets sampled from the uniform distribution
\begin{equation}\label{eq:uniform_target}
    r_t \sim \mathcal{U}\left(
    \begin{bmatrix}
        -1\\
        -1
    \end{bmatrix}, 
    \begin{bmatrix}
        1\\
        1
    \end{bmatrix}
    \right).
\end{equation}

We compare PLOT with the same benchmark controllers whose tunable parameters are the same as in Section \ref{app:hyperparameters}. For PLOT we take $W=2$, $\gamma=0.02$, and $M=0.7$, leaving the other parameters unchanged from Section \ref{app:hyperparameters}.

We perform $10$ independent realizations of the random targets over a horizon of length $T=150$ seconds and compute the averaged regret over these runs for each of the six considered algorithms, shown in Figure \ref{fig:random_targets}. The random reference updates combined with the low forgetting factor $\gamma$ result in $\hat{S}_t$-s with large singular values, and, consequently, large inputs, while the optimal action in this case is staying close to the origin. While the performance can be improved significantly by lowering $M$ and increasing $\gamma$, bringing the averaged regret to around $25$, this example shows that for a completely random agent with no dynamics a direct approach may perform better than an indirect one like PLOT.

\subsection{A Naive-RLS Controller}\label{app_sec:naive_RLS}

The multi-predictor setup of PLOT is crucial for its stability and regret performance. To showcase this we consider also a simpler controller that we refer to as Naive-RLS in Algorithm \ref{alg:naive_rls}. In contrast to PLOT, it obtains $W$-step ahead dynamics predictions by recursively applying only the one step ahead estimate $\hat{S}_{t+1|t}$, resulting in $r_{t+1} = \hat{S}_{t+1|t}r_t,~r_{t+2|t}= \hat{S}^2_{t+1|t}r_t,~\hdots,~r_{t+W|t} =\hat{S}^W_{t+1|t}r_t$, for $p=1$. As detailed in Section \ref{sec:algorithm}, this can result in unstable predictions in particular when $S_t$ is time-varying. For $p>1$, we define $z_{t|t}= z_t$, $r_{t|t}=r_t$, and, for $k=1,\hdots,W-1$
\begin{equation*}
    z_{t+k|t}=\begin{cases}
        \begin{bmatrix}r^\top_{t+k|t}&\cdots&r_{t|t}&r_{t-1}\cdots&r^\top_{t-p+k+1}\end{bmatrix}^\top \; \text{if}\; p>k,\\
        \begin{bmatrix}r^\top_{t+k|t}&\cdots&r^\top_{t-p+k+1|t}\end{bmatrix}^\top \; \text{if}\; p\leq k.
        
    \end{cases}
\end{equation*}
\begin{algorithm}[hbt!]
\caption{Naive-RLS}\label{alg:naive_rls}
\begin{algorithmic}[1]
\REQUIRE{Horizon $W$, forgetting factor $\gamma \in (0, 1)$}
\STATE Compute $X, K, \{K_0, ..., K_{W-1}\}$ as in \eqref{eq:P*}, \eqref{eq:K*}, \eqref{eq:Kd}.
\STATE Initialize $1$ RLS learner according to the Algorithm \ref{alg:RLS}, for $k=1$.\label{code:MPCinit}
\FOR{t = 0, ..., $T-1$}
  \STATE Observe system state $x_t$, target state $r_t$. \label{code:MPCmeasurement}
  \FOR{$k = 1, ..., W$}
    \IF{$t+k \leq T$}
        \STATE Update the $1$-step-ahead predictor and obtain $\hat{S}_{t+1|t}$ according to the Algorithm~\ref{alg:RLS}.\\
        \STATE Predict $r_{t+k|t} =\hat{S}_{t+1|t} z_{t+k-1|t} $
    \ELSE
        \STATE Set $r_{t+k-1|t} = 0$.\label{code:MPCcorner}
    \ENDIF
  \ENDFOR
  \STATE Compute $u_t^{\pi}$ as in~\eqref{eq:mpc action}
 \ENDFOR
\end{algorithmic}
\end{algorithm}

To demonstrate the failure of Naive-RLS for agents with time-varying dynamics, we consider the static (circular) and dynamic references of Sections \ref{sec:circle_reference} and \ref{sec:spiral_reference}, respectively. Figure \ref{fig:naive_rls_circle} depicts the former case, showing 
the comparable performance of Naive-RLS to PLOT, since the static $S$ remains stable even when exponentiated. In the latter case with a dynamic $S_t$, shown in Figure \ref{fig:naive_rls_spiral}, Naive-RLS becomes unstable due to exploding matrices $\hat{S}_{t+k|t}$. The implementation details of both controllers are provided below.

\textbf{PLOT}:
The prediction horizon is set to $W=10$, the forgetting factor to $\gamma = 0.8$, and $M=1$. The the augmented matrix is defined as $\Tilde{S}_t \triangleq \matr{S_{t} &v_t}$ for $t = 1, \hdots, T-1$ as in Section \ref{sec:problem statement}, and for all $k = 1, \hdots, W$, the $k$ learners are initialized as follows: $\hat{S}_{j+k|j} = \mathbf{I}_6, v_{j+k|j} = \mathbf{0}_6$ and $P_{j|j-k} = 10^{-4}\times\mathbf{I}_7$, for all $j = -1, \hdots, k-2$. 

\textbf{Naive-RLS}:
The prediction horizon is set to $W=10$, the forgetting factor to $\gamma = 0.8$, and $M=1$. The the augmented matrix is defined as $\Tilde{S}_t \triangleq \matr{S_{t} &v_t}$ for $t = 1, \hdots, T-1$ as in Section \ref{sec:problem statement}, and for all $k = 1, \hdots, W$, the $k$ learners are initialized as follows: $\hat{S}_{j+k|j} = \mathbf{I}_6, v_{j+k|j} = \mathbf{0}_6$ and $P_{j|j-k} = 10^{-4}\times\mathbf{I}_7$, for all $j = -1, \hdots, k-2$. 

\begin{figure}[!ht]
\centering
\begin{minipage}[t]{.48\textwidth}
  \centering
  \includegraphics[width=0.9\linewidth]{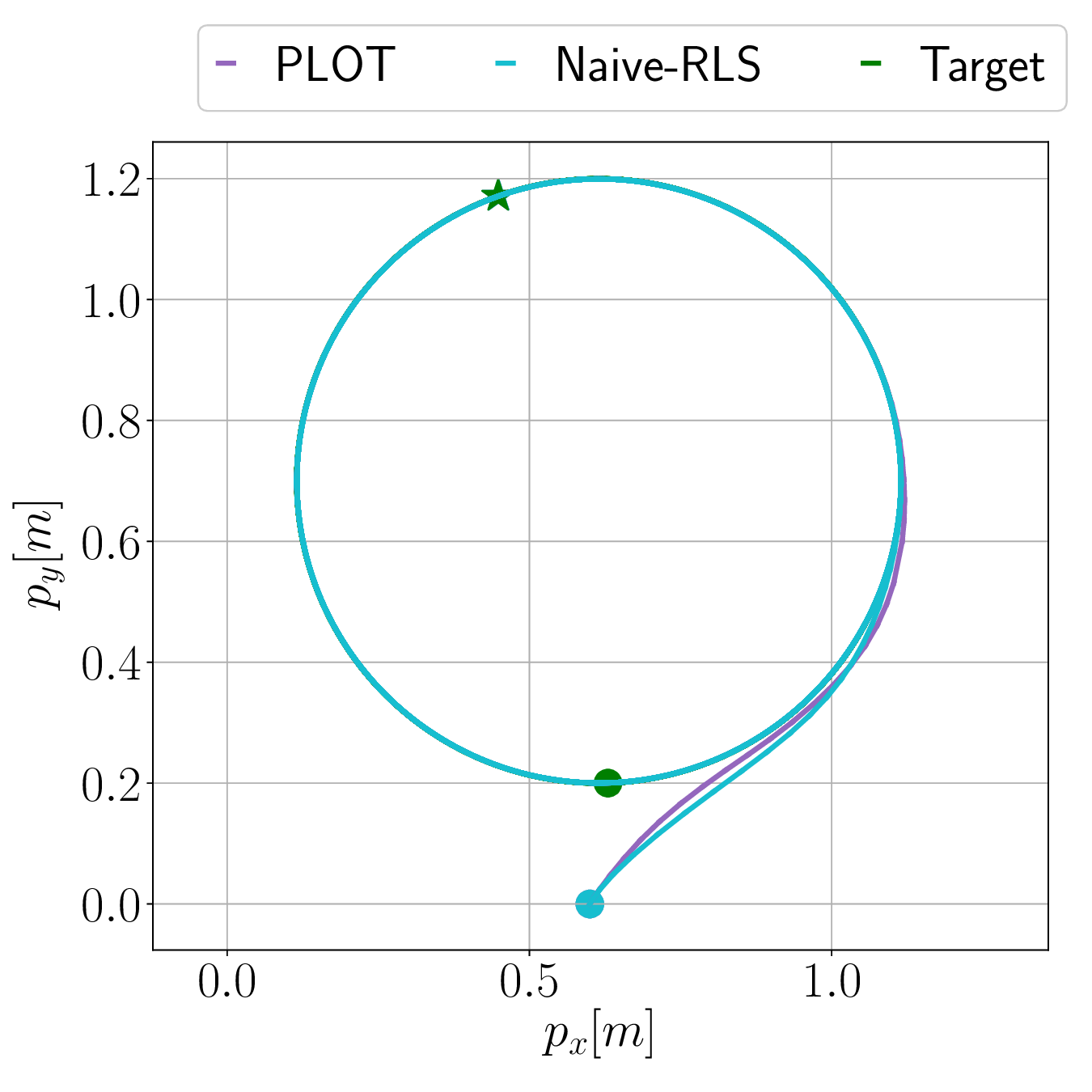}
  \caption{Naive-RLS shows comparable performance to PLOT on the static, circular target from Example \ref{exmp:circular}.}
  \label{fig:naive_rls_circle}
\end{minipage}
 \hfill
\begin{minipage}[t]{0.48\textwidth}
\centering
  \includegraphics[width=0.9\linewidth]{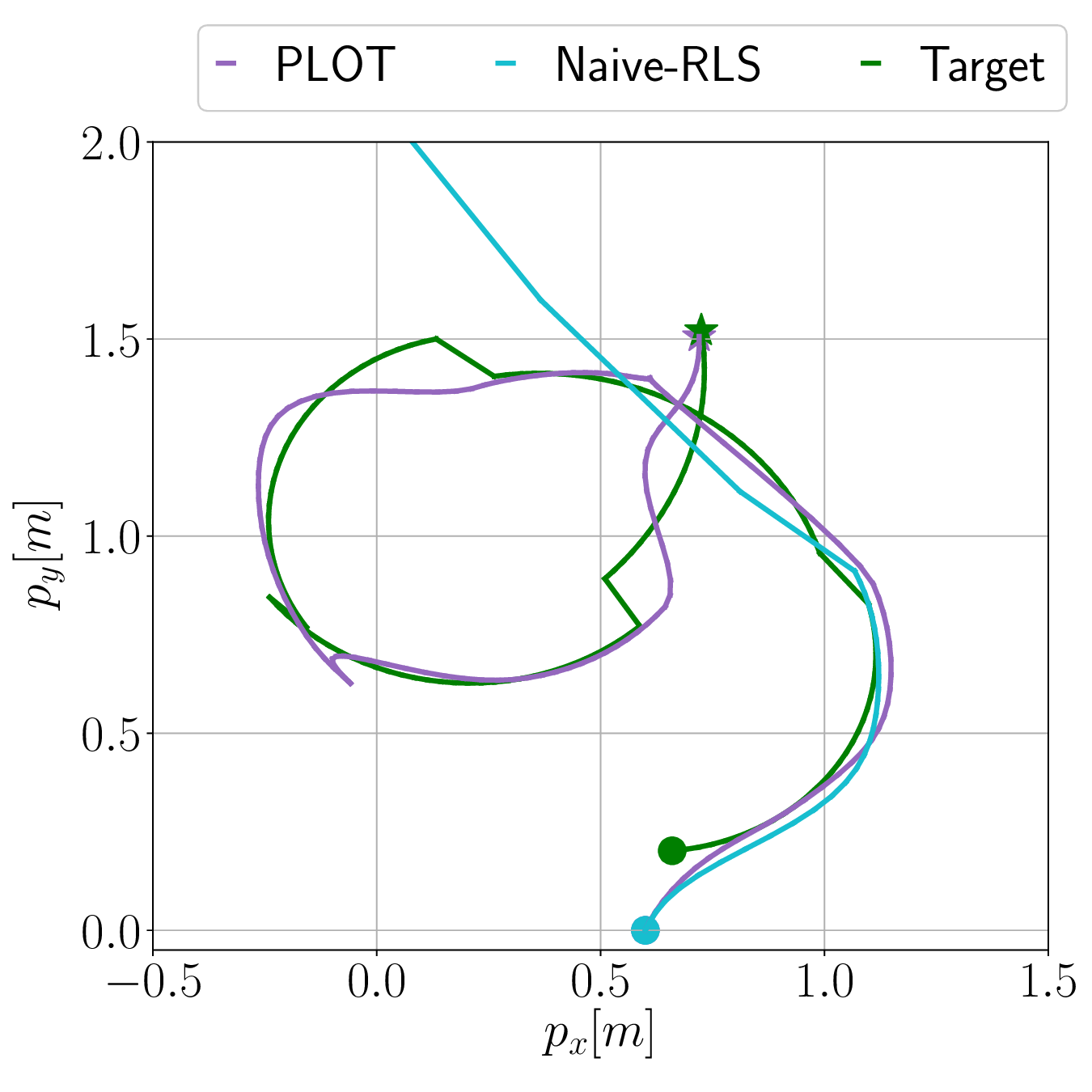}
  \caption{Naive-RLS becomes unstable while PLOT manages tracks the dynamic target of Section \ref{sec:spiral_reference}.}
  \label{fig:naive_rls_spiral}
\end{minipage}%
\end{figure}

}

\clearpage
\section{Implementation on Quadrotors}
\label{app:implementation}
The online tracking setting is of particular interest when applied to the control of quadrotors to track either a virtual target or another object of interest. We extend the simulation setup of Appendix \ref{app:simulations} to a practical one by applying PLOT to track a virtual sequentially revealed target using a Crazyflie 2.1 Quadrotor \cite{bitcraze} depicted in Figure \ref{fig:crazyflie}. The controllers use the linearised model of the quadrotor and the derived in Appendix \ref{app:quadrotor_model}. Below, we detail our experimental setup, as well as the results of the online tracking experiments.
\begin{figure}[!ht]
\centering
\begin{minipage}[t]{.25\textwidth}
   \centering
    \includegraphics[width=\linewidth]{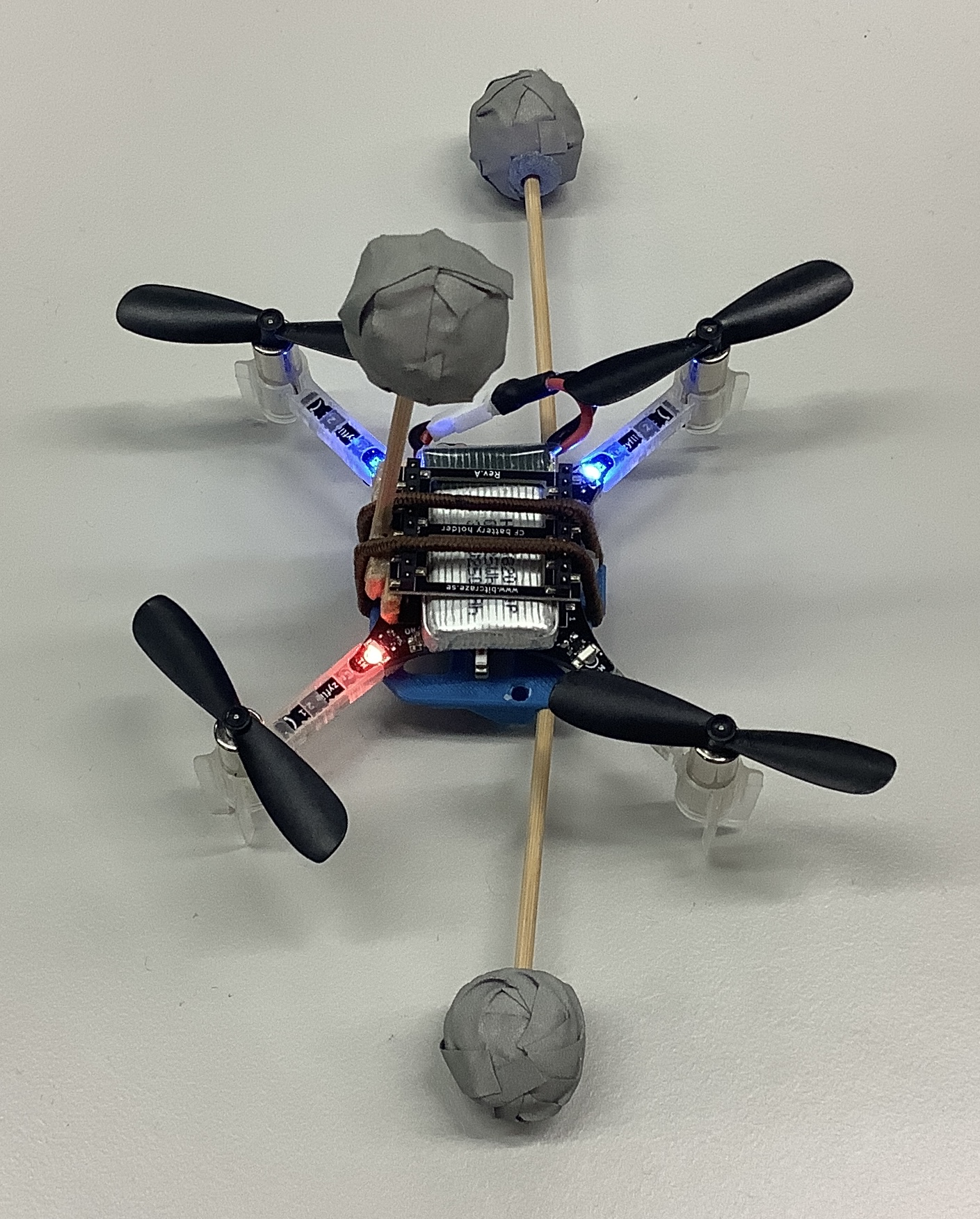}
    \caption{The Crazyflie 2.1 Quadrotor with three visual markers.}
    \label{fig:crazyflie}
\end{minipage}
 \hfill
\begin{minipage}[t]{0.65\textwidth}
\centering
  \includegraphics[width=0.8\linewidth]{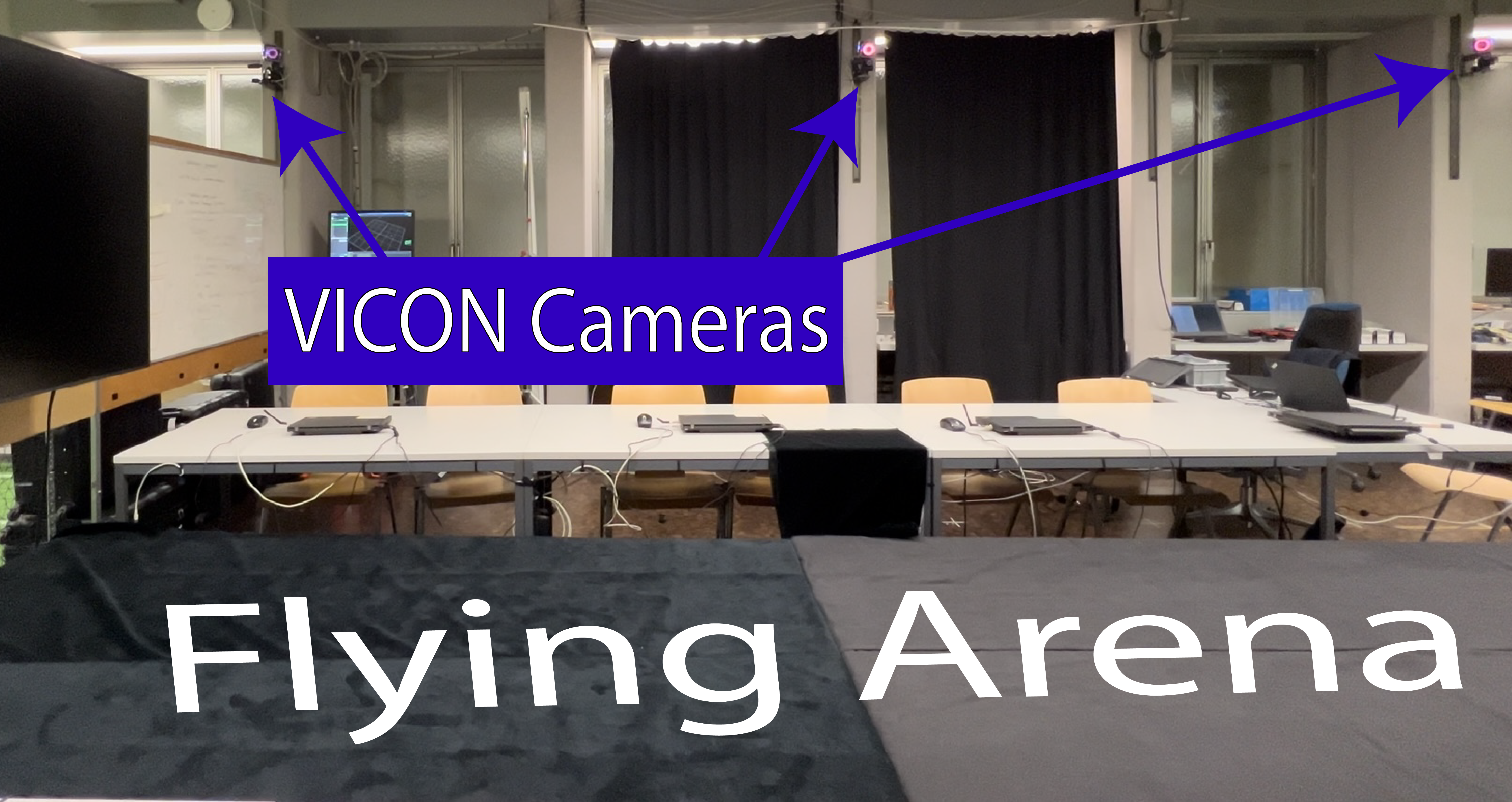}
  \caption{The D-FALL Laboratory of ETH Z\"urich, consisting of a dedicated Flying Arena and $5$ ($3$ shown) mounted VICON Cameras for state estimation.}
  \label{fig:vicon_setup}
\end{minipage}%
\end{figure}

\vspace{-0.5cm}
\subsection{Experimental Setup}

The Crazyflie quadrotors are equipped with a 3-axis accelerometer and gyroscope for onboard angular-rate control and 3 retro-reflective markers, as shown in Figure \ref{fig:crazyflie}. These markers are detected by $5$ VICON cameras mounted around the flying arena, shown in Figure \ref{fig:vicon_setup}, and the VICON system provides a spatial position estimate of the drone in its precalibrated reference frame of reference.

We set up the control architecture with the Robotic Operating System (ROS)~\cite{quigley2009ros}, which allows the controller to receive the state of the drone from the VICON mocap system in real-time. The control action is computed on a local computer and communicated with the drone by sending variable-sized packets through a $2.4$Ghz USB dongle \cite{Crazyradio}. The virtual reference trajectory is generated online after the drone has successfully taken off and is at a pre-defined random initial hovering position at a height of $0.4$ meters. The controller receives the  target 
state through the  ROS network as it flies at a predetermined rate of $10$ Hz.

\subsection{Online Tracking Experiments}

We run PLOT on the described setup to show its tracking performance. The cost matrices are taken to be the same as in Appendix \ref{app:quadrotor_model},  the prediction horizon is fixed to $W=5$ for the ``infinity"-shape and $W=3$ for the circle, the forgetting factor to $\gamma=0.8$ and no projection is performed. The initialization of PLOT is the same as in Appendix \ref{app:hyperparameters}. Two virtual reference target shapes, a circle and an ``infinity"-shape are generated online and published to the ROS network. The quadrotors start in the flying arena in Figure \ref{fig:vicon_setup} and follow the target as soon as state information is received. Figures \ref{fig:circle_crazyflie} and \ref{fig:figure8_crazyflie} below show the trajectory and state plots of PLOT and Naive LQR implemented for a horizon of $T=40$ seconds for the circle and ``infinity"-shape reference trajectories, respectively. Naive LQR is implemented as described in Appendix \ref{app:hyperparameters}.

As in the simulations, without any affine term, the naive LQR controller exhibits a delayed tracking behavior as expected. The constant offset can be noted clearly for both shapes. In comparison, the PLOT algorithm experiences a smaller tracking error. As opposed to the linear simulations in Appendix \ref{app:simulations}, the nonlinear hardware implementation has additional errors, especially for the circular example, due to imperfect mass value or linearization error, among other reasons. The VICON system is also known to introduce a drift through its localization algorithm when some of the cameras are malfunctioning. 

\begin{figure}
\begin{subfigure}{.48\textwidth}
\centering
  \includegraphics[width=0.9\linewidth]{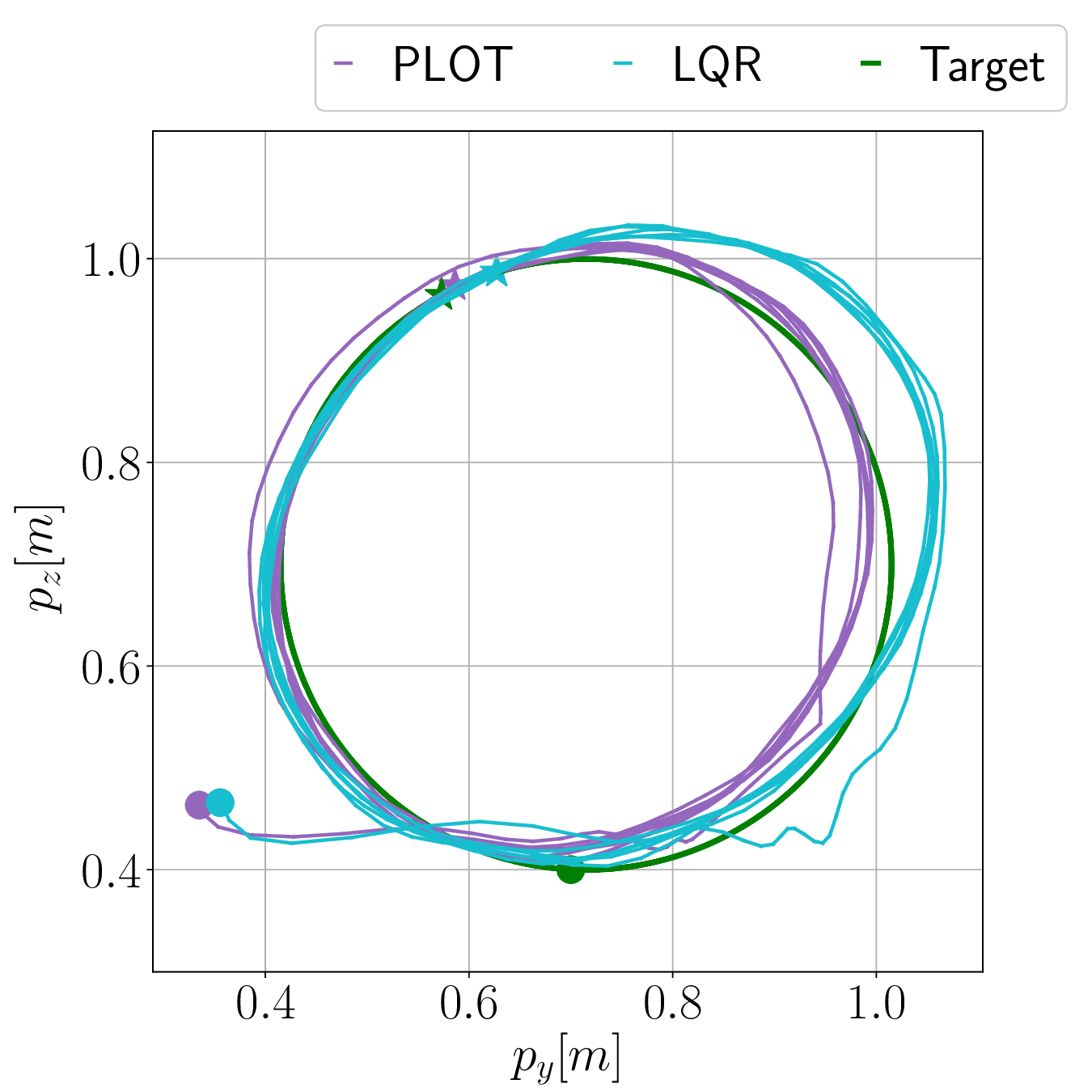}
  \caption{Trajectory plot.}
\end{subfigure}
\begin{subfigure}{.48\textwidth}
\centering
  \includegraphics[width=0.9\linewidth]{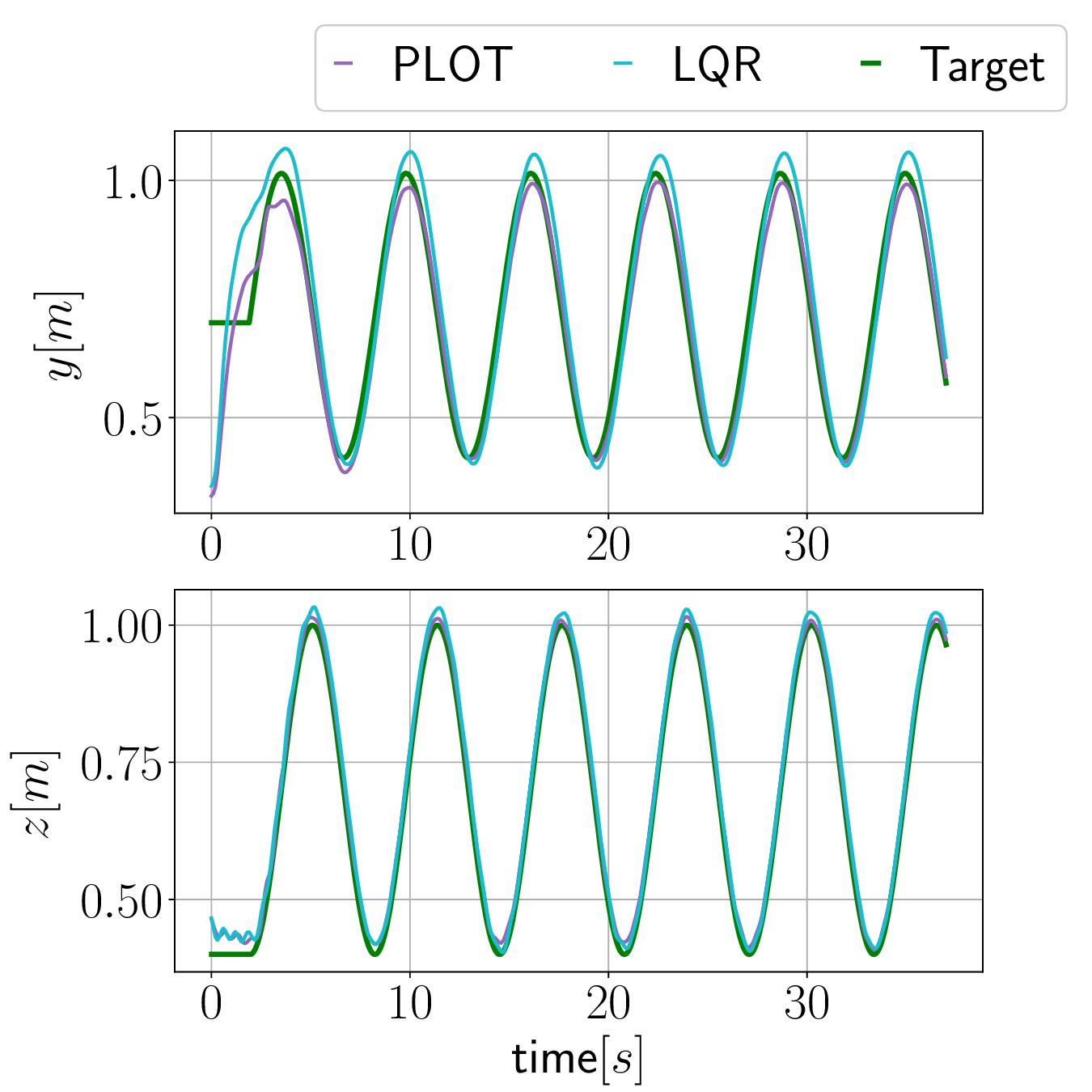}
  \caption{State plot.}
\end{subfigure}%
\caption{Circular reference tracked by Crazyflie Drones with the PLOT and Naive LQR controllers.}
\label{fig:circle_crazyflie}
\end{figure}

\begin{figure}
\begin{subfigure}{.48\textwidth}
\centering
  \includegraphics[width=0.9\linewidth]{figs/infinity_crazyflie.eps}
  \caption{Trajectory plot.}
\end{subfigure}
\begin{subfigure}{.48\textwidth}
\centering
  \includegraphics[width=0.9\linewidth]{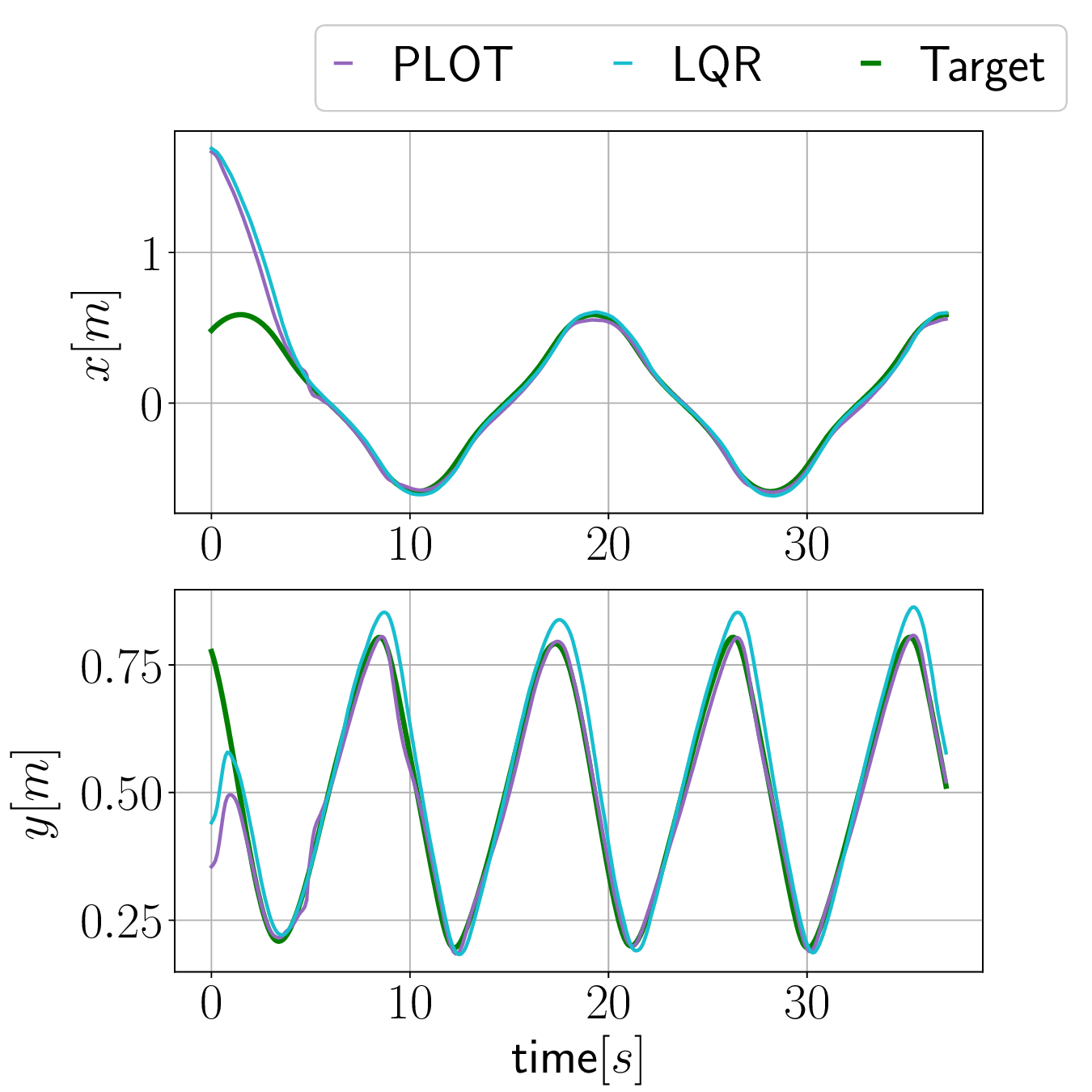}
  \caption{State plot.}
\end{subfigure}%
\caption{``Infinity"-shaped reference tracked by Crazyflie Drones with the PLOT and Naive LQR controllers.}
\label{fig:figure8_crazyflie}
\end{figure}

\end{document}